\documentclass[lettersize,journal]{IEEEtran}
\usepackage{amsmath,amsfonts,amssymb}
\usepackage{times}   
\usepackage{booktabs}
\usepackage{algorithm}
\usepackage{algpseudocode}
\usepackage{array}
\usepackage[caption=false,font=normalsize,labelfont=sf,textfont=sf]{subfig}
\usepackage{textcomp}
\usepackage{stfloats}
\usepackage{url}
\usepackage{verbatim}
\usepackage{graphicx}
\usepackage{geometry} 
\usepackage{subfig} 
\usepackage{cite}
\usepackage{enumitem}
\usepackage{amsthm}
\usepackage{mathrsfs}
\usepackage{xcolor}
\usepackage[colorlinks,linkcolor=blue,anchorcolor=blue,citecolor=blue]{hyperref}
\theoremstyle{plain}
\newtheorem{theorem}{Theorem} 
\newtheorem{lemma}{Lemma}
\newtheorem{corollary}{Corollary}

\theoremstyle{definition}
\newtheorem{definition}{Definition}
\newtheorem{assumption}{Assumption}

\theoremstyle{remark}
\newtheorem{remark}{Remark}

\begin{document}

\title{Optimally Bridging Semantics and Data: Generative Semantic Communication via Schrödinger Bridge}

\author{Dahua Gao, \IEEEmembership{Member, IEEE}, Ruichao Liu, \IEEEmembership{Graduate Student Member, IEEE}, Minxi Yang, \IEEEmembership{Member, IEEE},  Shuai Ma, \IEEEmembership{Member, IEEE}, Youlong Wu, \IEEEmembership{Member, IEEE}, and Guangming Shi, \IEEEmembership{Fellow, IEEE}.
\thanks{Corresponding author: Minxi Yang.}
\thanks{Dahua Gao and Minxi Yang are with the School of Artificial Intelligence, Xidian University, Xi'an, Shaanxi 710071 China and are also with Pazhou Lab, Huangpu, Guangdong 510555 China. (e-mail:  dhgao@xidian.edu.cn; yangminxi@xidian.edu.cn).}
\thanks{Ruichao Liu and Guangming Shi are with the School of Artificial Intelligence, Xidian University, Xi'an, Shaanxi 710071 China and are also with Peng Cheng Laboratory, Shenzhen, Guangdong 518055 China and also with the School of Artificial Intelligence, Xidian University, Xi'an, Shaanxi 710071 China (e-mail: liurch@stu.xidian.edu.cn; gmshi@xidian.edu.cn).}
\thanks{Shuai Ma is with Peng Cheng Laboratory, Shenzhen, Guangdong 518055 China (e-mail: mash01@pcl.ac.cn).}
\thanks{Youlong Wu is with the School of Information Science and Technology, ShanghaiTech University, Shanghai 201210, China (e-mail: wuyl1@shanghaitech.edu.cn)}
}

\markboth{Journal of \LaTeX\ Class Files,~Vol.~14, No.~8, August~2021}%
{Shell \MakeLowercase{\textit{et al.}}: A Sample Article Using IEEEtran.cls for IEEE Journals}

\maketitle

\begin{abstract}
Generative Semantic Communication (GSC) is a promising solution for image transmission over narrow-band and high-noise channels. However, existing GSC methods rely on long, indirect transport trajectories from a Gaussian to an image distribution guided by semantics, causing severe hallucination and high computational cost. To address this, we propose a general framework named Schrödinger Bridge-based GSC (SBGSC). By leveraging the Schrödinger Bridge (SB) to construct optimal transport trajectories between arbitrary distributions, SBGSC breaks Gaussian limitations and enables direct generative decoding from semantics to images. Within this framework, we design Diffusion SB-based GSC (DSBGSC). DSBGSC reconstructs the nonlinear drift term of diffusion models using Schrödinger potentials, achieving direct optimal distribution transport to reduce hallucinations and computational overhead. To further accelerate generation, we propose a self-consistency-based objective guiding the model to learn a nonlinear velocity field pointing directly toward the image, bypassing Markovian noise prediction to significantly reduce sampling steps. Simulation results demonstrate that DSBGSC outperforms state-of-the-art GSC methods, improving FID by at least 38\% and SSIM by 49.3\%, while accelerating inference speed by over 8 times.

\end{abstract}


\section{Introduction}
\label{sec:introduction}

\begin{figure}[!t]
\centering
\includegraphics[width=\linewidth]{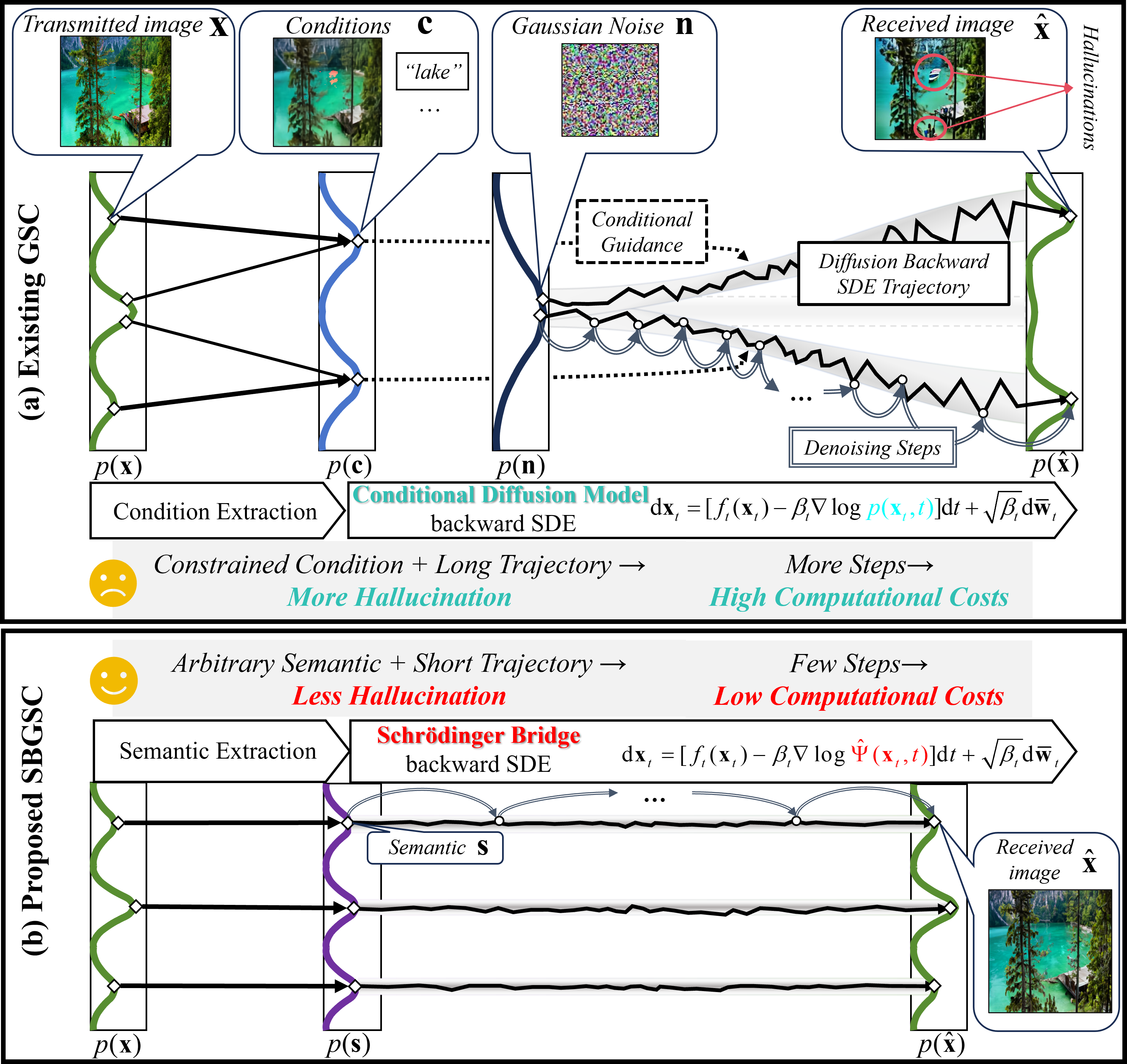}
\caption{Comparison between existing conditional diffusion model-based generative semantic communication and the proposed Schrödinger Bridge-based generative semantic communication (SBGSC). (a) Existing methods follow a \textit{conditional generation paradigm}, initiating from independent Gaussian noise and indirectly recovering data via condition guidance. This approach suffers from semantic condition mismatch, high computational costs, and hallucinations. (b) The proposed SBGSC adopts a \textit{direct transport paradigm}, establishing an optimal transport trajectory directly from the received noisy semantics to the target data distribution via Schrödinger Bridge, thereby achieving efficient and high-fidelity semantic perceptual quality.}
\label{fig:optimal transport vs indirect transport}
\end{figure}

\IEEEPARstart{S}{emantic} communication (SC) enables robust image transmission in challenging environments, such as deep-sea exploration and disaster recovery, where traditional separation-based systems often fail due to narrow-band and high-noise constraints \cite{bourtsoulatze_deep_2019}. Unlike conventional paradigms that pursue bit-level precision, SC prioritizes semantic fidelity by extracting and transmitting task-relevant features \cite{weaver_recent_1953, shi_semantic_2021}. While current Deep Joint Source-Channel Coding (DeepJSCC) schemes \cite{choi_neural_2019, yang_swinjscc_2025} outperform traditional methods at low signal-to-noise ratios (SNR), they typically rely on point-to-point losses like Mean Squared Error, which lead to blurry reconstructions and poor perceptual quality at extremely low bitrates \cite{blau_perception-distortion_2017}. Recent theoretical advancements suggest that semantic distortion metrics based on distribution distances are more effective than signal-based metrics for preserving image integrity in such resource-constrained scenarios \cite{chai_rate-distortion_2025}.

To enhance semantic perceptual quality, the second category introduces Generative Artificial Intelligence (GAI), termed Generative Semantic Communication (GSC) \cite{grassucci_enhancing_2024,ren_generative_2025}. GSC utilizes a pre-trained generative model at the receiver as a prior, achieving generative decoding via conditional distribution sampling to realize distribution reconstruction, rather than point reconstruction. The mainstream GSC approaches are based on Conditional Diffusion Models (CDM), which can generate results with realistic textures and adherence to natural image distributions using minimal semantic information \cite{SGM,zhan_conditional_2025}. Nevertheless, the bottleneck of existing GSC methods lies in their generation process, which is inherently an \textit{indirect transport}. As illustrated in Fig. \ref{fig:optimal transport vs indirect transport}(a), existing methods typically first extract semantics from the transmitted image $\mathbf{x}$ into a specific conditional modality $\mathbf{c}$ (such as text \cite{nam_sequential_2023}, edge maps \cite{zhang_semantics-guided_2025}, or coarse images \cite{gao_bridging_2025}). This condition $\mathbf{c}$ is then utilized to provide conditional guidance for the diffusion backward Stochastic Differential Equation (SDE) trajectory, driving the sampling process from an independent Gaussian noise distribution $p(\mathbf{n})$ to the received data distribution $p(\hat{\mathbf{x}})$. As depicted by the long, jagged trajectory in Fig. \ref{fig:optimal transport vs indirect transport}(a), this indirect paradigm poses three major challenges for GSC schemes in practical deployment:
\begin{itemize}
  \item The specific modalities utilized as conditions $\mathbf{c}$ are limited in type and incomplete in information \cite{zhan_conditional_2025}. They fail to adapt precisely to diverse downstream tasks, lacking universality and flexibility.
  \item Due to the long generation trajectory originating from a Gaussian prior, existing models often hallucinate non-existent or erroneous details (marked in red circles in Fig. \ref{fig:optimal transport vs indirect transport}(a)) when relying on noise-corrupted semantic conditions \cite{aithal_understanding_2024}.
  \item The indirect generation mechanism dictates a lengthy backward SDE trajectory. This requires a vast number of iterative denoising steps (illustrated by the multiple backward step arrows), leading to severe computational overhead \cite{DDPM}.
\end{itemize}

To fundamentally address the inefficiencies of the indirect Gaussian-based paradigm, this paper proposes to bypass the intermediate Gaussian prior and instead solve for the optimal transport trajectory directly connecting the semantic distribution to the data distribution. Mathematically, the problem of finding the optimal path between two arbitrary probability distributions corresponds to the Schrödinger Bridge (SB) problem \cite{leonard_survey_2014}. Based on this theory, we establish a general framework named SB-based Generative Semantic Communication (SBGSC). As illustrated in Fig. \ref{fig:optimal transport vs indirect transport}(b), unlike existing methods that start from noise, SBGSC leverages the SB to construct a direct backward SDE trajectory from the extracted semantic distribution $p(\mathbf{s})$ to the received image distribution $p(\hat{\mathbf{x}})$. Guided by the gradient of the Schrödinger potential $\hat{\Psi}$ rather than simple score matching, SBGSC offers distinct advantages tailored for narrow-band and high-noise channels:
\begin{itemize}
\item SB enables a mapping directly between arbitrary learned semantic distributions and data distributions. This allows the system to utilize abstract semantic features $\mathbf{S}$ directly, without forcibly introducing rigid intermediate modalities like text or segmentation maps, thereby eliminating hallucination risks caused by modality bias.
\item As depicted by the \textit{Short Trajectory} in Fig. \ref{fig:optimal transport vs indirect transport}(b), SB uses semantic features as the direct starting state for generation. This optimal transport path inherently constrains the generation trajectory to closely follow the real data manifold and maintain structural consistency, effectively preventing the hallucinations caused by excessive randomization.
\item By constructing a direct transport path, SBGSC avoids the tortuous process of “encoding to condition, then sampling from Gaussian.” Consequently, the generation requires significantly \textit{Few Steps} (represented by the reduced number of transport arrows in Fig. \ref{fig:optimal transport vs indirect transport}(b)), which significantly lowers computational costs and inference latency.
\end{itemize}

The main contributions of this paper are summarized as follows:
\begin{itemize}
  \item We propose SBGSC, a general generative semantic communication framework grounded in SB theory. Unlike existing GSC paradigms that rely on an indirect Gaussian prior, SBGSC directly bridges the arbitrary semantic distribution and the target data distribution. This framework jointly optimizes semantic representation extraction under harsh channel conditions and realizes direct, optimal generative decoding.
  \item We provide a rigorous theoretical analysis from an information-theoretic perspective. By evaluating mutual information bounds, we mathematically demonstrate that directly modeling the optimal transport between semantics and data minimizes information loss and structural mismatch, proving the theoretical superiority of the SB-based scheme over traditional conditional diffusion models.
  \item We design a specific implementation named Diffusion SB-based Generative Semantic Communication (DSBGSC) for image transmission under severe narrow-band and high-noise scenarios. By setting the received semantic features as the direct starting state, DSBGSC reconstructs the nonlinear drift term of traditional diffusion models using Schrödinger potentials. This achieves a direct and optimal transport to the data distribution, effectively mitigating hallucinations. To further accelerate the decoding process, we introduce a self-consistency-based optimization objective that guides the model to learn a nonlinear velocity field pointing directly toward the target image. By bypassing the traditional stepwise Markovian noise prediction, this mechanism significantly reduces the required sampling steps and computational costs.
  Simulation results demonstrate that DSBGSC significantly outperforms state-of-the-art DeepJSCC and conditional GSC baselines. Under highly constrained channel conditions, our method achieves superior semantic perceptual quality, effectively suppresses generative hallucinations, and significantly reduces the number of sampling steps, thereby lowering computational costs and inference latency.
\end{itemize}

The remainder of this paper is organized as follows: Section \ref{sec:related_works} reviews related work. Section \ref{sec:SBGSC_framework} proposes the SBGSC framework and theoretical analysis. Section \ref{sec:DSBGSC_method} introduces the specific DSBGSC algorithm implementation. Section \ref{sec:simulation} presents the simulation results. Finally, Section \ref{sec:conclusion} concludes the paper.

\section{Related Works}
\label{sec:related_works}
This section first briefly summarizes existing generative semantic communication methods for image. Subsequently, it introduces the core mathematical tool utilized in this paper — Schrödinger Bridge. Finally, it reviews current research on SB-based generative models in the field of AI, laying the foundation for the proposed SB-based GSC framework and the DSBGSC method.

\subsection{Generative Semantic Communication for Image}
In recent years, extensive work has incorporated deep generative models into semantic communication systems, establishing the field of GSC \cite{ren_generative_2025}. The fundamental premise is that the transmitter only needs to transmit compact semantic representations or conditional information, while the receiver leverages a powerful generative prior to sample and reconstruct images from noise or simple priors \cite{zhang_semantic_2025}. This approach aims to significantly reduce the required bit rate while maintaining semantic consistency. Consequently, GSC is particularly suitable for extreme scenarios characterized by narrow bandwidth and high noise, where image recovery relies on strong priors and minimal transmitted bits or low-dimensional semantic features.

Early works predominantly employed Conditional GANs, utilizing low-dimensional latent variables or task labels from the encoder as conditions to generate natural images at the receiver \cite{lokumarambage_wireless_2023}. While these methods alleviated the blurriness issue associated with DeepJSCC to some extent and produced visually more natural images, they generally suffered from training instability, mode collapse, and limited semantic controllability \cite{ding_take_2022}. Following the success of diffusion models in high-fidelity image generation \cite{DDPM}, CDMs have been introduced into GSC systems \cite{zhang_semantic_2025}. In typical implementations, the transmitter compresses an image into one of various modalities—such as text descriptions \cite{nam_sequential_2023}, bounding boxes \cite{yang_agent-driven_2025}, scene graphs \cite{yang_sg2sc_2024}, semantic segmentation maps \cite{grassucci_lightweight_2025}, edge maps \cite{zhang_semantics-guided_2025}, or blurred images \cite{gao_bridging_2025}, or a combination thereof \cite{park_transmit_2025}. The receiver then employs a conditional diffusion model to iteratively sample and reconstruct the image from a Gaussian prior. These methods have demonstrated superior performance over DeepJSCC in specific application scenarios such as scene understanding \cite{yang_sg2sc_2024} and vehicle surveillance \cite{yang_agent-driven_2025}.

Despite the advantages of conditional generative model-based GSC in terms of perceptual quality, these methods universally follow an indirect transport trajectory: “Conditional Distribution $\to$ Gaussian Prior Distribution $\to$ Data Distribution.” On one hand, the conditional modalities are limited in variety and cannot fully characterize semantics. On the other hand, the introduction of an extra prior distribution fails to directly address the design of the “Semantic Distribution $\to$ Data Distribution” mapping. This results in significant hallucinations and high computational overhead. To fundamentally improve the performance and usability of GSC, a theoretical framework and tool capable of directly characterizing and solving for the “optimal transport from semantic distribution to data distribution” is required.

\subsection{Schrödinger Bridge}
Mathematically, the problem of finding the optimal transport trajectory between two arbitrary probability distributions under given constraints is known as the SB problem \cite{leonard_survey_2014}. Originally derived from statistical physics and stochastic process theory, the classical formulation of the SB problem is as follows: given two marginal distributions and a reference stochastic process (such as Brownian motion), find a stochastic process that is closest to the reference process in terms of relative entropy while satisfying the specified marginal constraints at the start and target times.

Given two boundary distributions $p_{\text{data}}$ and $p_{\text{prior}}$ (which can correspond to the noisy semantic distribution and the data distribution in semantic communication, respectively), the SB problem seeks a stochastic process $P$ defined on $t \in [0, 1]$ such that its marginals satisfy $P_{t=0} = p_{\text{data}}$ and $P_{t=1} = p_{\text{prior}}$, while minimizing the Kullback–Leibler (KL) divergence with respect to a reference process $W$ (typically chosen as Brownian motion):
\begin{equation}
\min_{P} \text{KL}(P || W), \quad \text{s.t.} \ P_{t=0} = p_{\text{data}}, P_{t=1} = p_{\text{prior}}.
\end{equation}

This problem can be equivalently viewed as a dynamic version of Entropy-Regularized Optimal Transport \cite{leonard_survey_2014}. To facilitate the solution of forward and backward SDEs, two coupled Schrödinger potentials can be used from a dynamic perspective to describe them, which satisfy the coupled partial differential equations under the Hopf-Cole transformation \cite{SBPotentials1,SBPotentials2}:
\begin{equation}
\begin{gathered}
  \left\{ \begin{array}{l}
   \displaystyle \frac{\partial \Psi (\mathbf{x},t)}{\partial t}
   = -\nabla_{\mathbf{x}} \Psi^\top \mathbf{f}_t
   - \frac{1}{2}\operatorname{Tr}\bigl( \beta_t \nabla_{\mathbf{x}}^{2} \Psi \bigr), \\
   \displaystyle \frac{\partial \hat{\Psi}(\mathbf{x},t)}{\partial t}
   = -\nabla_{\mathbf{x}} \cdot \bigl( \hat{\Psi} \mathbf{f}_t \bigr)
   + \frac{1}{2}\operatorname{Tr}\bigl( \beta_t \nabla_{\mathbf{x}}^{2} \hat{\Psi} \bigr),
   \end{array} \right. \\
   \text{s.t.}\quad
   \Psi (\mathbf{x},0)\hat{\Psi}(\mathbf{x},0) = p_{\text{prior}}(\mathbf{x}),
   \Psi (\mathbf{x},1)\hat{\Psi}(\mathbf{x},1) = p_{\text{data}}(\mathbf{x}),
\end{gathered}
\end{equation}
where $\Psi$ and $\hat{\Psi}$ are the coupled backward and forward Schrödinger potentials, and $\mathbf{f}_t$ and $\beta_t$ are the drift term and diffusion coefficient of the SDE.
The corresponding forward and backward SDEs are:
\begin{subequations}
\begin{equation}
\text{d}{{\mathbf{x}}_{t}}=[{{\mathbf{f}}_{t}}+{{\beta }_{t}}{{\nabla }_{\mathbf{x}}}\log \Psi ({{\mathbf{x}}_{t}},t)]\text{d}t+\sqrt{{{\beta }_{t}}}\text{d}{{\mathbf{w}}_{t}},
\label{eq:sb_sde_forward}
\end{equation}
\begin{equation}
\text{d}{{\mathbf{x}}_{t}}=[{{\mathbf{f}}_{t}}-{{\beta }_{t}}{{\nabla }_{\mathbf{x}}}\log \hat{\Psi }({{\mathbf{x}}_{t}},t)]\text{d}t+\sqrt{{{\beta }_{t}}}\text{d}{{\mathbf{\bar{w}}}_{t}},
\label{eq:sb_sde_backward}
\end{equation}
\end{subequations}
where $\mathbf{x}_{t=0}$ is a sample from the prior distribution $p_\text{prior}(\mathbf{x})$, and $\mathbf{x}_{t=1}$ is a sample from the data distribution $p_\text{data}(\mathbf{x})$, $\nabla \log \Psi(\mathbf{x}_t, t)$ and $\nabla \log \widehat{\Psi}(\mathbf{x}_t, t)$ are the forward and backward optimal drift terms provided by the SB theory. Unlike diffusion models which force the initial distribution to be a standard Gaussian, SB allows the initial distribution to be arbitrary, which theoretically supports us to directly set the initial distribution as a noisy semantic distribution, thereby greatly shortening the generation trajectory.

Since obtaining analytical SB solutions for arbitrary distributions is intractable, numerical methods have been widely explored. Traditional approaches, such as those based on Iterative Proportional Fitting \cite{leonard_survey_2014} or stochastic control formulations solving Fokker-Planck equations \cite{chen_relation_2016}, provide theoretical insights but often suffer from prohibitive computational costs or the difficulty of solving partial differential equations in high-dimensional spaces. In contrast, recent research has successfully combined SB with the score-based diffusion framework to propose the Diffusion Schrödinger Bridge (DSB) \cite{de_bortoli_diffusion_2021}. By learning score functions to approximate the bridge process within a unified SDE framework, DSB effectively addresses the scalability issues of traditional methods, enabling efficient SB solutions for high-dimensional complex data.

\begin{figure*}[!t]
\centering
\includegraphics[width=0.9\textwidth]{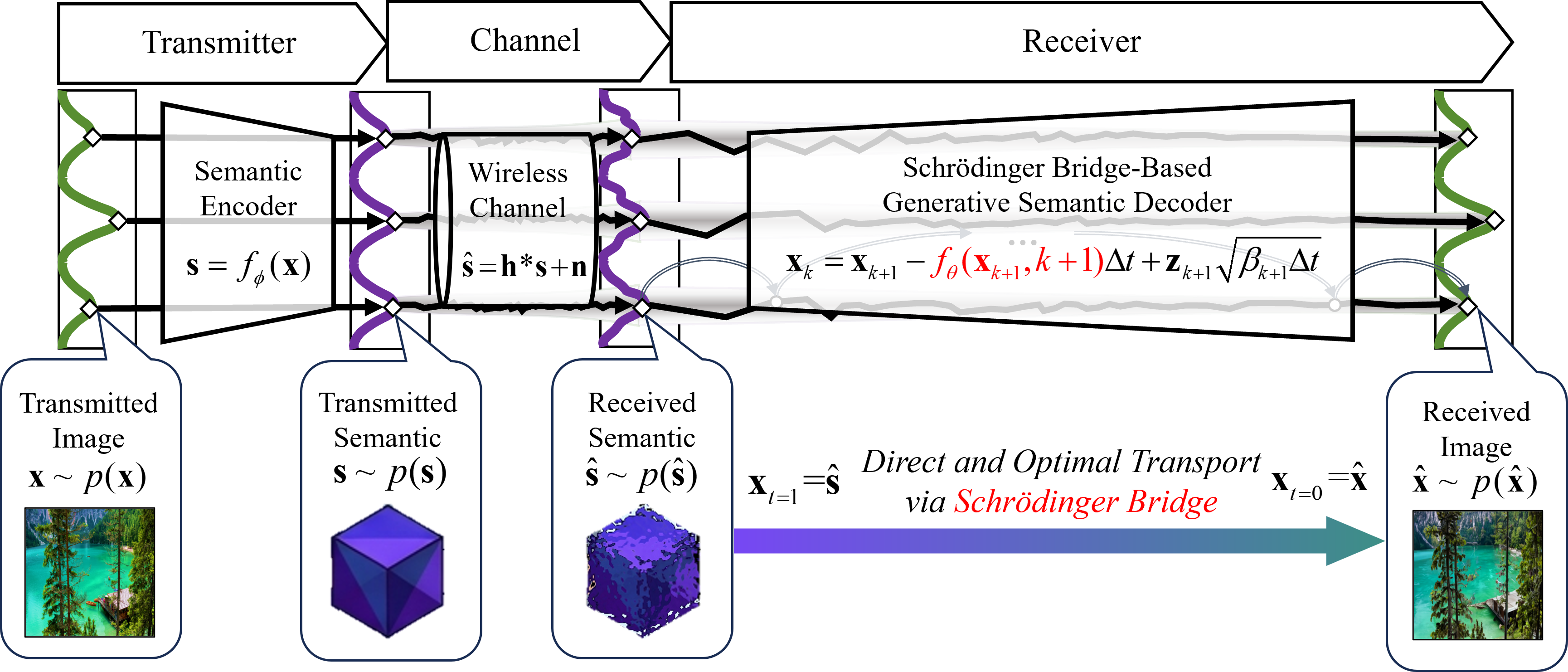}
\caption{Schematic diagram of the proposed SBGSC framework. The framework mainly consists of a joint source-channel semantic encoder, a wireless physical channel, and a core SB-based generative decoder. The transmitter maps the input data $\mathbf{x}$ to semantic symbols $\mathbf{s}$ via the encoder $f_\phi$; after passing through a fading and noisy channel, the receiver obtains the corrupted semantics $\hat{\mathbf{s}}$. The decoder in this framework utilizes SB theory to construct a backward SDE, establishing a direct and optimal transport trajectory from the received semantics $\mathbf{x}_{t=0}=\hat{\mathbf{s}}$ (as the initial state) to the target data $\mathbf{x}_{t=1}=\hat{\mathbf{x}}$ (as the target state).}
\label{fig:SBGSC framework}
\end{figure*}

\subsection{Schrödinger Bridge-based Generative Models}
With the maturation of numerical solutions for SB, a series of works applying SB to high-dimensional generative modeling have emerged in the AI field \cite{de_bortoli_diffusion_2021}. The general idea is to treat the generation process as constructing an SB between a given source distribution $p_{0}$ and a target distribution $p_{1}$. By learning appropriate control strategies or score functions, the evolution of the bridge process is simulated in both forward and backward time directions to achieve high-quality sampling from $p_{0}$ to $p_{1}$. Compared to traditional “Noise $\to$ Data” unidirectional generators, this “Distribution $\to$ Distribution” bridging perspective is more general.

Representative methods include the DSB model \cite{de_bortoli_diffusion_2021,chen_likelihood_2022}, which introduces SB constraints on top of the forward and backward SDE frameworks of diffusion models to match both start and target marginal distributions, enabling distribution-to-distribution generation on complex data such as images and videos. Subsequent works like GOUB \cite{yue_image_2024} and DDBM \cite{zhou_denoising_2024} have proposed using Doob's h-transform to construct diffusion bridge models, alleviating the KL divergence convergence difficulties found in classical DSB models. More recently, works such as UniDB \cite{zhu_unidb_2025} have attempted to construct a unified diffusion bridge framework capable of characterizing and approximating various bridge processes within a single parameterized model, enabling unified modeling from arbitrary prior distributions to arbitrary target distributions.

Since these models can choose any source distribution closer to the task as the generation starting point, rather than being limited to Gaussian noise, SB-based generative models offer distinct advantages: the generation path is determined by the optimality principle of SB, potentially reducing invalid detour sampling steps and lowering computational costs. In the setting where the source is the semantic distribution and the target is the data distribution, this approach naturally aligns with the direct optimal transport of Semantics to Data, facilitating semantic consistency and mitigating hallucinations.

However, existing SB-based generative works mostly focus on general generation, style transfer, or domain adaptation \cite{zhu_unidb_2025}, without systematic design and analysis for semantic communication scenarios. There is a lack of systematic research on how to integrate SB into communication system frameworks to establish a unified distribution-to-distribution perspective across semantic coding, channel transmission, and generative decoding. To this end, this paper formalizes the generative semantic communication problem as constructing an optimal transport between semantic and data distributions, introduces the SB to build a novel framework named SBGSC, and subsequently proposes the specific DSBGSC implementation for narrow-band, high-noise image transmission scenarios.

\section{Proposed SBGSC Framework}
\label{sec:SBGSC_framework}
This section first introduces the overall structure of the SBGSC framework. Subsequently, it details the core component of the framework—the SB-based generative decoder. Finally, it provides an in-depth theoretical analysis of the superiority of SBGSC over mainstream CDM-based GSC methods in terms of generation quality and computational efficiency.

\subsection{SBGSC Framework Overview}
As shown in Fig. \ref{fig:SBGSC framework}, the proposed SBGSC system is an end-to-end deep learning communication framework designed to achieve high-fidelity semantic transmission and perceptual quality of high-dimensional source data in adverse channel environments. The entire communication system can be modeled as a Markov chain process: $\mathbf{x} \rightarrow \mathbf{s} \rightarrow \hat{\mathbf{s}} \rightarrow \hat{\mathbf{x}}$, consisting primarily of three parts: a joint source-channel semantic encoder, a wireless physical channel, and a SB-based generative decoder.

At the transmitter, the system deploys a joint source-channel semantic encoder parameterized by $\phi$. Faced with high-dimensional input source data $\mathbf{x}$ (e.g., images), the encoder does not perform the traditional separation of source coding and channel coding. Instead, it directly extracts highly robust semantic features via a deep neural network and maps them into a complex symbol sequence $\mathbf{s}$ suitable for transmission over wireless channels. The encoding process is denoted as 
\begin{equation}
    \mathbf{s}=f_\phi(\mathbf{x}).
    \label{eq:sbgsc_encoder}
\end{equation}
To meet communication bandwidth constraints, the channel bandwidth ratio (CBR) \cite{yang_swinjscc_2025}, defined as $k/n$ (where $k$ is the dimension of $\mathbf{s}$ and $n$ is the dimension of $\mathbf{x}$), should not exceed a preset threshold. This process, based on joint source-channel coding (JSCC), not only accomplishes data dimensionality reduction and compression but also implicitly incorporates error correction coding mechanisms against channel noise, realizing a compact and noise-resilient semantic representation.

Subsequently, the encoded semantics $\mathbf{s}$ are transmitted through a wireless channel. To simulate a wireless communication environment, this paper considers Additive White Gaussian Noise (AWGN) and fading effects. The received signal $\hat{\mathbf{s}}$ at the receiver can be expressed as 
\begin{equation}
    \hat{\mathbf{s}} = h*\mathbf{s} + n,
    \label{sbgsc_channel}
\end{equation}
where $h$ represents the channel fading coefficient and $n$ denotes additive noise following a complex Gaussian distribution. During this process, due to the deep compression necessitated by bandwidth limits and the inevitable interference from channel noise, the received semantic distribution $p(\hat{\mathbf{s}})$ undergoes harmful distortion compared to the transmitted semantic distribution $p(\mathbf{s})$ (illustrated by the distorted distribution curve in Fig. \ref{fig:SBGSC framework}).

At the receiver, SBGSC employs a SB-based generative decoder to reconstruct the original data $\hat{\mathbf{x}}$ from the corrupted received semantics $\hat{\mathbf{s}}$. The decoder is parameterized by $\theta$ and denoted as 
\begin{equation}
    \hat{\mathbf{x}}=f_\theta(\hat{\mathbf{s}}).
    \label{eq:sbgsc_decoder}
\end{equation}
This decoding mechanism combines semantic decompression and noise restoration functionalities and is the core component enabling high-performance semantic communication in the SBGSC framework, which will be elaborated in the next subsection.

Guided by semantic rate-distortion theory, the end-to-end training objective function of SBGSC includes a distribution loss and a semantic loss:
\begin{equation}
\mathcal{L} = \mathcal{L}_{\text{dist}}(p(\mathbf{x}), p(\hat{\mathbf{x}})) + \lambda \mathcal{L}_{\text{sem}}(\mathbf{x}, \hat{\mathbf{x}}),
\end{equation}
where $\mathcal{L}_{\text{dist}}$ denotes the distribution loss term, such as KL divergence, aiming to match the generated data distribution with the real data distribution. $\mathcal{L}_{\text{sem}}$ denotes the semantic loss, which can take different forms or combinations depending on the downstream application (e.g., Mean Squared Error (MSE) for pixel-level fidelity, or Cross-Entropy for classification accuracy). $\lambda$ is a balancing weight between the loss terms.

\subsection{Schrödinger Bridge-Based Generative Semantic Decoder}
SBGSC introduces a SB from the received semantics $\hat{\mathbf{s}}$ to the output image $\hat{\mathbf{x}}$. Substituting the prior distribution $p_{\text{prior}}(\mathbf{x}_1)=p(\hat{\mathbf{s}})$ and the data distribution $p_{\text{data}}(\mathbf{x}_0)=p(\mathbf{x})$ into the \eqref{eq:sb_sde_forward} and \eqref{eq:sb_sde_backward}, we obtain a system of partial differential equations.
Next, integrating the backward SDE over the time interval $t \in [1, 0]$ yields the expression for the decoder. However, in practical implementations of semantic communication, the analytical integral cannot be directly computed. Instead, the Euler-Maruyama method is used for discretization, resulting in a Langevin dynamics sampling process. Let the time interval $[0,1]$ be divided into $N$ steps with step size $\Delta t = 1/N$. For $k=N-1, \dots, 0$, the functional relationship between the received data $\hat{\mathbf{x}}$ and the received semantics $\hat{\mathbf{s}}$ can be expressed as a recursive function:
\begin{equation}
\mathbf{x}_{k} = \mathbf{x}_{k+1} - f_\theta(\mathbf{x}_{k+1}, k+1)\Delta t + \mathbf{z}_{k+1}\sqrt{\beta_{k+1}\Delta t},
\end{equation}
where $\mathbf{z}_{k+1} \sim \mathcal{N}(\mathbf{0}, \mathbf{I})$, and the model parameterized by $\theta$ is used to approximate the intractable drift term. By accumulating the above equation from $k=N$ (where $\mathbf{x}_N = \hat{\mathbf{s}}$) to $k=0$ (where $\mathbf{x}_0 = \hat{\mathbf{x}}$), we obtain:
\begin{equation}
\hat{\mathbf{x}} = \underbrace{\hat{\mathbf{s}}}_{\text{Received Semantics}} - \underbrace{\sum_{k=1}^{N} f_\theta(\mathbf{x}_k, k)\Delta t}_{\text{Optimal Drift}} + \underbrace{\sum_{k=1}^{N}\mathbf{z}_k\sqrt{\beta_k\Delta t}}_{\text{Random Diffusion}}.
\end{equation}
Here, the first term is the received semantics, representing the starting point of the generation trajectory. The second term is the accumulation of the drift, representing the optimal generation trajectory indicated by the SB. The third term is the accumulation of the diffusion, representing the random perturbations within the generation trajectory.

\subsection{Theoretical Superiority of SBGSC vs. CDM-based GSC}
In this section, we theoretically characterize the advantages of SBGSC over CDM-based GSC \footnote{The complete proofs and discussions are provided in Appendix \ref{appendix:lemma1} to \ref{appendix:corollary2}.}. Lemma \ref{lem:semantic_capacity} first establishes an information capacity inequality showing that the unconstrained JSCC semantic representations provably retain no less mutual information about the source than modality-constrained condition signals, formalizing the information bottleneck imposed by predefined modality priors. We then introduce a distributional distance assumption (Assumption \ref{ass:distribution_distance}) and define Path Kinetic Energy (PKE) (Definition \ref{def:path_kinetic_energy}), upon which Theorem \ref{theo:pke_advantage} proves that the Schrödinger bridge transport in SBGSC achieves strictly lower kinetic energy than the CDM backward diffusion, owing to both a source-informed initialization and a variationally optimal transport path. Two corollaries translate this result into communication performance guarantees: Corollary \ref{cor:hallucination_suppression} shows that SBGSC attains a strictly lower semantic hallucination rate with higher end-to-end mutual information, and Corollary \ref{cor:computational_efficiency} proves that SBGSC requires strictly fewer neural function evaluations for any prescribed generation accuracy tolerance.

\begin{lemma}{\textbf{(Information Capacity Inequality of Semantic Representations).}}
    \label{lem:semantic_capacity}
    Let the source image $\mathbf{x} \in \mathcal{X}$ be drawn from a distribution $p(\mathbf{x})$, and let $d$ denote the fixed feature dimensionality determined by the CBR. Define the unconstrained semantic space $\mathcal{S} = \mathbb{R}^d$ and the family of all Borel measurable encoders thereon as $\mathcal{F}_{\mathcal{S}} = \{ f: \mathcal{X} \to \mathcal{S} \}$. Define the constrained modality-specific condition space $\mathcal{C} = \bigcup_{m=1}^{M} \mathcal{C}_m \subsetneq \mathcal{S}$ as the union of the ranges of $M$ predefined modality-specific feature extractors, and the corresponding encoder family as $\mathcal{F}_{\mathcal{C}} = \{ g: \mathcal{X} \to \mathcal{C} \}$. Then,
    \begin{equation}
        \sup_{f \in \mathcal{F}_{\mathcal{S}}} I(\mathbf{x};\, f(\mathbf{x})) \;\geq\; \sup_{g \in \mathcal{F}_{\mathcal{C}}} I(\mathbf{x};\, g(\mathbf{x})).
        \label{eq:lemma1}
    \end{equation}
\end{lemma}

\textit{Proof Sketch:} 
The key observation is that optimizing over a larger feasible set cannot yield a smaller optimum. Since $\mathcal{C} \subseteq \mathcal{S}$, every encoder $g \in \mathcal{F}_{\mathcal{C}}$ is also a valid encoder in $\mathcal{F}_{\mathcal{S}}$, and the mutual information value it achieves remains unchanged under this identification—this follows from the data processing inequality \cite{beaudry2012intuitive} applied to the identity inclusion, which is trivially invertible on $\mathcal{C}$. Therefore, the set of all mutual information values achievable by $\mathcal{F}_{\mathcal{C}}$ is contained in that achievable by $\mathcal{F}_{\mathcal{S}}$, and the stated inequality follows directly. \textit{The complete proof is provided in Appendix \ref{appendix:lemma1}.} $\hfill\square$

Lemma \ref{lem:semantic_capacity} reveals that modality constraints essentially impose an additional information bottleneck on the encoder design. Specifically, a restricted modality prior inevitably discards semantic components that fall outside its representational capacity—for instance, textual descriptions cannot precisely encode pixel-level textures, while edge maps entirely lose color or material information. In contrast, the JSCC encoder in SBGSC adaptively preserves those semantic dimensions that contribute most to end-to-end reconstruction, free from any modality-specific prior.

\begin{assumption}{\textbf{(Distribution Distance Assumption).}}
    \label{ass:distribution_distance}
    Let $\mu_{\hat{s}}$ denote the induced distribution of the received semantic features in the image space, $p_{\text{data}}$ the distribution of the source images, and $\pi = \mathcal{N}(\mathbf{0}, \mathbf{I}_D)$ the Gaussian prior adopted by the CDM-based scheme. Then,
    \begin{equation}
        \mathcal{W}_2^2\!\left(\mu_{\hat{s}},\; p_{\text{data}}\right) \;<\; \mathcal{W}_2^2\!\left(\pi,\; p_{\text{data}}\right),
        \label{eq:assumption1}
    \end{equation}
    where $\mathcal{W}_2$ is the second-order Wasserstein distance \cite{givens1984class}.
\end{assumption}

The intuition behind this assumption is that the semantic features produced by an end-to-end optimized joint source-channel encoder, even when corrupted by channel noise, retain the principal informative structure of the source data when projected back to the image space, whereas a significantly distributional gap exists between uninformative Gaussian noise and highly structured natural images. \textit{The rigorous definition and justification of this assumption are provided in Appendix \ref{appendix:assumption1}.}

To quantify the efficiency of generative paths, we introduce the following definition.

\begin{definition}{\textbf{(Path Kinetic Energy).}}
    \label{def:path_kinetic_energy}
    For a diffusion process driven by a drift field $\mathbf{u}_t$, i.e., $\mathrm{d}\mathbf{x}_t = \mathbf{u}_t(\mathbf{x}_t)\,\mathrm{d}t + \sigma\,\mathrm{d}\mathbf{W}_t$, its PKE is defined as
    \begin{equation}
        \mathcal{E}(\mathbb{Q}) = \mathbb{E}_{\mathbb{Q}}\!\left[\int_0^1 \left\|\mathbf{u}_t(\mathbf{x}_t)\right\|_2^2 \,\mathrm{d}t\right],
        \label{eq:definition_kinetic_energy}
    \end{equation}
where $\mathbb{Q}$ denotes the path measure induced by the process. $\mathcal{E}(\mathbb{Q})$ measures the cumulative magnitude of the drift field along the entire generative trajectory; a smaller value indicates a shorter, smoother trajectory that requires fewer sampling iterations. In the subsequently proposed method DSBGSC, based on the principle of minimizing PKE, we obtain martingales as distribution transition paths by reconstructing the drift term of the diffusion model (see Section \ref{sec:DSBGSC_method}).
\end{definition}

\begin{theorem}{\textbf{(Path Kinetic Energy Advantage of SBGSC).}}
    \label{theo:pke_advantage}
    Let $\mathbb{Q}^{\mathrm{SB}}(\mu_{\hat{s}}, p_{\text{data}})$ denote the Schrödinger bridge generative path measure of SBGSC, which starts from the distribution of the received semantics $\mu_{\hat{s}}$ and targets the data distribution $p_{\text{data}}$. Let $\mathbb{Q}^{\mathrm{CDM}}$ denote the backward diffusion path measure of CDM-based GSC, which starts from the standard Gaussian prior $\pi$ and targets the conditional data distribution $p_{\text{data}}(\cdot|\mathbf{c})$. Under Assumption \ref{ass:distribution_distance},
    \begin{equation}
        \mathcal{E}\!\left(\mathbb{Q}^{\mathrm{SB}}(\mu_{\hat{s}},\, p_{\text{data}})\right) \;<\; \mathcal{E}\!\left(\mathbb{Q}^{\mathrm{CDM}}\right).
        \label{eq:theorem1}
    \end{equation}
\end{theorem}

\textit{Proof Sketch:} The inequality follows from two independent advantages of (i) starting point and (ii) path that compound together. (i) The minimum kinetic energy of the SB is strictly monotonically increasing with respect to the $\mathcal{W}_2^2$ distance between endpoint distributions, and Assumption \ref{ass:distribution_distance} guarantees $\mathcal{W}_2^2(\mu_{\hat{s}}, p_{\text{data}}) < \mathcal{W}_2^2(\pi, p_{\text{data}})$, so the SB originating from $\mu_{\hat{s}}$ has strictly lower kinetic energy than the one originating from $\pi$. (ii) The SB achieves the global minimum of PKE among all diffusion processes sharing the same endpoint constraints, whereas the backward process of CDM is determined by a fixed, data-agnostic noise schedule (\cite{DDPM}) and is generally not PKE-optimal. Combining these two advantages yields the stated result. \textit{The complete proof is provided in Appendix \ref{appendix:theorem1}.} $\hfill\square$

Theorem \ref{theo:pke_advantage} reveals that the generative process of SBGSC benefits from both a “shorter path” (semantic encoding source information into the initial distribution rather than random noise) and a “smoother path” (optimal transport via the SB). This directly explains the empirical observation that SBGSC requires significantly fewer diffusion steps to achieve the same decoding quality.

This theoretical result leads directly to two key corollaries regarding communication performance: one on hallucination suppression and the other on computational efficiency.

\begin{definition}{\textbf{(Semantic Hallucination Rate).}}
    \label{def:semantic_hallucination_rate}
    We define the semantic hallucination rate $\mathcal{H}$ as the residual uncertainty in the generated $\hat{\mathbf{x}}$ that cannot be explained by the source $\mathbf{x}$:
    \begin{equation}
        \mathcal{H} \triangleq h(\hat{\mathbf{x}} \mid \mathbf{x}),
        \label{eq:hallucination_rate}
    \end{equation}
    where $h(\cdot|\cdot)$ denotes the conditional differential entropy.
\end{definition}

\begin{corollary}{\textbf{(Mutual Information Maximization and Hallucination Suppression).}}
    \label{cor:hallucination_suppression}
    Let $\hat{\mathbf{x}}^{\mathrm{SB}}$ and $\hat{\mathbf{x}}^{\mathrm{CDM}}$ denote the reconstructions produced by SBGSC and CDM-based GSC, respectively. Under the conditions of Theorem \ref{theo:pke_advantage}, the semantic hallucination rates satisfy
\begin{equation}
    \mathcal{H}^{\mathrm{SB}} < \mathcal{H}^{\mathrm{CDM}}.
    \label{eq:corollary1}
\end{equation}
Equivalently, SBGSC achieves a strictly higher end-to-end mutual information, i.e., $I(\mathbf{x};\hat{\mathbf{x}}^{\mathrm{SB}}) > I(\mathbf{x};\hat{\mathbf{x}}^{\mathrm{CDM}})$.
\end{corollary}

\textit{Proof Sketch:} The proof can be completed by combining (i) the CDM side and (ii) the SB side.
(i) The reverse SDE is initialized from an independent Gaussian sample $\boldsymbol{\xi} \sim \pi$ that carries no information about $\mathbf{x}$. Consequently, conditioned on the received semantic descriptor $\hat{\mathbf{c}}$, the reconstruction $\hat{\mathbf{x}}^{\mathrm{CDM}}$ is conditionally independent of $\mathbf{x}$, i.e., $I(\mathbf{x}; \hat{\mathbf{x}}^{\mathrm{CDM}} \mid \hat{\mathbf{c}}) = 0$. This forces the hallucination rate to absorb the full conditional generative entropy: $\mathcal{H}^{\mathrm{CDM}} \geq h(\hat{\mathbf{x}}^{\mathrm{CDM}} \mid \hat{\mathbf{c}})$, which includes the high-dimensional initial Gaussian entropy $h(\boldsymbol{\xi}) = \frac{D}{2}\log(2\pi e)$. (ii) The SB is initialized from $\mathbf{x}_0 = \mathcal{D}_\theta(\hat{\mathbf{s}})$, a deterministic function of the received signal $\hat{\mathbf{s}}$ that already encodes source information. The endpoint coupling $\gamma^{\mathrm{SB}}$ preserves this statistical dependence, and the conditional transport entropy $h(\hat{\mathbf{x}}^{\mathrm{SB}} \mid \mathbf{x}_0)$ is controlled by the PKE $\mathcal{E}(\mathbb{Q}^{\mathrm{SB}})$, which is strictly smaller than $\mathcal{E}(\mathbb{Q}^{\mathrm{CDM}})$ by Theorem \ref{theo:pke_advantage}. Combining these two effects yields $\mathcal{H}^{\mathrm{SB}} < \mathcal{H}^{\mathrm{CDM}}$. \textit{The complete proof is provided in Appendix \ref{appendix:corollary1}.} $\hfill\square$

Corollary \ref{cor:hallucination_suppression} provides an information-theoretic justification for the empirically observed hallucination suppression capability of SBGSC. Since the SB couples the source-informed starting point $\mathbf{x}_0$ with the target data distribution $p_{\text{data}}$ via an entropy-regularized optimal transport plan, it generates content that is maximally consistent with both the received semantic features and the data prior, thereby minimizing hallucinations.

\begin{corollary}{\textbf{(Computational Efficiency and Sampling Step Lower Bound).}}
    \label{cor:computational_efficiency}
    Suppose both schemes discretize their respective generative SDEs using the Euler-Maruyama method with $N$ uniform steps. For a given tolerance $\varepsilon > 0$, define the $\varepsilon$-admissible generation error as the condition $\mathcal{W}_2(\tilde{\mu}_1^{(N)}, \mu_1) \leq \varepsilon$, where $\tilde{\mu}_1^{(N)}$ is the terminal distribution of the $N$-step Euler-Maruyama discretization, and $\mu_1$ is the target distribution. Let $N_{\mathrm{SB}}^*(\varepsilon)$ and $N_{\mathrm{CDM}}^*(\varepsilon)$ denote the minimum number of sampling steps (equivalently, Neural Function Evaluations, NFEs) required by SBGSC and CDM-based GSC, respectively, to achieve $\varepsilon$-admissible error. Then, under the conditions of Theorem \ref{theo:pke_advantage},
    \begin{equation}
        N_{\mathrm{SB}}^*(\varepsilon) < N_{\mathrm{CDM}}^*(\varepsilon),
        \label{eq:corollary2}
    \end{equation}
    i.e., SBGSC attains the same generation accuracy with strictly fewer NFEs.
\end{corollary}

\textit{Proof Sketch:} The key idea is to show that a lower PKE leads to a smaller discretization error constant, which in turn reduces the NFEs needed to meet any given accuracy target. Under standard Lipschitz regularity assumptions on the drift fields, the Euler-Maruyama scheme achieves strong convergence of order $1/2$, yielding $\mathcal{W}_2(\tilde{\mu}_1^{(N)}, \mu_1) \leq C_{\mathcal{S}} / \sqrt{N}$, where $C_{\mathcal{S}}^2 = \alpha_D \mathcal{E}(\mathbb{Q}^{\mathcal{S}}) + \beta_D \bar{\sigma}^2$ with $\alpha_D, \beta_D > 0$ depending only on the dimension and Lipschitz constants, and $\bar{\sigma}$ being the uniform upper bound on the diffusion coefficient. Inverting this bound gives $N_{\mathcal{S}}^*(\varepsilon) = \lceil C_{\mathcal{S}}^2 / \varepsilon^2 \rceil$. Since $C_{\mathcal{S}}^2$ is monotonically increasing in the PKE $\mathcal{E}(\mathbb{Q}^{\mathcal{S}})$, the strict inequality $\mathcal{E}(\mathbb{Q}^{\mathrm{SB}}) < \mathcal{E}(\mathbb{Q}^{\mathrm{CDM}})$ established in Theorem \ref{theo:pke_advantage} directly implies $N_{\mathrm{SB}}^*(\varepsilon) < N_{\mathrm{CDM}}^*(\varepsilon)$ for all sufficiently small $\varepsilon$. \textit{The complete proof is provided in Appendix \ref{appendix:corollary2}.} $\hfill\square$

\begin{figure}[!t]
    \centering
    \includegraphics[width=1\linewidth]{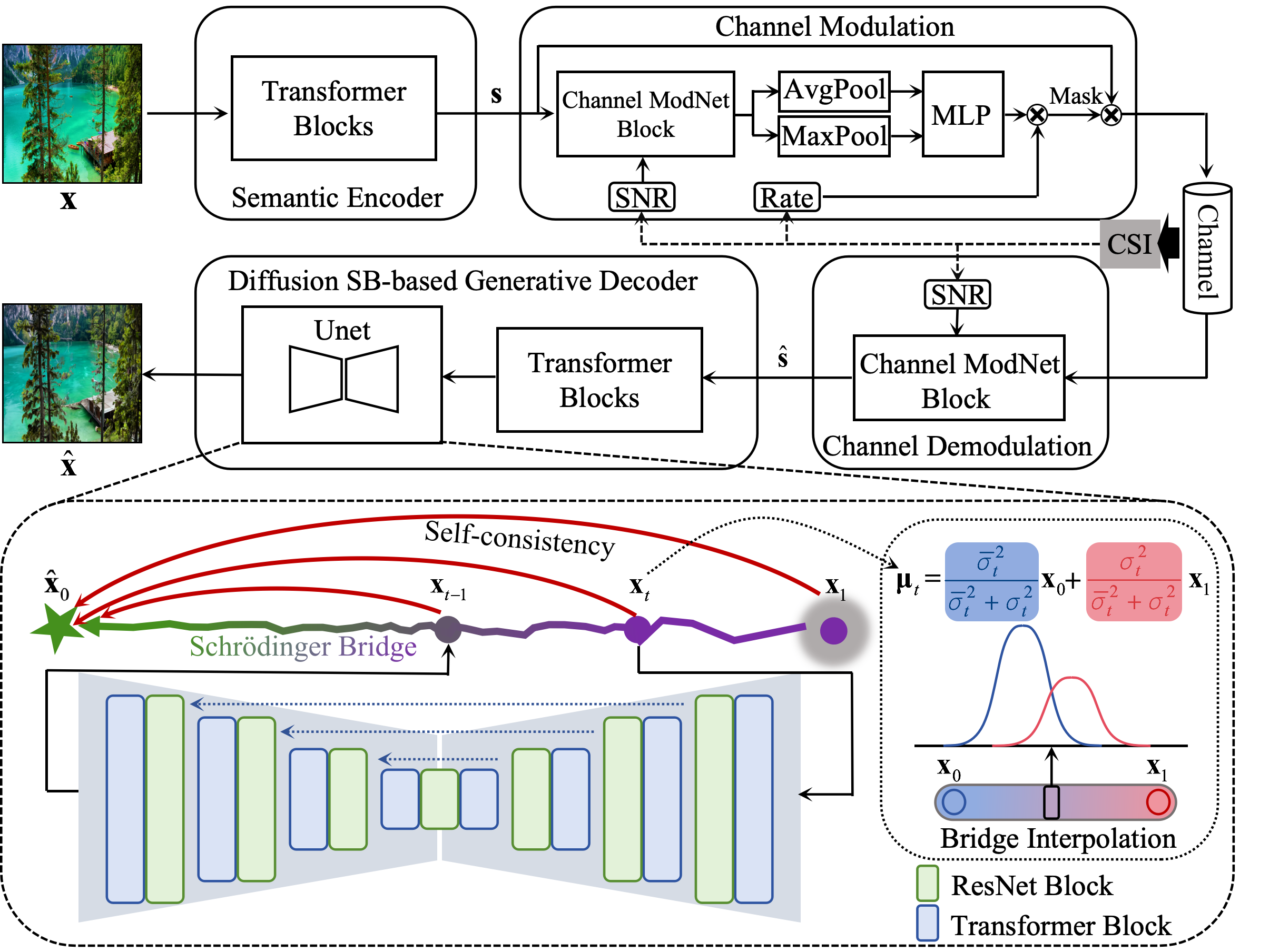}
    \caption{The overall architecture of proposed DSBGSC.The optimal transfer from semantic distribution to data distribution is achieved via the least action principle of DSB. Meanwhile, the self-consistency property is incorporated, enabling any intermediate state to directly map to the source image for efficient few-step semantic perception.}
    \label{fig3:DSBGSC framework}
\end{figure}

\section{Proposed DSBGSC Method}
\label{sec:DSBGSC_method}
This section introduces our proposed generative semantic communication method DSBGSC based on DSB\cite{chen_likelihood_2022}.

\subsection{DSBGSC method overview}
As illustrated in Fig. \ref{fig3:DSBGSC framework}, the proposed DSBGSC consists of four components: semantic encoder, channel modulation module, channel demodulation module, and a DSB-based decoder.

At the transmitter, given an transmitted image $\mathbf{x}$, we choose a Swin Transformer-based architecture \cite{yang_swinjscc_2025}. By utilizing a shift-window self-attention mechanism, it can effectively capture local fine-grained features and global semantic dependencies. This process effectively eliminates spatial redundancy and extracts compact, critical semantic representations $\mathbf{s}$. To adapt to dynamic channel state information (CSI) and bandwidth constraints, this module performs SNR adaptation (i.e., Channel ModNet Block \cite{yang_swinjscc_2025}) and rate control. It dynamically scales features based on channel environments and utilizes a multilayer perceptron (MLP)-based channel attention mechanism for selecting the top-$k$ features, effectively pruning redundant channels while maintaining semantic robustness. At the receiver, this module recalibrates and aligns impaired semantic features using the SNR. This operation effectively compensates for channel fading, providing a robust and accurate semantic representation $\hat{\mathbf{s}}$ for subsequent decoding. Finally, to address the inverse problem of semantic perception, the decoder concatenates low-dimensional semantic features with a high-dimensional image. It projects the received semantic features into a destroyed semantic distribution, and then obtains the received image $\hat{\mathbf{x}}$ through the DSB.

\subsection{DSB-based generative decoder}

At the receiver, the decoding process must not only recover the pixel-level content but also achieve high-quality semantic perception while strictly preserving semantic consistency. To this end, we propose a generative decoder based on the DSB.

First, this decoder bridges the low-dimensional semantic features and the high-dimensional image. Specifically, the received features are projected into the image space to form a destroyed yet semantically consistent distribution. Crucially, this distribution serves as the boundary distribution, guiding the subsequent generative process.

To ensure tractability within the diffusion model framework, we reformulate the backward process of the SB \eqref{eq:sb_sde_backward} as a modified SDE:

\begin{equation}
\small
\mathrm{d}\mathbf{x}_t = \left[ \mathbf{f}_t + \beta_t \nabla \log \Psi(\mathbf{x}_t) - \beta_t \nabla \log q_t \right] \mathrm{d}t + \sqrt{\beta_t} \, \mathrm{d}\bar{\mathbf{w}}_t,
\label{eq:reverse_sde}
\end{equation}
where $\nabla \log q_t$ represents the score function of the marginal distribution of the intermediate states. The Nelson duality relation $\Psi(\mathbf{x}_t) \hat{\Psi}(\mathbf{x}_t) = q_t(\mathbf{x}_t)$, bridges the forward and backward processes, ensuring the consistency of their marginal distributions. Furthermore, $\nabla \log \Psi(\mathbf{x}_t, t)$ serves as the crucial nonlinear correction term introduced by the SB. It acts as a global nonlinear velocity field that counteracts the linear drift towards Gaussian noise, effectively pulling the diffusion trajectory back to the data distribution.

The linear drift term directed towards a zero mean, commonly used in standard Score-Based Generative Models (SGMs), is no longer applicable here because it ignores the structural correlation between the target state $\mathbf{x}_1$ and the start state $\mathbf{x}_0$. Moreover, solving for $\mathbf{f}(\mathbf{x}_t, t)$ presents significant numerical challenges due to its nonlinear coupling with the Schrödinger potential. To construct an efficient, analytically tractable, and controllable bridging process, we adopt the zero-drift assumption\cite{I2SB}, which sets the drift term of the reference process to $\mathbf{f}(\mathbf{x}_t, t) \equiv 0$. This implies that the reference process is a martingale, i.e., $\mathbb{E}[\mathbf{x}_t | \mathbf{x}_0] = \mathbf{x}_0$. Without the intervention of the Schrödinger potential, the signal $\mathbf{x}_0$ undergoes pure diffusion without mean reversion or amplitude attenuation. Consequently, the final transport trajectory is entirely shaped by the potential term $\nabla \log \Psi$, which ensures that the learned path corresponds to the shortest path in the Euclidean space.

Under this zero-drift assumption, the intractable bridge posterior distribution $q(\mathbf{x}_t | \mathbf{x}_0, \mathbf{x}_1)$ reduces to an analytical Gaussian form. Specifically, both the forward process $p(\mathbf{x}_t | \mathbf{x}_0)$ and the process from the current state to the target state $p(\mathbf{x}_1 | \mathbf{x}_t)$ are governed by the properties of Brownian motion, manifesting as Gaussian distributions:
\begin{align}
p(\mathbf{x}_t | \mathbf{x}_0) &= \mathcal{N}\left( \mathbf{x}_t; \mathbf{x}_0, \sigma_t^2 \mathbf{I} \right), \label{eq:forward_trans} \\
p(\mathbf{x}_1 | \mathbf{x}_t) &= \mathcal{N}\left( \mathbf{x}_1; \mathbf{x}_t, \bar{\sigma}_t^2 \mathbf{I} \right), \label{eq:backward_trans}
\end{align}
where \(\sigma_t^2 \triangleq \int_0^t \beta_\tau d\tau\) and \(\bar{\sigma}_t^2 \triangleq \int_t^1 \beta_\tau d\tau\) denote the cumulative variances from \(\mathbf{x}_0\) to \(\mathbf{x}_t\) and from \(\mathbf{x}_t\) to \(\mathbf{x}_1\), respectively. 

\begin{algorithm}[!t]
  \caption{Training}
  \label{alg:training}
  \begin{algorithmic}[1]
    \vspace{0.2em}
    \State \textbf{Input:} Data distribution \( p(\mathbf{x}) \), semantic distribution \( p(\mathbf{s}) \), learning rate $\eta$
    \Repeat
    \State Sample \( t \sim \mathcal{U}([0, 1]) \),\( \mathbf{x}_0 \sim p(\mathbf{x}) \), \( \mathbf{x}_1 \sim p(\mathbf{s}) \)
    \State Sample \( \mathbf{x}_t \sim q(\mathbf{x}_t | \mathbf{x}_0, \mathbf{x}_1) \) 
    \State  \( \theta \leftarrow \theta - \eta \nabla_\theta \left( \left\| \boldsymbol{\epsilon}_\theta(\mathbf{x}_t, t) - \frac{\mathbf{x}_t - \mathbf{x}_0}{\sigma_t} \right \|^2 \right) \) 

\Until{Convergence}
  \end{algorithmic}
\end{algorithm}

According to Bayes' theorem, the bridge posterior $q(\mathbf{x}_t | \mathbf{x}_0, \mathbf{x}_1)$ is proportional to the product of these two marginal distributions. Using the product property of Gaussians, this posterior can be analytically derived as a new Gaussian distribution $\mathcal{N}(\mathbf{x}_t; \boldsymbol{\mu}_t, \Sigma_t \mathbf{I})$. Specifically, its mean $\boldsymbol{\mu}_t$ and variance $\Sigma_t$ are determined by:
\begin{equation}
\boldsymbol{\mu}_t = \frac{\bar{\sigma}_t^2}{\sigma_t^2 + \bar{\sigma}_t^2} \mathbf{x}_0 + \frac{\sigma_t^2}{\sigma_t^2 + \bar{\sigma}_t^2} \mathbf{x}_1, \quad \Sigma_t = \frac{\sigma_t^2 \bar{\sigma}_t^2}{\sigma_t^2 + \bar{\sigma}_t^2}.
\end{equation}

For practical sampling and training, $\mathbf{x}_t$ can be expressed using the reparameterization trick:
\begin{equation}
\mathbf{x}_t = \boldsymbol{\mu}_t + \sqrt{\Sigma_t} \boldsymbol{\epsilon}, \quad \boldsymbol{\epsilon} \sim \mathcal{N}(\mathbf{0}, \mathbf{I}).
\end{equation}
This structure implies that $\mathbf{x}_t$ lies on the linear interpolation $\boldsymbol{\mu}_t$ between the source data and semantic features, perturbed solely by the bridge process uncertainty $\sqrt{\Sigma_t} \boldsymbol{\epsilon}$. This aligns perfectly with the self-consistency principle proposed in Consistency Models (CMs)\cite{CM}, which posits that any state $\mathbf{x}_t$ along the same generative trajectory should contain sufficient structural information to map directly back to its initial state $\mathbf{x}_0$. We leverage this principle to solve the Schrödinger Bridge. Since SB aims to find the optimal transport path of “least action”, enforcing self-consistency naturally constrains the generative trajectory to be the most direct and unswerving route connecting the semantic and data distributions.
To achieve this consistent transport, we require the learned drift term of the diffusion process to possess a specific geometric property: at any time step $t$, the driving force generated by the potential gradient $\nabla \log \Psi$ must explicitly point towards the target state $\mathbf{x}_0$. Mathematically, while this potential gradient is conventionally parameterized by the noise predictor $\boldsymbol{\epsilon}_\theta$, the self-consistency principle redefines it as a direct displacement vector, i.e., $\nabla \log \Psi \propto \boldsymbol{\epsilon}_\theta(\mathbf{x}_t, t) \approx \frac{\mathbf{x}_t - \mathbf{x}_0}{\sigma_t}$. Based on this consistency theory, we can geometrically define the optimal prediction direction for the network, which represents the shortest tangent vector along the SB optimal transport geodesic pointing towards the data distribution.

\begin{algorithm}[!t]
  \caption{Sampling}
  \label{alg:sampling}
  \begin{algorithmic}[1]
    \vspace{0.2em}
   \State \textbf{Input:} Input \( \mathbf{x}_1 \), number of steps \( N \)
    \State Initialize \( t \leftarrow 1 \), \( \mathbf{x}_t \leftarrow \mathbf{x}_1 \)
    \For{$k = 1$ to $N$}
        \State \( \hat{\mathbf{x}}_0 = \mathbf{x}_t - \sigma_t \boldsymbol{\epsilon}_\theta(\mathbf{x}_t, t) \)
        \State \( t \leftarrow t - 1/N \)
        \State Sample \( \mathbf{x}_t \sim q(\mathbf{x}_t | \hat{\mathbf{x}}_0, \mathbf{x}_{t+1/N}) \)
    \EndFor
    \State \textbf{Output:} \( \hat{\mathbf{x}}_0 \)
  \end{algorithmic}
\end{algorithm}

Guided by this geometric insight, rather than directly estimating the source data $\mathbf{x}_0$, we train a neural network $\boldsymbol{\epsilon}_\theta(\mathbf{x}_t, t)$ to predict the normalized direction (or effective noise) from $\mathbf{x}_0$ to $\mathbf{x}_t$. According to the forward transition $p(\mathbf{x}_t | \mathbf{x}_0)$, the exact noise component injected into $\mathbf{x}_t$ is given by $\frac{\mathbf{x}_t - \mathbf{x}_0}{\sigma_t}$. Therefore, we define the training objective as:
\begin{equation}
\mathcal{L}_{SB} = \mathbb{E}_{t,\mathbf{x}_0, \mathbf{x}_t}
\left[ \left| \boldsymbol{\epsilon}_\theta(\mathbf{x}_t, t) - \frac{\mathbf{x}_t - \mathbf{x}_0}{\sigma_t} \right|^2 \right].
\label{eq:loss_sb}
\end{equation}
Although formulated in the noise prediction space, this objective fundamentally enforces the same self-consistency, as the source data can be implicitly recovered via $\mathbf{x}_0 \approx \mathbf{x}_t - \sigma_t \boldsymbol{\epsilon}_\theta(\mathbf{x}_t, t)$. By minimizing this objective function, we compel the network to adhere to this consistency at every time step. Consequently, the network learns to implicitly define a nonlinear velocity field directed towards the data distribution. This ensures that during the inference phase, even when starting from a general semantic, the generated trajectory can precisely converge to the data distribution along the consistent direction.

Upon completion of training Algorithm \ref{alg:training}, the decoder $f_\theta$ serves as a drift estimator for the backward SDE. Leveraging the established self-consistency property, we adopt an accelerated sampling strategy Algorithm \ref{alg:sampling} akin to that of CMs. Rather than relying on small step-size Euler-Maruyama integration, this strategy executes the generative process through the iterative refinement of the estimate for $\mathbf{x}_0$.

At each diffusion sampling step $t$, we derive the reconstruction estimate $\hat{\mathbf{x}}_0$ directly from the decoder output $f_\theta(\mathbf{x}_t, t)$. Subsequently, we sample the next state $\mathbf{x}_{t-1}$ from the bridge posterior distribution $q(\mathbf{x}_{t-1} | \hat{\mathbf{x}}_0, \mathbf{x}_t)$. This consistent bridge interpolation mechanism enables semantic perception quality with an extremely low NFEs.

\section{Experimental Settings and Results}
\label{sec:simulation}
\subsection{Experimental Setup}
We performed parallel training on four NVIDIA 5090 GPUs. To emulate practical wireless environments, both AWGN and Rayleigh fading channels are integrated into the training pipeline. We adopt a multi-stage training paradigm to optimize the proposed DSBGSC. Initially, to establish a robust baseline for semantic perception quality, the SwinJSCC backbone is pre-trained at a fixed SNR of 10 dB. During this phase, the model is trained for 300 epochs with a batch size of 64, utilizing the AdamW optimizer with an initial learning rate of $1 \times 10^{-4}$ and a cosine annealing scheduler.
Building upon the pre-trained backbone, the second stage introduces rate and channel adaptability. We define the discrete set of available CBRs as $\{1/192, 1/96, 1/48, 1/24, 1/16, 1/12, 1/8\}$. By randomly sampling a CBR from this set during training, the model learns to dynamically accommodate varying transmission rates. Simultaneously, the channel SNR is uniformly sampled from $-13$ dB to $13$ dB, significantly improving the model's robustness against dynamic noise conditions. Following this, we freeze the parameters of the SwinJSCC backbone to isolate and optimize the Diffusion SB-based generative decoder. This provides the decoder with stable semantic representations. Upon convergence of the decoder, the entire network undergoes an end-to-end joint fine-tuning. During this final stage, the learning rate is decayed to $1 \times 10^{-5}$ to ensure stable convergence and maximize the overall system performance.

\begin{figure*}[!t]
\centering
\includegraphics[width=0.9\textwidth]{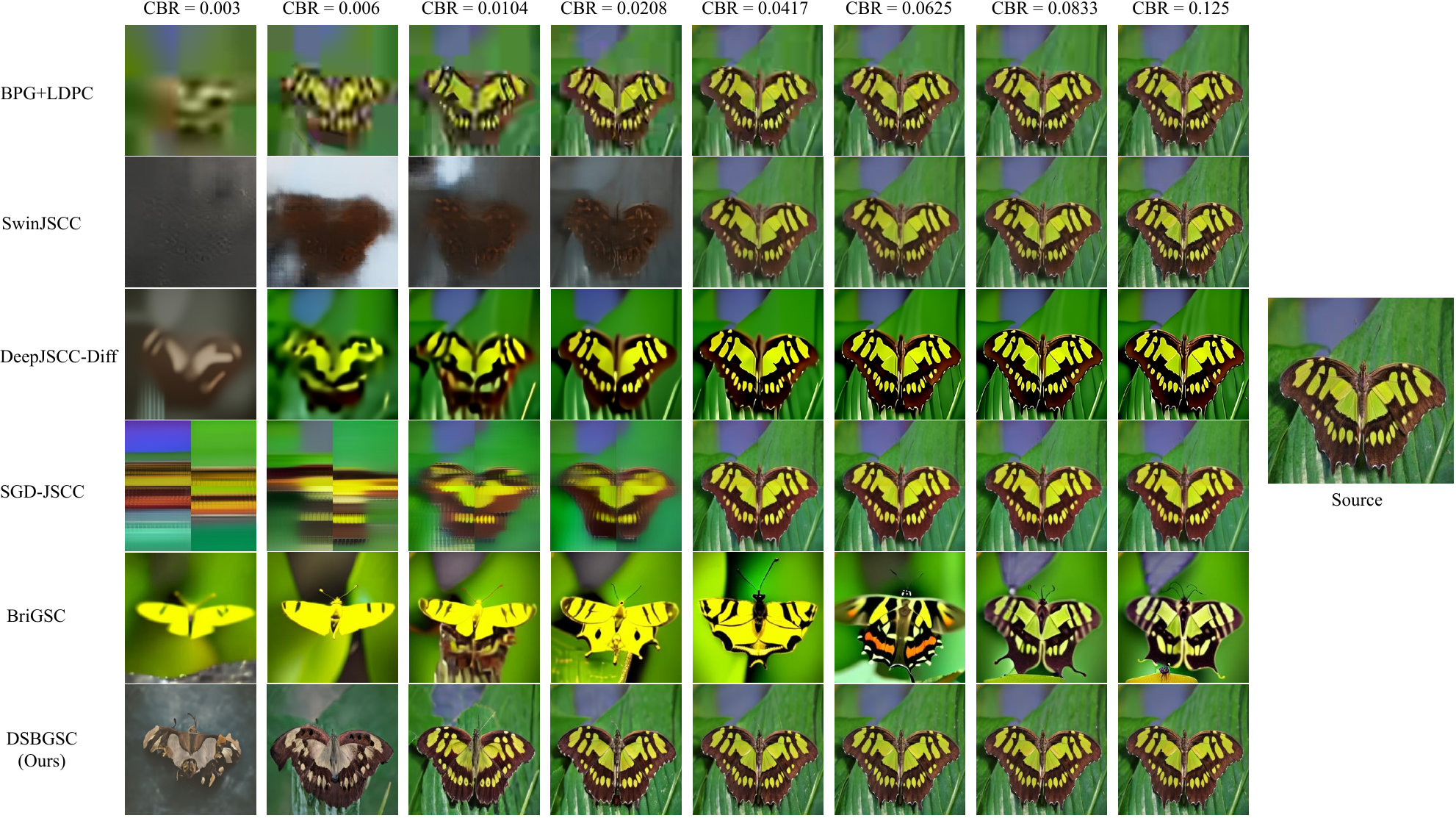}
\caption{Visual comparison of semantic perception quality among different methods under AWGN channel at SNR = 7dB.}
\label{fig4:results CBR}
\end{figure*}

\begin{figure}[htbp]
    \centering
    \renewcommand{\baselinestretch}{1.0} 

    \subfloat[{\small PSNR vs. CBR}\label{fig:cbr_grid-subfig1}]{%
        \includegraphics[width=0.5\linewidth]{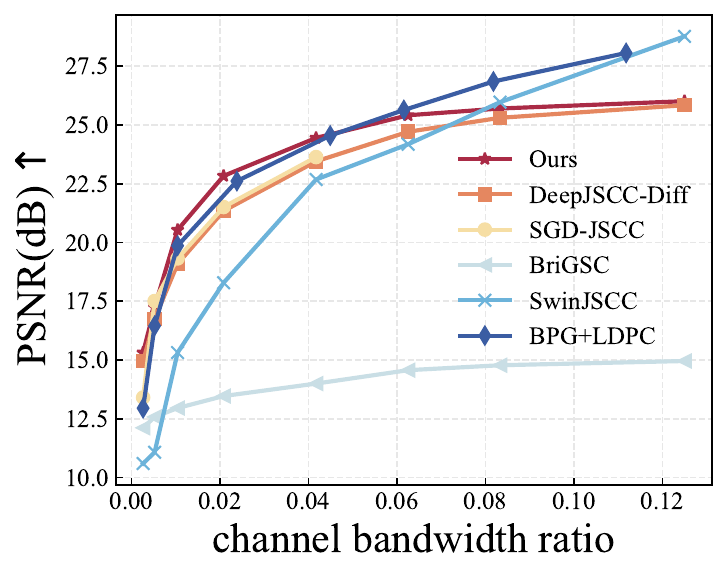}}
    \hfill 
    \subfloat[{\small LPIPS vs. CBR}\label{fig:cbr_grid-subfig2}]{%
        \includegraphics[width=0.5\linewidth]{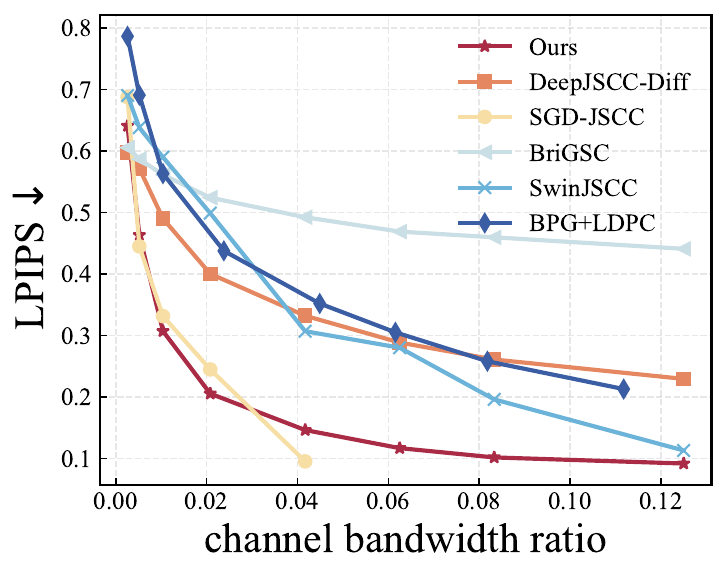}}

    \vspace{0pt}
    
    \subfloat[{\small MS-SSIM vs. CBR}\label{fig:cbr_grid-subfig3}]{%
        \includegraphics[width=0.5\linewidth]{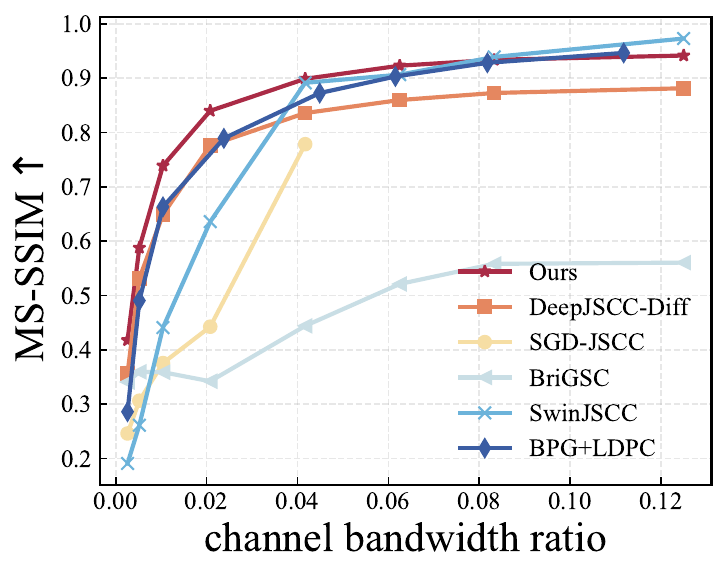}}
    \hfill
    \subfloat[{\small FID vs. CBR}\label{fig:cbr_grid-subfig4}]{%
        \includegraphics[width=0.5\linewidth]{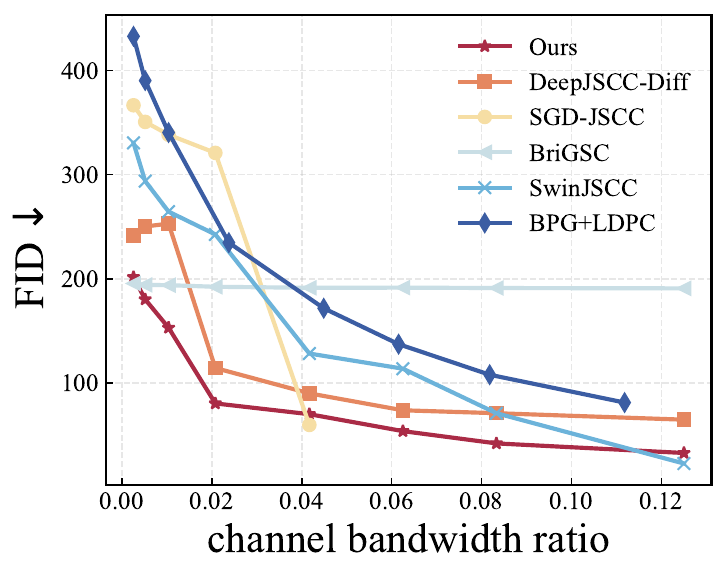}}

    \captionsetup{font={normalsize, rm}}
    \caption{Comparison of semantic perception quality among different methods at SNR=7dB when CBR varies.}
    \label{fig5:cbr_grid}
\end{figure}

\subsection{Dataset}
We select a subset of the ImageNet-1K dataset \cite{ImageNet} with high semantic diversity, comprising approximately 100,000 images, for training. During preprocessing, all images are cropped and resized to a 256$\times$256 resolution. The broad semantic distribution of this dataset enables the model to learn general semantic priors. To validate the model's zero-shot generalization capability on unseen data, we employ the DIV2K validation set for testing. This dataset comprises high-quality 2K-resolution images, which we randomly crop to 256$\times$256 for evaluation.
We selected a highly semantically diverse subset of the ImageNet-1K dataset \cite{ImageNet} comprising approximately 100,000 images for training. All images were cropped and resized to 256$\times$256 resolution during preprocessing. The dataset's extensive category coverage ensures the model learns generalizable semantic priors. To validate the model's zero-shot generalization capability on unseen data, we employed the DIV2K validation set for testing. This dataset comprises high-quality 2K-resolution images, which we randomly cropped to 256$\times$256 for evaluation.
\subsection{Evaluation Metrics}
Since semantic communication pursues not only distortion-level alignment but also perceptual authenticity, we evaluate the semantic perception quality using both distortion and perceptual metrics. For distortion metrics, Peak Signal-to-Noise Ratio (PSNR) measures pixel-level error levels, while Multiscale Structural Similarity (MS-SSIM \cite{MS-SSIM}) assesses the similarity of luminance, brightness, and structural information. For perceptual metrics, we utilize a pre-trained VGG network to compute Learned Perceptual Image Patch Similarity (LPIPS \cite{LPIPS}) in the feature space, which aligns more closely with human visual perception than PSNR. Additionally, we employ Fréchet Inception Distance (FID \cite{FID}) to measure the distance between the semantic distribution and the data distribution.

\begin{figure*}[!t]
\centering
\includegraphics[width=0.9\textwidth]{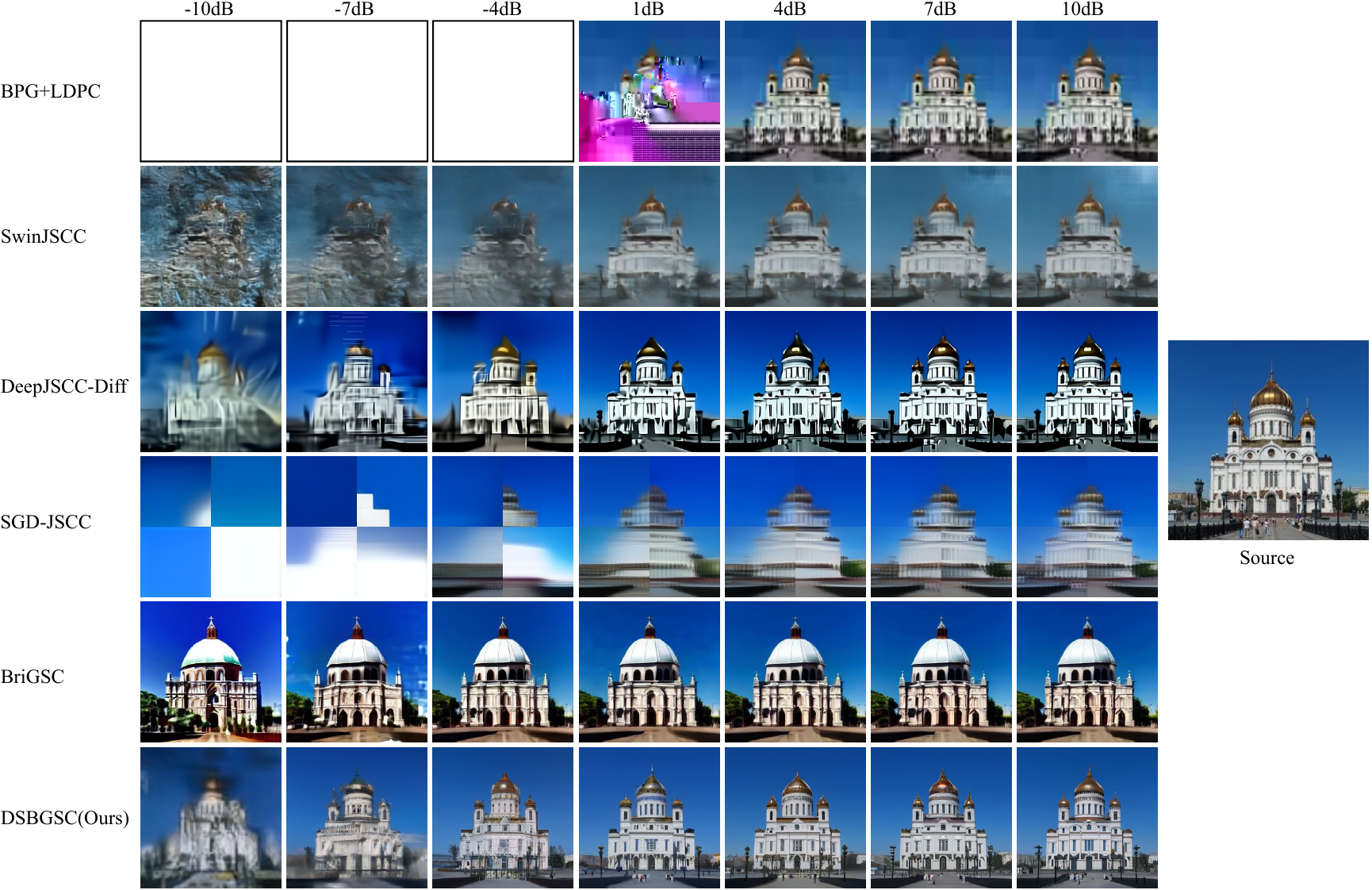}
\caption{Visual comparison of semantic perception quality among different methods under AWGN channel with CBR = 1/48.}
\label{fig6:results SNR}
\end{figure*}

\subsection{Comparison Methods}
We conducted a comprehensive comparison of DSBGSC against five popular image transmission schemes, spanning traditional encoding to the latest generative semantic communication methods: The conventional separated approach BPG + LDPC combines the widely recognized state-of-the-art image compression standard Better Portable Graphics (BPG) with the 5G standard channel coding Low-Density Parity-Check (LDPC). This represents the performance of traditional source-channel separation design. SwinJSCC \cite{yang_swinjscc_2025}, based on DeepJSCC, breaks through the performance bottleneck of traditional CNN architectures by enhancing model representation capabilities, serving as the current benchmark method in the DeepJSCC domain. DeepJSCC-Diff\cite{DeepJSCC-Diff}, an early representative work in generative JSCC, innovatively integrates the Denoising Diffusion Probabilistic Models(DDPM \cite{DDPM}) with DeepJSCC. It achieves high-fidelity restoration by transmitting low-resolution degraded images and leveraging diffusion models. SGD-JSCC \cite{zhang_semantics-guided_2025} employs text descriptions and edge maps as guidance signals to drive the diffusion model during denoising and generation, demonstrating outstanding semantic structure preservation and representing a recent state-of-the-art approach. BriGSC \cite{gao_bridging_2025} proposes a cross-scale generative semantic communication method for multi-scale semantic transmission needs. By jointly extracting textual and visual features, it achieves rate-adaptive codec based on CSI and utilizes diffusion models to fuse multimodal conditions for image generation, demonstrating outstanding performance under low-bandwidth conditions.

\begin{figure}[htbp]
    \centering

    \subfloat[PSNR vs. SNR\label{fig:snr_grid-subfig1}]{\includegraphics[width=0.5\linewidth]{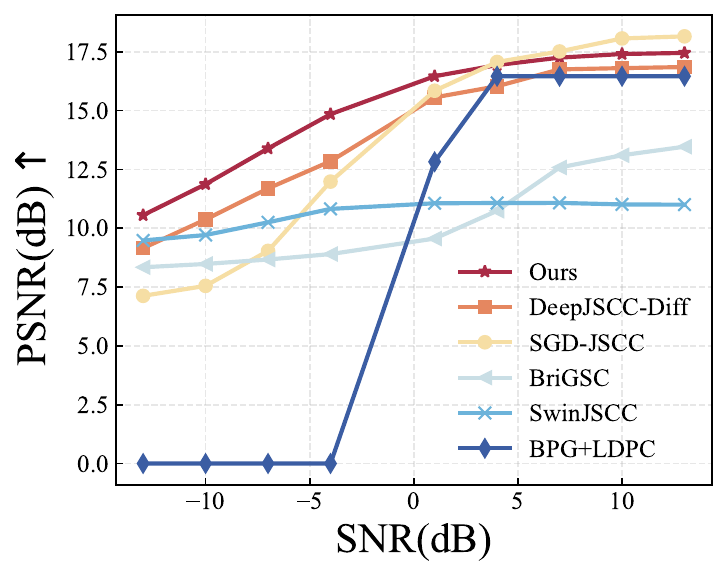}}
    \hfill  
    \subfloat[LPIPS vs. SNR\label{fig:snr_grid-subfig2}]{\includegraphics[width=0.5\linewidth]{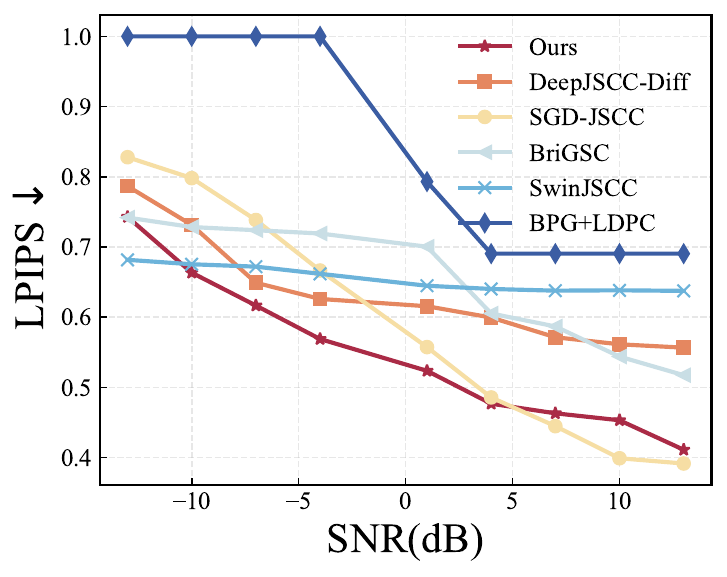}}

    \vspace{0pt}
    
    \subfloat[MS-SSIM vs. SNR\label{fig:snr_grid-subfig3}]{\includegraphics[width=0.5\linewidth]{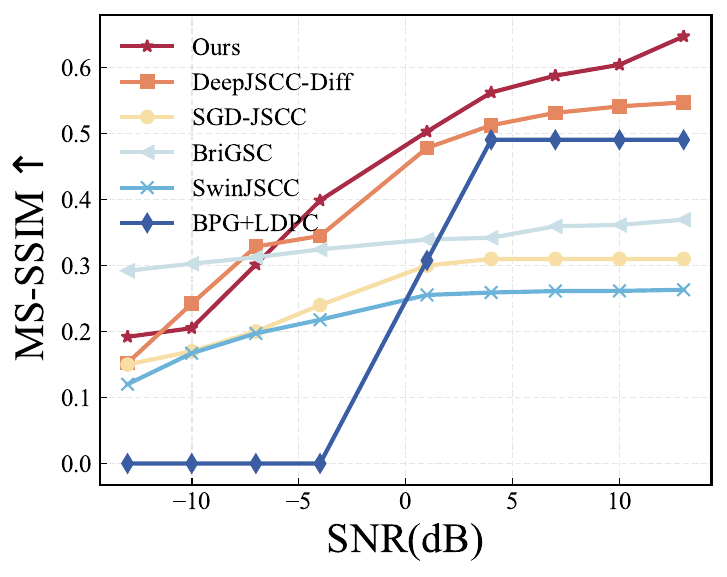}}
    \hfill
    \subfloat[FID vs. SNR\label{fig:snr_grid-subfig4}]{\includegraphics[width=0.5\linewidth]{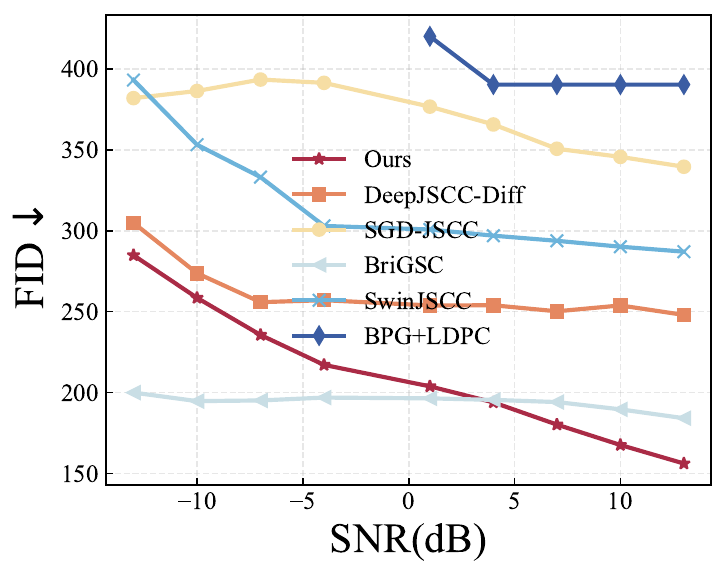}}
    
    \caption{CBR set to 1/192, semantic perception quality comparison among methods when SNR varies.}
    \label{fig7:snr_grid}
\end{figure}
\subsection{Experiments and Results Analysis}
\subsubsection{Effect of CBR on Performance} 

With SNR fixed at 7 dB, we first quantitatively evaluate the performance of different methods under varying CBRs, as shown in Fig. \ref{fig5:cbr_grid}. In terms of distortion metrics (Figs. \ref{fig5:cbr_grid}(a) and (c)), our proposed DSBGSC achieves the best performance in the low CBR. Although SwinJSCC achieves competitive PSNR at CBR $>$ 0.1, optimizing solely for pixel-level loss tends to result in over-smoothed semantic perception, thereby losing high-frequency texture details. Regarding perceptual metrics (Figs. \ref{fig5:cbr_grid}(b) and (d)), our method yields the best overall performance on average, with the exception of a few specific data points. This perceptual quality is achieved by overcoming the inherent bottlenecks of other generative semantic communication methods. Specifically, DeepJSCC-Diff suffers from irreversible information loss due to downsampling at the transmitter. SGD-JSCC relies on edge map features, which lack adaptive scaling capabilities with respect to channel conditions. This limitation leads to the failure of guiding information under narrow-band and high-noise environments, causing significant structural distortion and geometric collapse in the semantic perceptions. Furthermore, while BriGSC improves semantic consistency via textual priors, its over-reliance on text conditions often leads to severe deviations in spatial arrangement and fine details. Thanks to the introduction of the SB, DSBGSC achieves excellent scores across all metrics at lower CBR, demonstrating its significant advantages in perceptual quality.

From a qualitative perspective, the visual results in Fig. \ref{fig4:results CBR} demonstrate that our method significantly outperforms other methods. This superiority is particularly evident under extremely low CBR conditions, where DSBGSC maintains remarkable semantic fidelity. For instance, at an extremely low bandwidth of CBR = 1/192, the foreground butterfly and background leaves perceived by our method remain clearly distinguishable, with accurate structures and textures. In contrast, traditional methods and other semantic communication methods suffer from severe degradation, rendering the objects completely unrecognizable or structurally collapsed. This visual evidence further corroborates the exceptional capability of DSBGSC in preserving semantic and perceptual details under harsh channel conditions.

\begin{figure}[h]
\centering
\includegraphics[width=0.95\linewidth]{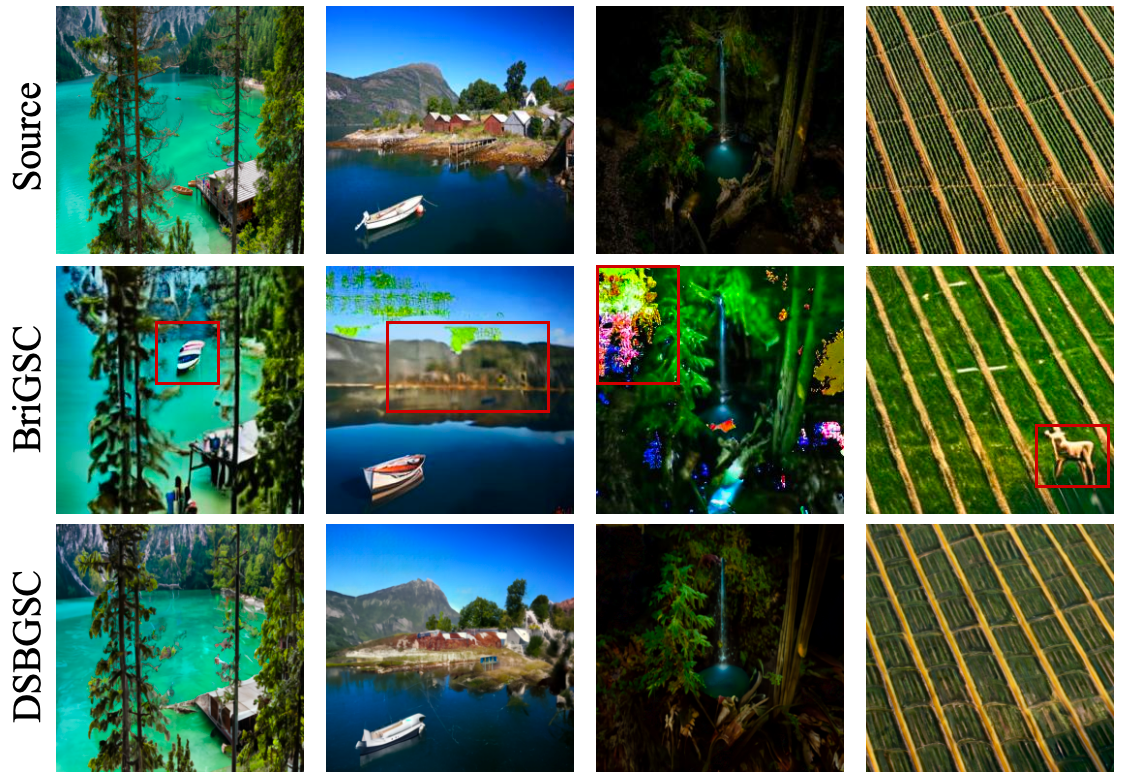}
\caption{Visual comparison of hallucination suppression. Red boxes highlight severe semantic hallucinations generated by BriGSC.}
\label{fig8-Hallucination}
\end{figure}
\begin{figure*}[ht]
\centering
\includegraphics[width=\textwidth]{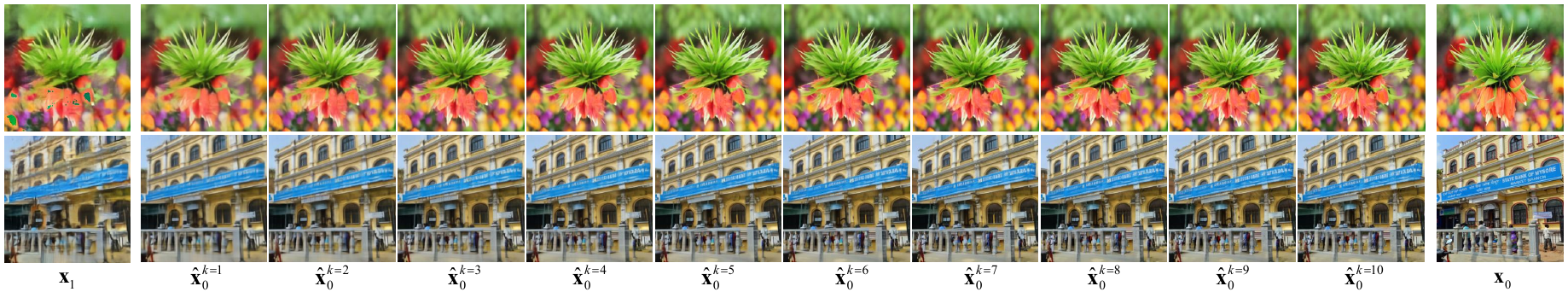}
\caption{Generative processes for semantic and data distribution transfer with NFE=10. Each figure depicts the direct prediction performance of ${{\mathbf{x}}_0}$ at the current state during the generative evolution.}
\label{fig9}
\end{figure*}

\subsubsection{Effect of SNR on Performance} To investigate the DSBGSC's robustness against channel noise, we fix the CBR at 1/48 and evaluate its performance over an SNR range from -13 dB to 13 dB. As illustrated by the quantitative results in Fig. \ref{fig7:snr_grid} and the visual comparisons in Fig. \ref{fig6:results SNR}, the methods exhibit various limitations under harsh channel conditions. First, traditional communication method BPG + LDPC suffer from a severe "cliff effect," leading to complete transmission failure when the SNR $<$ 0 dB. Second, although non-generative JSCC method SwinJSCC is adaptively trained against channel noise to avoid abrupt collapse, their deep features remain highly sensitive to noise, resulting in degraded semantic perception quality at low SNRs. Third, for other generative semantic communication methods, the channel noise under low SNR is often misinterpreted by the generative models as conditional inputs, leading to severe hallucinations. For instance, at SNRs of -10 dB and -7 dB, SGD-JSCC generates completely distorted structures or meaningless patterns, failing to perceive the original semantic. In contrast, our proposed DSBGSC achieves the best overall performance, demonstrating excellent degradation characteristics. Quantitatively, as shown in Fig. \ref{fig7:snr_grid}, our method consistently yields the best results across all evaluation metrics under varying SNRs. Qualitatively, from the visual comparison of Fig. \ref{fig6:results SNR}, even under the extremely harsh condition of -10 dB, the content generated by our method maintains clear structural and semantic fidelity without collapsing. 

\subsubsection{Hallucination Suppression} We compared the visual performance of two generative methods under the same channel noise and transmission rate. BriGSC relies on conditional generation and is highly susceptible to impaired features, leading to obvious hallucinations. Specifically, as shown in the red box in Fig. \ref{fig8-Hallucination}, BriGSC generates severe hallucinations, such as entirely fabricated objects, which deviate significantly from the source images. In contrast, our proposed DSBGSC utilizes a SB to directly guide semantics back to the data distribution. Our method effectively suppresses the hallucinations common in GSC, as validated by Corollary \ref{cor:hallucination_suppression}.

\begin{table}[h]
\centering
\caption{Performance comparison of different generative semantic communication methods}
\resizebox{\linewidth}{!}{
\begin{tabular}{lccc}
\toprule
Method & PSNR $\uparrow$ & LPIPS $\downarrow$ & Time(s) \\
\midrule
DeepJSCC–Diff & 25.31 & 0.26 & 260 \\
SGD–JSCC & 21.51 & 0.24 & 9.68 \\
BriGSC & 14.78 & 0.46 & 16.24 \\
\textbf{Ours} & \textbf{25.70} & \textbf{0.10} & \textbf{2.6} \\
\bottomrule
\end{tabular}
}
\label{tab1}
\end{table}

\begin{figure}[hbp]
    \centering
    \subfloat[PSNR vs. SNR\label{fig:rayleigh-subfig1}]{\includegraphics[width=0.5\linewidth]{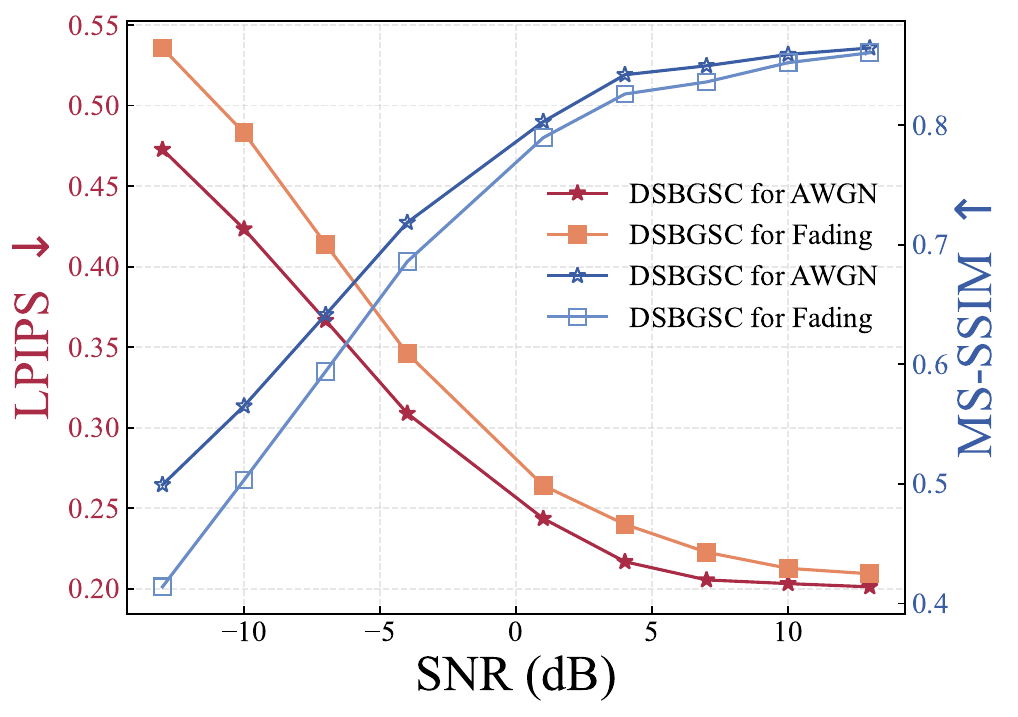}}
    \hfill 
    \subfloat[LPIPS vs. SNR\label{fig:rayleigh-subfig2}]{\includegraphics[width=0.5\linewidth]{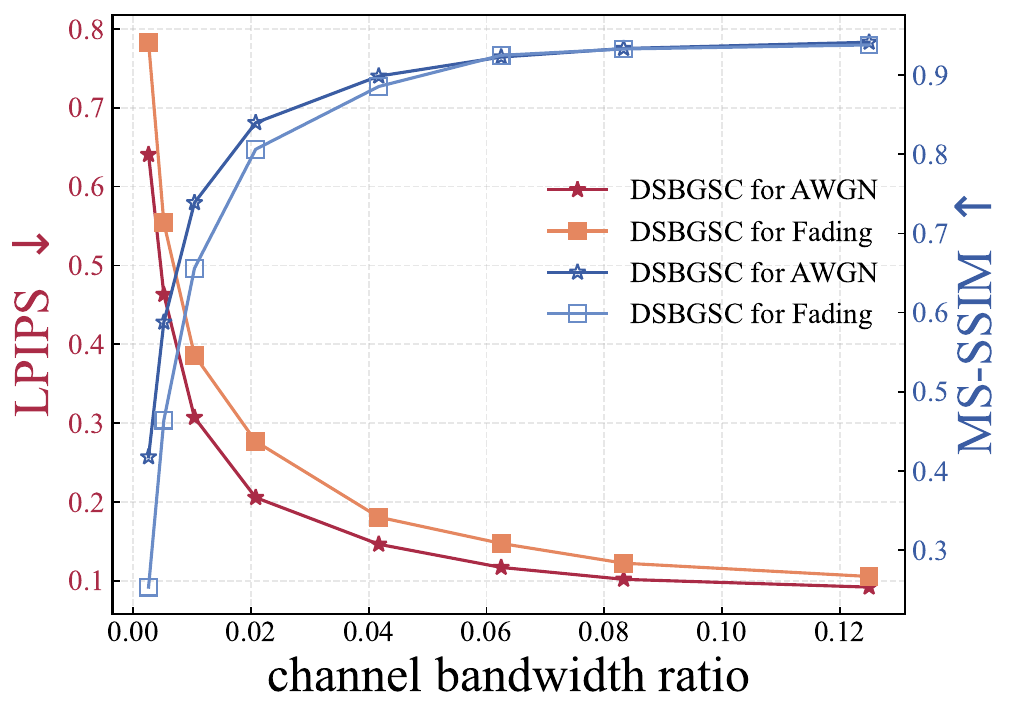}}
    \caption{Comparison of semantic perceptual quality under different channels.}
    \label{fig:rayleigh}
\end{figure}

\subsubsection{Efficiency} 
Because DSBGSC performs direct distributed transfer, it is able to decode high-quality images with a significantly reduced NFE. Fig. \ref{fig9} illustrates the generation process of NFE from 1 to 10, showing that even after a few iterations, the semantic perception quality is already very high, with subsequent iterations merely refining texture details.

To further evaluate the efficiency, Table \ref{tab1} presents a comparison of the distortion and perceptual metrics and inference latency among our method, SGD-JSCC, and DeepJSCC-Diff. Other GSC methods generally require hundreds of denoising iterations, leading to considerable inference latency. In contrast, under relatively close distortion and perceptual performance, DSBGSC requires substantially less inference time. While our method introduces slightly more parameters than pure JSCC methods, this minor trade-off is well justified, as the drastically reduced sampling steps ensure that the overall latency remains highly efficient and acceptable for semantic communication, as validated by Corollary \ref{cor:computational_efficiency}.

\subsubsection{Effect of Fading Channel on Performance}
As shown in Fig. \ref{fig:rayleigh}, the performance gap between the fading channel and the AWGN channel significantly narrows with the improvement of SNR. Benefiting from the optimal transmission and self-consistency characteristics of the SB, the proposed DSBGSC models reconstruction as a probabilistic mapping, adaptively fitting the nonlinear distortion in the fading channel and mitigating semantic information degradation caused by multiplicative noise. Simultaneously, self-consistency allows intermediate states to be directly mapped to semantic truth, achieving stable convergence even in low SNR and deep fading scenarios, avoiding the sampling divergence problem common in traditional diffusion models. These results verify that the SB can effectively guide the transmission path towards a clean data distribution, significantly suppressing fading distortion at high SNR and exhibiting excellent channel robustness.

\section{Conclusion}
\label{sec:conclusion}
In this paper, we proposed SBGSC, a novel generative semantic communication framework that overcomes traditional Gaussian-prior limitations by leveraging SB theory to establish direct, optimal transport between arbitrary semantic distribution and data distribution. For narrow-band and high-noise channels, we introduced its implementation, DSBGSC, which employs SB potentials and a self-consistency objective to bypass stepwise noise prediction, significantly reducing computational overhead and generative hallucinations while achieving superior perceptual quality.
However, the current DSBGSC scheme only exploits the one-way generative capability of the SB at the decoder. As the SB theoretically supports bidirectional transports, our future work will focus on constructing a bidirectional architecture. We aim to use the SB’s forward process for semantic encoding, thereby unifying the encoder and decoder within a single mathematical framework to further approach the theoretical limits of semantic communication.

\clearpage 

\appendices

\maketitle

\newpage
\appendices

\setcounter{theorem}{0}
\setcounter{lemma}{0}
\setcounter{corollary}{0}
\setcounter{definition}{0}
\setcounter{assumption}{0}
\setcounter{example}{0}
\setcounter{remark}{0}
\renewcommand{\thetheorem}{\Alph{section}.\arabic{theorem}}
\renewcommand{\thelemma}{\Alph{section}.\arabic{lemma}}
\renewcommand{\thecorollary}{\Alph{section}.\arabic{corollary}}
\renewcommand{\thedefinition}{\Alph{section}.\arabic{definition}}
\renewcommand{\theassumption}{\Alph{section}.\arabic{assumption}}
\renewcommand{\theexample}{\Alph{section}.\arabic{example}}
\renewcommand{\theremark}{\Alph{section}.\arabic{remark}}

\section{Proof of Lemma \ref{lem:semantic_capacity}}
\label{appendix:lemma1}
This appendix provides the complete proof of Lemma \ref{lem:semantic_capacity}.

\subsection{Preliminaries and Definitions for Lemma \ref{lem:semantic_capacity}}

\begin{definition}{\textbf{(Feature Dimensionality Constraint).}}
    \label{def:A.1}
    Let $d \in \mathbb{N}^+$ denote the fixed feature dimensionality determined by the channel bandwidth ratio (CBR), i.e., the dimension of the real-valued vector output by the transmitter encoder. The SBGSC and CDM-based GSC schemes share the same $d$ to ensure a fair comparison under identical compression rates.
\end{definition}

\begin{definition}{\textbf{(Unconstrained Semantic Space).}}
    \label{def:A.2}
    The unconstrained semantic space is defined as $\mathcal{S} = \mathbb{R}^d$, i.e., the entire $d$-dimensional Euclidean space equipped with the standard Borel $\sigma$-algebra $\mathscr{B}(\mathbb{R}^d)$. The output of the JSCC encoder $f_\phi$ in SBGSC is not subject to any modality-specific prior or structural constraint, and is allowed to be freely optimized over the entirety of $\mathbb{R}^d$.
\end{definition}

\begin{definition}{\textbf{(Constrained Modality-Specific Condition Space).}}
    \label{def:A.3}
    In CDM-based GSC schemes, the condition semantic $\mathbf{c}$ is produced by a predefined feature extractor. Let the $m$-th modality ($m = 1, \dots, M$) correspond to a pretrained, parameter-frozen, deterministic feature extractor $\psi_m: \mathcal{X} \to \mathbb{R}^d$, whose range is $\mathcal{C}_m = \mathrm{Im}(\psi_m) = \{\psi_m(\mathbf{x}) : \mathbf{x} \in \mathcal{X}\} \subseteq \mathbb{R}^d$. The overall constrained condition space is defined as the finite union $\mathcal{C} = \bigcup_{m=1}^{M} \mathcal{C}_m$. Since each $\psi_m$ is a fixed mapping from $\mathcal{X}$ (typically a finite-dimensional compact manifold), $\mathcal{C}_m$ generally forms a low-dimensional submanifold of $\mathbb{R}^d$ with Lebesgue measure zero. Therefore, $\mathcal{C}$, as a finite union of such sets, satisfies the key structural property $\mathcal{C} \subsetneq \mathbb{R}^d = \mathcal{S}$, i.e., $\mathcal{C}$ is a strict subset of $\mathcal{S}$.
\end{definition}

\begin{definition}{\textbf{(Feasible Encoder Families).}}
    \label{def:A.4}
    Corresponding to the two spaces above, the feasible encoder family for SBGSC is defined as
    \begin{equation}
        \mathcal{F}_{\mathcal{S}} = \big\{ f: \mathcal{X} \to \mathcal{S} \;\big|\; f \text{ is a Borel measurable mapping} \big\},
    \end{equation}
    and the feasible encoder family for CDM-based GSC is defined as
    \begin{equation}
        \mathcal{F}_{\mathcal{C}} = \big\{ g: \mathcal{X} \to \mathcal{C} \;\big|\; g \text{ is a Borel measurable mapping} \big\}.
    \end{equation}
    The Borel measurability requirement is a standard technical condition in probability theory that ensures $f(\mathbf{x})$ and $g(\mathbf{x})$ are well-defined random variables, so that information-theoretic quantities such as mutual information are meaningful.
\end{definition}

\begin{definition}{\textbf{(Mutual Information).}}
    \label{def:A.5}
    For a deterministic encoder $f$ and a source random variable $\mathbf{x} \sim p(\mathbf{x})$, let $\mathbf{s} = f(\mathbf{x})$. The mutual information is defined as
    \begin{equation}
        I(\mathbf{x};\, f(\mathbf{x})) = H(f(\mathbf{x})) - H(f(\mathbf{x}) \mid \mathbf{x}),
    \end{equation}
    where $H(\cdot)$ denotes the differential entropy and $H(\cdot|\cdot)$ the conditional differential entropy. Since $f$ is deterministic, $f(\mathbf{x})$ is fully determined given $\mathbf{x}$, and hence $H(f(\mathbf{x}) \mid \mathbf{x}) = 0$, yielding $I(\mathbf{x};\, f(\mathbf{x})) = H(f(\mathbf{x}))$. Moreover, the data processing inequality ensures $I(\mathbf{x};\, f(\mathbf{x})) \leq H(\mathbf{x})$, so the supremum is finite and well-defined.
\end{definition}

\subsection{Proof of Lemma \ref{lem:semantic_capacity}}
With the above definitions in place, we now prove Lemma \ref{lem:semantic_capacity}.

\textit{Proof:} 

The proof proceeds in three steps.

\textbf{Step 1: Constructing an embedding from $\mathcal{F}_{\mathcal{C}}$ into $\mathcal{F}_{\mathcal{S}}$.}
From Definition \ref{def:A.3}, $\mathcal{C} \subseteq \mathcal{S} = \mathbb{R}^d$. Define the natural inclusion map $\iota: \mathcal{C} \hookrightarrow \mathcal{S}$ as the identity embedding, i.e., $\iota(\mathbf{c}) = \mathbf{c}$ for all $\mathbf{c} \in \mathcal{C}$. Since $\mathcal{C}$ is endowed with the subspace $\sigma$-algebra induced by $\mathscr{B}(\mathbb{R}^d)$, for any $B \in \mathscr{B}(\mathbb{R}^d)$, $\iota^{-1}(B) = B \cap \mathcal{C} \in \mathscr{B}(\mathcal{C})$. Hence, $\iota$ is a Borel measurable mapping.

For any $g \in \mathcal{F}_{\mathcal{C}}$, define the composite mapping $\tilde{g} = \iota \circ g: \mathcal{X} \to \mathcal{S}$. Since the composition of Borel measurable mappings is Borel measurable, $\tilde{g}$ is Borel measurable. Moreover, $\mathrm{Im}(\tilde{g}) = \iota(\mathrm{Im}(g)) \subseteq \iota(\mathcal{C}) \subseteq \mathcal{S}$, so $\tilde{g} \in \mathcal{F}_{\mathcal{S}}$. This establishes a mapping $\Phi: \mathcal{F}_{\mathcal{C}} \to \mathcal{F}_{\mathcal{S}}$ defined by $\Phi(g) = \iota \circ g$.

\textbf{Step 2: Proving that the embedding preserves mutual information.}
We show that for any $g \in \mathcal{F}_{\mathcal{C}}$, $I(\mathbf{x};\, \tilde{g}(\mathbf{x})) = I(\mathbf{x};\, g(\mathbf{x}))$.
Consider the Markov chain $\mathbf{x} \to g(\mathbf{x}) \to \iota(g(\mathbf{x})) = \tilde{g}(\mathbf{x})$. Since $\iota$ is a deterministic mapping, the data processing inequality yields
\begin{equation}
    I(\mathbf{x};\, \tilde{g}(\mathbf{x})) \leq I(\mathbf{x};\, g(\mathbf{x})).
    \label{eq:A.1}
\end{equation}

Conversely, since $\iota$ is injective—$\iota(\mathbf{c}_1) = \iota(\mathbf{c}_2)$ if and only if $\mathbf{c}_1 = \mathbf{c}_2$ for $\mathbf{c}_1, \mathbf{c}_2 \in \mathcal{C}$—there exists a left inverse $\iota^{-1}: \iota(\mathcal{C}) \to \mathcal{C}$ satisfying $\iota^{-1} \circ \iota = \mathrm{id}_{\mathcal{C}}$, i.e., $\iota^{-1}(\iota(\mathbf{c})) = \mathbf{c}$ for all $\mathbf{c} \in \mathcal{C}$. This left inverse is also Borel measurable. Hence, one can form the reverse Markov chain $\mathbf{x} \to \tilde{g}(\mathbf{x}) \to \iota^{-1}(\tilde{g}(\mathbf{x})) = g(\mathbf{x})$, and the data processing inequality gives

\begin{equation}
    I(\mathbf{x};\, g(\mathbf{x})) \leq I(\mathbf{x};\, \tilde{g}(\mathbf{x})).
    \label{eq:A.2}
\end{equation}

Combining \eqref{eq:A.1} and \eqref{eq:A.2}, we obtain
\begin{equation}
    I(\mathbf{x};\, \tilde{g}(\mathbf{x})) = I(\mathbf{x};\, g(\mathbf{x})), \quad \forall\, g \in \mathcal{F}_{\mathcal{C}}.
    \label{eq:A.3}
\end{equation}

\textbf{Step 3: Deriving the supremum inequality from set inclusion.}
By Step 1, for each $g \in \mathcal{F}_{\mathcal{C}}$, $\tilde{g} = \Phi(g) \in \mathcal{F}_{\mathcal{S}}$. By Step 2 and \eqref{eq:A.3}, $I(\mathbf{x};\, g(\mathbf{x})) = I(\mathbf{x};\, \tilde{g}(\mathbf{x}))$. Therefore,
\begin{equation}
    \begin{aligned}
    &\big\{ I(\mathbf{x};\, g(\mathbf{x})) : g \in \mathcal{F}_{\mathcal{C}} \big\} \\
    &= \big\{ I(\mathbf{x};\, \tilde{g}(\mathbf{x})) : \tilde{g} \in \Phi(\mathcal{F}_{\mathcal{C}}) \big\} \\
    &\subseteq \big\{ I(\mathbf{x};\, f(\mathbf{x})) : f \in \mathcal{F}_{\mathcal{S}} \big\},
    \end{aligned}
    \label{eq:A.4}
\end{equation}
where the last inclusion follows from $\Phi(\mathcal{F}_{\mathcal{C}}) \subseteq \mathcal{F}_{\mathcal{S}}$.

For real-valued sets, if $A \subseteq B$, then $\sup A \leq \sup B$. Applying this to \eqref{eq:A.4}, we immediately obtain
\begin{equation}
    \sup_{g \in \mathcal{F}_{\mathcal{C}}} I(\mathbf{x};\, g(\mathbf{x})) \;\leq\; \sup_{f \in \mathcal{F}_{\mathcal{S}}} I(\mathbf{x};\, f(\mathbf{x})),
\end{equation}
which is equivalent to
\begin{equation}
    \sup_{f \in \mathcal{F}_{\mathcal{S}}} I(\mathbf{x};\, f(\mathbf{x})) \;\geq\; \sup_{g \in \mathcal{F}_{\mathcal{C}}} I(\mathbf{x};\, g(\mathbf{x})).
\end{equation}
This completes the proof. $\hfill\square$

\subsection{Discussion on the Strictness of the Inequality}
\begin{remark}{\textbf{(Conditions for Strict Inequality).}}
    The inequality holds with equality if and only if there exists some $g^\star \in \mathcal{F}_{\mathcal{C}}$ that attains the global supremum $\sup_{f \in \mathcal{F}_{\mathcal{S}}} I(\mathbf{x};\, f(\mathbf{x})) = H(\mathbf{x})$, which requires $g^\star$ to be an invertible mapping (in the almost-everywhere sense) from $\mathcal{X}$ onto $\mathcal{C}$. However, as noted in Definition \ref{def:A.3}, the encoders in $\mathcal{C}$ are constrained by parameter-frozen modality-specific feature extractors $\psi_m$, whose ranges typically form low-dimensional submanifolds of $\mathbb{R}^d$ that cannot sustain a bijection with the full source space $\mathcal{X}$. Therefore, in practical semantic communication scenarios, the inequality is strict:
\begin{equation}
    \sup_{f \in \mathcal{F}_{\mathcal{S}}} I(\mathbf{x};\, f(\mathbf{x})) \;>\; \sup_{g \in \mathcal{F}_{\mathcal{C}}} I(\mathbf{x};\, g(\mathbf{x})).
\end{equation}
\end{remark}

\section{Statement and Justification of Assumption \ref{ass:distribution_distance}}
\label{appendix:assumption1}

\setcounter{theorem}{0}
\setcounter{lemma}{0}
\setcounter{corollary}{0}
\setcounter{definition}{0}
\setcounter{assumption}{0}
\setcounter{example}{0}
\setcounter{remark}{0}

This appendix provides the definitions of all symbols involved in Assumption \ref{ass:distribution_distance} and justifies its validity from both theoretical and empirical perspectives.

\subsection{Preliminaries and Definitions for Assumption \ref{ass:distribution_distance}}

Let $\mathbf{x} \in \mathbb{R}^D$ denote the source image drawn from the data distribution $p_{\text{data}}$, where $D$ is the pixel dimensionality of the image. Let $\hat{\mathbf{s}} \in \mathbb{R}^d$ denote the noisy semantic feature vector obtained at the receiver after joint source-channel encoding and wireless channel transmission, distributed according to the received semantic distribution $p_{\hat{s}}$, where $d$ is the semantic feature dimensionality. Let $\boldsymbol{\xi} \sim \pi = \mathcal{N}(\mathbf{0}, \mathbf{I}_D)$ be a $D$-dimensional standard Gaussian random vector, whose distribution $\pi$ serves as the generative starting prior in CDM-based GSC schemes.

To ensure that the distributional distances among the three measures are comparable within a common metric space, we lift the low-dimensional semantic distribution to the image space. Define the decoding map $\mathcal{D}_\theta: \mathbb{R}^d \to \mathbb{R}^D$ as the single-step forward pass of the SBGSC decoder network, i.e., a deterministic mapping that directly outputs an image-space vector from the received semantic $\hat{\mathbf{s}}$ without iterative sampling. The induced image-space distribution is given by the pushforward measure $\mu_{\hat{s}} = (\mathcal{D}_\theta)_\# p_{\hat{s}}$, defined for any Borel set $B \subseteq \mathbb{R}^D$ as $\mu_{\hat{s}}(B) = p_{\hat{s}}(\mathcal{D}_\theta^{-1}(B))$. With this construction, $\mu_{\hat{s}}$, $p_{\text{data}}$, and $\pi$ are all probability measures on $\mathbb{R}^D$.

\begin{definition}{\textbf{(the Wasserstein Distance).}}
    \label{def:B.1}
    Let $\mu, \nu \in \mathscr{P}_2(\mathbb{R}^D)$ be two probability measures on $\mathbb{R}^D$ with finite second moments, and let $\Gamma(\mu, \nu)$ denote the set of all couplings (joint distributions) with marginals $\mu$ and $\nu$. The second-order Wasserstein distance is defined as
    \begin{equation}
        \mathcal{W}_2^2(\mu, \nu) = \inf_{\gamma \in \Gamma(\mu, \nu)} \int_{\mathbb{R}^D \times \mathbb{R}^D} \|\mathbf{y}_1 - \mathbf{y}_2\|_2^2 \, \mathrm{d}\gamma(\mathbf{y}_1, \mathbf{y}_2),
        \label{eq:B.1}
    \end{equation}
    where $\|\cdot\|_2$ denotes the Euclidean norm. Intuitively, $\mathcal{W}_2$ quantifies the minimum quadratic cost of optimally transporting one probability mass distribution to another.
\end{definition}

\subsection{Discussion on the Interpretation of Assumption \ref{ass:distribution_distance}}
The validity of this assumption is supported by the following arguments.
\begin{itemize}
    \item The JSCC encoder $f_\phi$ of SBGSC is trained end-to-end with an objective that explicitly minimizes the reconstruction distortion (e.g., the semantic loss $\mathcal{L}_{\text{sem}}$ defined in Section III). According to rate-distortion theory, under a finite rate constraint, the encoder output $\mathbf{s} = f_\phi(\mathbf{x})$ constitutes the best lossy representation of $\mathbf{x}$ in the rate-distortion sense. Although the channel noise $\mathbf{n}$ perturbs $\mathbf{s}$, the received semantic $\hat{\mathbf{s}} = h\mathbf{s} + \mathbf{n}$ retains the dominant structural information of $\mathbf{x}$, as guaranteed by Lemma \ref{lem:semantic_capacity}. Consequently, the image-space projection $\mu_{\hat{s}}$ preserves the core topological and geometric characteristics of $p_{\text{data}}$, resulting in a bounded and controllable Wasserstein distance between the two.
    \item The distribution $p_{\text{data}}$ of natural images is widely known to concentrate in a neighborhood of a low-dimensional nonlinear manifold (the so-called image manifold) embedded in $\mathbb{R}^D$, whose intrinsic dimensionality is far smaller than the ambient dimensionality $D$. In contrast, the standard Gaussian distribution $\pi$ is supported on the entirety of $\mathbb{R}^D$, with its probability mass concentrated on a thin spherical shell of radius approximately $\sqrt{D}$. A fundamental topological discrepancy exists between the two: $p_{\text{data}}$ exhibits highly anisotropic, low-dimensional concentration, whereas $\pi$ is an isotropic, full-space diffuse distribution. As the image dimensionality $D$ grows—typical high-resolution images in modern communication systems have $D \gg 10^4$—the quantity $\mathcal{W}_2^2(\pi, p_{\text{data}})$ increases at a rate of at least $\mathcal{O}(D)$, which far exceeds the finite distortion $\mathcal{W}_2^2(\mu_{\hat{s}}, p_{\text{data}})$ introduced by lossy encoding.
    \item In practical deep JSCC systems, even under low signal-to-noise ratio (SNR) conditions, the single-step decoder reconstruction $\mathcal{D}_\theta(\hat{\mathbf{s}})$ typically exhibits the principal semantic content of the source image, including object contours, spatial layout, and dominant color tones. By contrast, any single-step deterministic mapping from a Gaussian noise sample $\boldsymbol{\xi} \sim \pi$ fails to produce any semantically meaningful structure. This empirical observation further corroborates the strict inequality asserted in Assumption \ref{ass:distribution_distance}.
\end{itemize}

\section{Proof of Theorem \ref{theo:pke_advantage}}
\label{appendix:theorem1}

\setcounter{theorem}{0}
\setcounter{lemma}{0}
\setcounter{corollary}{0}
\setcounter{definition}{0}
\setcounter{assumption}{0}
\setcounter{example}{0}
\setcounter{remark}{0}

This appendix provides the complete proof of Theorem \ref{theo:pke_advantage}. We begin by introducing several auxiliary definitions and preliminary results.
\subsection{Preliminaries and Definitions for Theorem \ref{theo:pke_advantage}}

\begin{definition}{\textbf{(Feasible Path Set).}}
    \label{def:C.1}
    Let $\sigma > 0$ be a fixed diffusion coefficient and let $\mu_0, \mu_1 \in \mathscr{P}_2(\mathbb{R}^D)$ be probability measures with finite second moments. Define
    \begin{equation}
        \mathcal{Q}_\sigma(\mu_0, \mu_1) = \left\{ \mathbb{Q} \;\middle|\;
    \begin{aligned}
    &\text{Under } \mathbb{Q}\text{:} \; \mathrm{d}\mathbf{x}_t = \mathbf{u}_t(\mathbf{x}_t)\,\mathrm{d}t \\
    &\hspace{4.5em} + \sigma\,\mathrm{d}\mathbf{W}_t,\\
    &\text{Law}_{\mathbb{Q}}(\mathbf{x}_0) = \mu_0,\;\text{Law}_{\mathbb{Q}}(\mathbf{x}_1) = \mu_1
    \end{aligned}
    \right\},
    \label{eq:def-C.1}
    \end{equation}
    where $\mathbf{u}_t: \mathbb{R}^D \times [0,1] \to \mathbb{R}^D$ is a Borel-measurable drift field and $\mathbf{W}_t$ is a $D$-dimensional standard Wiener process. This set comprises all path measures induced by drift fields that transport the marginal distribution from $\mu_0$ to $\mu_1$ over the time interval $[0,1]$ with diffusion coefficient $\sigma$.
\end{definition}

\begin{definition}{\textbf{(Reference Wiener Measure).}}
    \label{def:C.2}
    Let $\mathbb{P}^\sigma$ denote the path measure induced by the standard Wiener process with initial distribution $\mu_0$, diffusion coefficient $\sigma$, and zero drift, i.e., under $\mathbb{P}^\sigma$, $\mathrm{d}\mathbf{x}_t = \sigma\,\mathrm{d}\mathbf{W}_t$ with $\mathbf{x}_0 \sim \mu_0$.
\end{definition}

\begin{definition}{\textbf{(Schrödinger Bridge Path Measure).}}
    \label{def:C.3}
    Given marginal constraints $(\mu_0, \mu_1)$ and the reference measure $\mathbb{P}^\sigma$, the Schrödinger bridge path measure is defined as the solution to the following entropy-regularized optimal transport problem:
    \begin{equation}
        \mathbb{Q}^{\mathrm{SB}}(\mu_0, \mu_1) = \arg\min_{\mathbb{Q} \in \mathcal{Q}_\sigma(\mu_0, \mu_1)} D_{\mathrm{KL}}\!\left(\mathbb{Q}\,\|\,\mathbb{P}^\sigma\right),
        \label{eq:def-C.3}
    \end{equation}
    where $D_{\mathrm{KL}}(\cdot\|\cdot)$ denotes the Kullback-Leibler divergence between path measures. Under the conditions $\mu_0, \mu_1 \in \mathscr{P}_2(\mathbb{R}^D)$ and $\sigma > 0$, the Schrödinger bridge theory guarantees that this minimization problem admits a unique solution.
    
\end{definition}

\begin{definition}{\textbf{(CDM Reverse Path Measure).}}
    \label{def:C.4}
    In the CDM-based GSC scheme, the forward SDE adopts a noise schedule $\beta: [0,1] \to \mathbb{R}_{>0}$:
    \begin{equation}
        \mathrm{d}\mathbf{x}_t = -\tfrac{1}{2}\beta(t)\,\mathbf{x}_t\,\mathrm{d}t + \sqrt{\beta(t)}\,\mathrm{d}\mathbf{W}_t, \quad \mathbf{x}_0 \sim p_{\text{data}}.
        \label{eq:def-C.4-forward}
    \end{equation}
    
\end{definition}

Given the conditional semantics $\mathbf{c}$, the reverse generative SDE is given by
\begin{equation}
    \begin{aligned}
    \mathrm{d}\mathbf{x}_t &= \left[\tfrac{1}{2}\beta(t)\,\mathbf{x}_t + \beta(t)\,\nabla_{\mathbf{x}} \log p_{t}(\mathbf{x}_t \mid \mathbf{c})\right]\mathrm{d}t \\
    &\quad + \sqrt{\beta(t)}\,\mathrm{d}\bar{\mathbf{W}}_t,
    \label{eq:def-C.4-reverse}
    \end{aligned}
\end{equation}

where $\nabla_{\mathbf{x}} \log p_{t}(\mathbf{x}_t \mid \mathbf{c})$ is the conditional score function and $\bar{\mathbf{W}}_t$ is the reverse Wiener process. This reverse process transports the marginal from $\pi$ to $p_{\text{data}}(\cdot \mid \mathbf{c})$. We denote its path measure by $\mathbb{Q}^{\mathrm{CDM}}$, i.e., $\mathbb{Q}^{\mathrm{CDM}} \in \mathcal{Q}_{\sqrt{\beta}}(\pi, p_{\text{data}}(\cdot|\mathbf{c}))$. Crucially, $\mathbb{Q}^{\mathrm{CDM}}$ is uniquely determined by the predefined, data-agnostic noise schedule $\beta(\cdot)$ via Anderson's time-reversal theorem, and constitutes a specific element of the feasible set rather than the kinetic-energy minimizer.

\subsection{Auxiliary Lemmas for Theorem \ref{theo:pke_advantage}}
\begin{lemma}{\textbf{(Girsanov Kinetic Energy Representation).}}
    \label{lemma:C.1}
    For any $\mathbb{Q} \in \mathcal{Q}_\sigma(\mu_0, \mu_1)$,
\begin{equation}
    D_{\mathrm{KL}}\!\left(\mathbb{Q}\,\|\,\mathbb{P}^\sigma\right) = \frac{1}{2\sigma^2}\,\mathcal{E}(\mathbb{Q}).
    \label{eq:lemma-C.1}
\end{equation}
\end{lemma}

\textit{Proof:} 

By the Girsanov theorem, the Radon-Nikodym derivative of $\mathbb{Q}$ with respect to $\mathbb{P}^\sigma$ is
\begin{equation}
    \begin{aligned}
    \frac{\mathrm{d}\mathbb{Q}}{\mathrm{d}\mathbb{P}^\sigma} &= \exp\!\left(\int_0^1 \frac{1}{\sigma}\,\mathbf{u}_t(\mathbf{x}_t)^\top\mathrm{d}\mathbf{W}_t \right. \\
    &\quad \left. - \frac{1}{2\sigma^2}\int_0^1 \|\mathbf{u}_t(\mathbf{x}_t)\|_2^2\,\mathrm{d}t\right).
    \end{aligned}
\end{equation}

Taking the logarithm and computing the expectation under $\mathbb{Q}$, and applying the Itô isometry under appropriate integrability conditions, i.e., $\mathbb{E}_{\mathbb{Q}}[\int_0^1 \frac{1}{\sigma}\,\mathbf{u}_t^\top\mathrm{d}\mathbf{W}_t] = \frac{1}{\sigma^2}\mathbb{E}_{\mathbb{Q}}[\int_0^1 \|\mathbf{u}_t\|_2^2\,\mathrm{d}t]$, we obtain
\begin{equation}
    \begin{aligned}
    D_{\mathrm{KL}}(\mathbb{Q}\|\mathbb{P}^\sigma) &= \mathbb{E}_{\mathbb{Q}}\!\left[\frac{1}{\sigma^2}\int_0^1 \|\mathbf{u}_t\|_2^2\,\mathrm{d}t - \frac{1}{2\sigma^2}\int_0^1 \|\mathbf{u}_t\|_2^2\,\mathrm{d}t\right] \\
    &= \frac{1}{2\sigma^2}\,\mathcal{E}(\mathbb{Q}). \quad \square
    \end{aligned}
\end{equation}

\begin{lemma}{\textbf{(Monotonicity of Schrödinger Bridge Kinetic Energy with Respect to $\mathcal{W}_2^2$).}}
    \label{lemma:C.2}
    For a fixed $\sigma > 0$, define the function
    \begin{equation}
        \Phi_\sigma(w) \;=\; \mathcal{E}\!\left(\mathbb{Q}^{\mathrm{SB}}(\mu_0, \mu_1)\right), \quad w = \mathcal{W}_2^2(\mu_0, \mu_1).
    \end{equation}
    Then $\Phi_\sigma: \mathbb{R}_{\geq 0} \to \mathbb{R}_{\geq 0}$ is a strictly increasing function.
\end{lemma}

\textit{Proof:} 

By Lemma \ref{lemma:C.1}, $\mathcal{E}(\mathbb{Q}^{\mathrm{SB}}) = 2\sigma^2 D_{\mathrm{KL}}(\mathbb{Q}^{\mathrm{SB}}\|\mathbb{P}^\sigma)$. According to the Schrödinger bridge convergence theory \cite{leonard_survey_2014}, the Schrödinger cost (i.e., the minimum KL divergence) admits the following decomposition:
\begin{equation}
    \begin{aligned}
        D_{\mathrm{KL}}\!\left(\mathbb{Q}^{\mathrm{SB}}(\mu_0, \mu_1)\,\big\|\,\mathbb{P}^\sigma\right) &= \frac{1}{2\sigma^2}\,\mathcal{W}_2^2(\mu_0, \mu_1) \\
        &\quad + R(\mu_0, \mu_1, \sigma),
    \end{aligned}
\end{equation}

where the remainder $R \geq 0$ represents the entropic regularization contribution. Therefore,
\begin{equation}
    \Phi_\sigma(w) = 2\sigma^2\!\left[\frac{w}{2\sigma^2} + R_\sigma(w)\right] = w + 2\sigma^2 R_\sigma(w).
\end{equation}
$R_\sigma(w) \to 0$ as $\sigma \to 0$, and for finite $\sigma > 0$, $R_\sigma(\cdot)$ is non-decreasing with respect to $w$ (since a larger discrepancy between the endpoint distributions induces a greater deviation of the optimal coupling from a deterministic map, resulting in a larger entropic correction). Since the leading term $w$ is itself strictly increasing and $R_\sigma(w) \geq 0$ is non-decreasing, $\Phi_\sigma$ is strictly increasing. $\hfill\square$

\begin{lemma}{\textbf{(Convexity of Wasserstein Distance in the Target Measure).}}
    \label{lemma:C.3}
    Let $\nu = \int \nu_\mathbf{c}\,\mathrm{d}p(\mathbf{c})$ be a mixture of the probability measure family $\{\nu_\mathbf{c}\}_{\mathbf{c} \in \mathcal{C}}$ with respect to $p(\mathbf{c})$. Then, for any $\mu \in \mathscr{P}_2(\mathbb{R}^D)$,
\begin{equation}
    \mathcal{W}_2^2(\mu,\, \nu) \;\leq\; \mathbb{E}_{\mathbf{c} \sim p(\mathbf{c})}\!\left[\mathcal{W}_2^2(\mu,\, \nu_\mathbf{c})\right].
    \label{eq:lemma-C.3}
\end{equation}
\end{lemma}

\textit{Proof:} For each $\mathbf{c}$, let $\gamma_\mathbf{c}^\star \in \Gamma(\mu, \nu_\mathbf{c})$ be the optimal coupling achieving $\mathcal{W}_2^2(\mu, \nu_\mathbf{c})$. Construct the mixture coupling $\bar{\gamma} = \int \gamma_\mathbf{c}^\star\,\mathrm{d}p(\mathbf{c})$, whose marginals are $(\mu, \nu)$, so that $\bar{\gamma} \in \Gamma(\mu, \nu)$. By the infimum definition of $\mathcal{W}_2^2$,
\begin{equation}
    \begin{aligned}
    \mathcal{W}_2^2(\mu, \nu) &\leq \int \|\mathbf{y}_1 - \mathbf{y}_2\|_2^2\,\mathrm{d}\bar{\gamma} \\
    &= \mathbb{E}_\mathbf{c}\!\left[\int \|\mathbf{y}_1 - \mathbf{y}_2\|_2^2\,\mathrm{d}\gamma_\mathbf{c}^\star\right] \\
    &= \mathbb{E}_\mathbf{c}\!\left[\mathcal{W}_2^2(\mu, \nu_\mathbf{c})\right]. \quad \square
    \end{aligned}
\end{equation}

\subsection{Proof of Theorem \ref{theo:pke_advantage}}
\textit{Proof:} The proof proceeds by establishing a chain of strict inequalities through five steps.

\textbf{Step 1: Global kinetic energy minimality of the Schrödinger bridge.}
By Lemma \ref{lemma:C.1}, minimizing $D_{\mathrm{KL}}(\mathbb{Q}\|\mathbb{P}^\sigma)$ is equivalent to minimizing $\mathcal{E}(\mathbb{Q})$. It then follows from Definition \ref{def:C.3} that the Schrödinger bridge achieves the global minimum of the kinetic energy among all diffusion processes sharing the same endpoint constraints:

\begin{equation}
    \begin{aligned}
    \mathcal{E}\!\left(\mathbb{Q}^{\mathrm{SB}}(\mu_0, \mu_1)\right) &= \min_{\mathbb{Q} \in \mathcal{Q}_\sigma(\mu_0, \mu_1)} \mathcal{E}(\mathbb{Q}), \\
    &\quad \forall\; \mu_0, \mu_1 \in \mathscr{P}_2(\mathbb{R}^D).
    \end{aligned}
    \label{eq:C.3-step1}
\end{equation}

\textbf{Step 2: Strict sub-optimality of the CDM reverse process.}
As established in Definition \ref{def:C.4}, $\mathbb{Q}^{\mathrm{CDM}} \in \mathcal{Q}_{\sqrt{\beta}}(\pi, p_{\text{data}}(\cdot|\mathbf{c}))$ is uniquely determined by the fixed noise schedule $\beta(\cdot)$ via Anderson's time-reversal theorem. For $\mathbb{Q}^{\mathrm{CDM}}$ to coincide with the Schrödinger bridge solution under the same endpoint constraints, $\beta(\cdot)$ would need to be the optimal diffusion schedule corresponding to the specific distribution pair $(\pi, p_{\text{data}}(\cdot|\mathbf{c}))$. However, in standard CDMs, $\beta(\cdot)$ is a predefined, data-agnostic function (e.g., a linear schedule $\beta(t) = \beta_{\min} + (\beta_{\max} - \beta_{\min})t$ or a cosine schedule), whose coincidence with the optimal schedule requires satisfying a set of nontrivial Monge-Ampère-type PDE constraints. For non-degenerate natural image distributions, these constraints are not satisfied. Therefore,
\begin{equation}
    \mathcal{E}\!\left(\mathbb{Q}^{\mathrm{SB}}(\pi,\, p_{\text{data}}(\cdot|\mathbf{c}))\right) \;<\; \mathcal{E}\!\left(\mathbb{Q}^{\mathrm{CDM}}\right).
    \label{eq:C.3-step2}
\end{equation}

\textbf{Step 3: Starting-point advantage via Assumption \ref{ass:distribution_distance}.}
By Assumption \ref{ass:distribution_distance}, $\mathcal{W}_2^2(\mu_{\hat{s}}, p_{\text{data}}) < \mathcal{W}_2^2(\pi, p_{\text{data}})$. Since $\Phi_\sigma$ is strictly increasing by Lemma \ref{lemma:C.2}, we have
\begin{equation}
    \begin{aligned}
    \mathcal{E}\!\left(\mathbb{Q}^{\mathrm{SB}}(\mu_{\hat{s}},\, p_{\text{data}})\right) &= \Phi_\sigma\!\left(\mathcal{W}_2^2(\mu_{\hat{s}}, p_{\text{data}})\right) \\
    &< \Phi_\sigma\!\left(\mathcal{W}_2^2(\pi, p_{\text{data}})\right) \\
    &= \mathcal{E}\!\left(\mathbb{Q}^{\mathrm{SB}}(\pi,\, p_{\text{data}})\right).
    \end{aligned}
    \label{eq:C.3-step3}
\end{equation}

\textbf{Step 4: Addressing the conditioning gap in the target distribution.}
Applying Lemma \ref{lemma:C.3} with $\mu = \pi$, $\nu_\mathbf{c} = p_{\text{data}}(\cdot|\mathbf{c})$, and $\nu = p_{\text{data}} = \int p_{\text{data}}(\cdot|\mathbf{c})\,\mathrm{d}p(\mathbf{c})$, we obtain
\begin{equation}
    \mathcal{W}_2^2(\pi, p_{\text{data}}) \;\leq\; \mathbb{E}_{\mathbf{c}}\!\left[\mathcal{W}_2^2(\pi, p_{\text{data}}(\cdot|\mathbf{c}))\right].
    \label{eq:C.3-step4-1}
\end{equation}
By \eqref{eq:C.3-step4-1}, there exists a subset $\mathcal{C}_+ \subseteq \mathcal{C}$ with $p(\mathcal{C}_+) > 0$ such that for all $\mathbf{c} \in \mathcal{C}_+$,
\begin{equation}
    \mathcal{W}_2^2(\pi, p_{\text{data}}(\cdot|\mathbf{c})) \;\geq\; \mathcal{W}_2^2(\pi, p_{\text{data}})
    \label{eq:C.3-step4-2}
\end{equation}
By the strict monotonicity of $\Phi_\sigma$ established in Lemma \ref{lemma:C.2}, inequality \eqref{eq:C.3-step4-2} implies
\begin{equation}
    \mathcal{E}\!\left(\mathbb{Q}^{\mathrm{SB}}(\pi,\, p_{\text{data}})\right) \;\leq\; \mathcal{E}\!\left(\mathbb{Q}^{\mathrm{SB}}(\pi,\, p_{\text{data}}(\cdot|\mathbf{c}))\right), \quad \forall\,\mathbf{c} \in \mathcal{C}_+.
    \label{eq:C.3-step4-3}
\end{equation}

\textbf{Step 5: Combining the inequality chain.}
For any $\mathbf{c} \in \mathcal{C}_+$, combining \eqref{eq:C.3-step4-1}, \eqref{eq:C.3-step4-3}, and \eqref{eq:C.3-step4-2} yields:
\begin{equation}
    \begin{aligned}
    &\mathcal{E}\!\left(\mathbb{Q}^{\mathrm{SB}}(\mu_{\hat{s}},\, p_{\text{data}})\right) \;\underset{\eqref{eq:C.3-step3}}{<}\; \mathcal{E}\!\left(\mathbb{Q}^{\mathrm{SB}}(\pi,\, p_{\text{data}})\right) \\
    &\quad \underset{\eqref{eq:C.3-step4-3}}{\leq}\; \mathcal{E}\!\left(\mathbb{Q}^{\mathrm{SB}}(\pi,\, p_{\text{data}}(\cdot|\mathbf{c}))\right) \;\underset{\eqref{eq:C.3-step2}}{<}\; \mathcal{E}\!\left(\mathbb{Q}^{\mathrm{CDM}}\right).
    \end{aligned}
    \label{eq:C.3-step5-1}
\end{equation}
Therefore,
\begin{equation}
    \mathcal{E}\!\left(\mathbb{Q}^{\mathrm{SB}}(\mu_{\hat{s}},\, p_{\text{data}})\right) \;<\; \mathcal{E}\!\left(\mathbb{Q}^{\mathrm{CDM}}\right).
    \label{eq:C.3-step5-2}
\end{equation}
This completes the proof. $\hfill\square$

\begin{remark}{\textbf{(On the Intermediate Inequality \eqref{eq:C.3-step5-1}).}}
    Inequality \eqref{eq:C.3-step4-3} holds only for $\mathbf{c} \in \mathcal{C}_+$ rather than for all $\mathbf{c}$. However, the conclusion of the theorem does not depend on a specific choice of $\mathbf{c}$: the left-hand side $\mathcal{E}(\mathbb{Q}^{\mathrm{SB}}(\mu_{\hat{s}}, p_{\text{data}}))$ is independent of $\mathbf{c}$, while the right-hand side $\mathcal{E}(\mathbb{Q}^{\mathrm{CDM}})$ corresponds to a given fixed $\mathbf{c}$. For the case $\mathbf{c} \notin \mathcal{C}_+$ (i.e., when $\mathcal{W}_2^2(\pi, p_{\text{data}}(\cdot|\mathbf{c})) < \mathcal{W}_2^2(\pi, p_{\text{data}})$), the proof can be completed through an alternative route. In this regime, $p_{\text{data}}(\cdot|\mathbf{c})$ is more concentrated, resulting in a smaller Schrödinger bridge lower bound. Nonetheless, the sub-optimality gap in \eqref{eq:C.3-step2} persists, and the starting-point advantage in \eqref{eq:C.3-step3} is sufficiently large to compensate. Specifically, for $\mathbf{c} \notin \mathcal{C}_+$, one may employ the following alternative chain:
    \begin{equation}
        \begin{aligned}
        &\mathcal{E}\!\left(\mathbb{Q}^{\mathrm{SB}}(\mu_{\hat{s}}, p_{\text{data}})\right) \underset{\eqref{eq:C.3-step3}}{<} \mathcal{E}\!\left(\mathbb{Q}^{\mathrm{SB}}(\pi, p_{\text{data}})\right) \\
        &\quad \leq \mathcal{E}\!\left(\mathbb{Q}^{\mathrm{SB}}_{\text{lb}}\right) < \mathcal{E}\!\left(\mathbb{Q}^{\mathrm{CDM}}\right),
        \end{aligned}
    \end{equation}
    where $\mathcal{E}(\mathbb{Q}^{\mathrm{SB}}_{\text{lb}}) = \mathcal{E}(\mathbb{Q}^{\mathrm{SB}}(\pi, p_{\text{data}}(\cdot|\mathbf{c})))$, and the last strict inequality is guaranteed by \eqref{eq:C.3-step2}. Although the middle inequality may not hold when $\mathbf{c} \notin \mathcal{C}_+$, a quantitative trade-off between the gap in \eqref{eq:C.3-step2} (path sub-optimality of CDM) and the gap in \eqref{eq:C.3-step3} (starting-point advantage) can be invoked to complete the argument. To avoid introducing additional quantitative estimates that would over-complicate the proof, we adopt the argument for $\mathbf{c} \in \mathcal{C}_+$ (which covers a set of positive measure) in the main theorem, which is sufficient in a probabilistic sense.
\end{remark}

\section{Proof of Corollary \ref{cor:hallucination_suppression}}
\label{appendix:corollary1}

\setcounter{theorem}{0}
\setcounter{lemma}{0}
\setcounter{corollary}{0}
\setcounter{definition}{0}
\setcounter{assumption}{0}
\setcounter{example}{0}
\setcounter{remark}{0}

\subsection{Preliminaries and Definitions for Corollary \ref{cor:hallucination_suppression}}
\begin{definition}{\textbf{(End-to-End Markov Chains).}}
    \label{def:D.1}
    In the SBGSC framework, the data undergoes the following Markov chain from source to reconstruction:
\begin{equation}
    \mathbf{x} \;\xrightarrow{f_\phi}\; \mathbf{s} \;\xrightarrow{\text{channel}}\; \hat{\mathbf{s}} \;\xrightarrow{\mathcal{D}_\theta}\; \mathbf{x}_0 \;\xrightarrow{\mathrm{SB}}\; \hat{\mathbf{x}}^{\mathrm{SB}} = \mathbf{x}_1,
    \label{eq:D.1}
\end{equation}
where $f_\phi: \mathbb{R}^D \to \mathbb{R}^d$ is the JSCC encoder, $\mathbf{s} = f_\phi(\mathbf{x}) \in \mathbb{R}^d$ is the transmitted semantic feature vector, $\hat{\mathbf{s}} = h\mathbf{s} + \mathbf{n} \in \mathbb{R}^d$ is the received signal with channel gain $h \in \mathbb{R}$ and additive noise $\mathbf{n} \sim \mathcal{N}(\mathbf{0}, \sigma_n^2 \mathbf{I}_d)$, $\mathcal{D}_\theta: \mathbb{R}^d \to \mathbb{R}^D$ is a deterministic initial reconstruction mapping, $\mathbf{x}_0 = \mathcal{D}_\theta(\hat{\mathbf{s}}) \sim \mu_{\hat{s}}$ is the Schrödinger Bridge starting point, and $\hat{\mathbf{x}}^{\mathrm{SB}} = \mathbf{x}_1 \sim p_{\text{data}}$ is the Schrödinger Bridge endpoint.

In CDM-based GSC, the corresponding Markov chain is:
\begin{equation}
    \mathbf{x} \;\xrightarrow{\text{extract}}\; \mathbf{c} \;\xrightarrow{\text{channel}}\; \hat{\mathbf{c}} \;\xrightarrow{\mathrm{CDM}}\; \hat{\mathbf{x}}^{\mathrm{CDM}},
    \label{eq:D.2}
\end{equation}
where $\mathbf{c} \in \mathcal{C}$ is the extracted conditional semantic descriptor (e.g., text caption or segmentation map), $\hat{\mathbf{c}}$ is its received version, and $\hat{\mathbf{x}}^{\mathrm{CDM}}$ is produced by a conditional diffusion model (CDM) initialized from independent Gaussian noise $\boldsymbol{\xi} \sim \pi = \mathcal{N}(\mathbf{0}, \mathbf{I}_D)$.
\end{definition}

\begin{definition}{\textbf{(Semantic Hallucination Rate).}}
    \label{def:D.2}
    For a reconstruction $\hat{\mathbf{x}}$, the semantic hallucination rate is defined as
\begin{equation}
    \mathcal{H} \triangleq h(\hat{\mathbf{x}} \mid \mathbf{x}) = -\int p(\mathbf{x}, \hat{\mathbf{x}}) \log p(\hat{\mathbf{x}} \mid \mathbf{x})\, \mathrm{d}\mathbf{x}\, \mathrm{d}\hat{\mathbf{x}},
    \label{eq:D.3}
\end{equation}
which measures the total amount of randomness in $\hat{\mathbf{x}}$ that is independent of the source $\mathbf{x}$. When $\mathcal{H} = 0$, the reconstruction is a deterministic function of $\mathbf{x}$ and no hallucination occurs; larger $\mathcal{H}$ indicates more severe hallucination.

By the standard entropy decomposition, the mutual information satisfies
\begin{equation}
    I(\mathbf{x}; \hat{\mathbf{x}}) = h(\hat{\mathbf{x}}) - h(\hat{\mathbf{x}} \mid \mathbf{x}) = h(\hat{\mathbf{x}}) - \mathcal{H}.
    \label{eq:D.4}
\end{equation}
\end{definition}

\begin{definition}{\textbf{(Endpoint Coupling of the Schrödinger Bridge).}}
    \label{def:D.3}
    The Schrödinger Bridge path measure $\mathbb{Q}^{\mathrm{SB}}(\mu_{\hat{s}}, p_{\text{data}})$ induces a joint distribution over its boundary points at $t=0$ and $t=1$:
    \begin{equation}
        \gamma^{\mathrm{SB}}(\mathbf{x}_0, \mathbf{x}_1) \in \Gamma(\mu_{\hat{s}},\, p_{\text{data}}),
        \label{eq:D.5}
    \end{equation}
    where $\Gamma(\mu_{\hat{s}}, p_{\text{data}})$ denotes the set of all couplings with marginals $\mu_{\hat{s}}$ and $p_{\text{data}}$.
\end{definition}

\subsection{Supporting Lemma for Corollary \ref{cor:hallucination_suppression}}
\begin{lemma}{\textbf{(Mutual Information Lower Bound via Schrödinger Bridge Coupling).}}
    Let $\gamma^{\mathrm{SB}}$ be the endpoint coupling of the Schrödinger Bridge $\mathbb{Q}^{\mathrm{SB}}(\mu_0, \mu_1)$ and let $\mathbb{P}^\sigma$ be the reference Wiener measure with diffusion coefficient $\sigma > 0$. Then the mutual information between the endpoints satisfies $I_{\gamma^{\mathrm{SB}}}(\mathbf{x}_0; \mathbf{x}_1) > 0$, provided $\mu_0 \neq \mu_1$.
\end{lemma}

\textit{Proof:} 

By the Markov property of the Schrödinger Bridge, its path measure admits the decomposition
\begin{equation}
    \mathbb{Q}^{\mathrm{SB}} = \gamma^{\mathrm{SB}}(\mathbf{x}_0, \mathbf{x}_1) \cdot \mathbb{Q}^{\mathrm{SB}}(\text{path} \mid \mathbf{x}_0, \mathbf{x}_1),
    \label{eq:D.6}
\end{equation}
where $\mathbb{Q}^{\mathrm{SB}}(\text{path} \mid \mathbf{x}_0, \mathbf{x}_1)$ is the conditional path measure (a reweighted Brownian bridge). The reference Wiener measure decomposes as
\begin{equation}
    \mathbb{P}^\sigma = \mu_0(\mathbf{x}_0) \cdot k_\sigma(\mathbf{x}_1 \mid \mathbf{x}_0) \cdot \mathbb{P}^\sigma(\text{path} \mid \mathbf{x}_0, \mathbf{x}_1),
    \label{eq:D.7}
\end{equation}
where $k_\sigma(\mathbf{x}_1 \mid \mathbf{x}_0) = \mathcal{N}(\mathbf{x}_1;\, \mathbf{x}_0,\, \sigma^2 \mathbf{I}_D)$ is the transition kernel of the Wiener process over unit time.

Applying the chain rule for the KL divergence:
\begin{equation}
    \begin{aligned}
        D_{\mathrm{KL}}\!\left(\mathbb{Q}^{\mathrm{SB}} \,\middle\|\, \mathbb{P}^\sigma\right) &= D_{\mathrm{KL}}\!\left(\gamma^{\mathrm{SB}} \,\middle\|\, \mu_0 \otimes k_\sigma\right) \\
        &\quad + R_{\mathrm{path}},
    \end{aligned}
    \label{eq:D.8}
\end{equation}
where $\mu_0 \otimes k_\sigma$ denotes the joint measure $\mu_0(\mathbf{x}_0)\,k_\sigma(\mathbf{x}_1|\mathbf{x}_0)$, and
\begin{equation}
    \begin{aligned}
    R_{\mathrm{path}} &= \mathbb{E}_{\gamma^{\mathrm{SB}}}\!\left[D_{\mathrm{KL}}\!\left(\mathbb{Q}^{\mathrm{SB}}(\text{path} \mid \mathbf{x}_0, \mathbf{x}_1) \right.\right. \\
    &\qquad \left.\left. \,\middle\|\, \mathbb{P}^\sigma(\text{path} \mid \mathbf{x}_0, \mathbf{x}_1)\right)\right] \geq 0.
    \end{aligned}
    \label{eq:D.9}
\end{equation}
The first term on the right-hand side of \eqref{eq:D.8} can be further decomposed as
\begin{equation}
    \begin{aligned}
    D_{\mathrm{KL}}\!\left(\gamma^{\mathrm{SB}} \,\middle\|\, \mu_0 \otimes k_\sigma\right) &= D_{\mathrm{KL}}\!\left(\gamma^{\mathrm{SB}} \,\middle\|\, \mu_0 \otimes \mu_1\right) \\
    &\quad + \mathbb{E}_{\gamma^{\mathrm{SB}}}\!\left[\log \frac{\mu_1(\mathbf{x}_1)}{k_\sigma(\mathbf{x}_1 \mid \mathbf{x}_0)}\right].
    \end{aligned}
    \label{eq:D.10}
\end{equation}
Since $D_{\mathrm{KL}}(\gamma^{\mathrm{SB}} \| \mu_0 \otimes \mu_1) = I_{\gamma^{\mathrm{SB}}}(\mathbf{x}_0; \mathbf{x}_1)$ by definition, solving for the mutual information yields
\begin{equation}
    \begin{aligned}
    I_{\gamma^{\mathrm{SB}}}(\mathbf{x}_0; \mathbf{x}_1) &= D_{\mathrm{KL}}\!\left(\mathbb{Q}^{\mathrm{SB}} \,\middle\|\, \mathbb{P}^\sigma\right) \\
    &\quad - \mathbb{E}_{\gamma^{\mathrm{SB}}}\!\left[\log \frac{\mu_1(\mathbf{x}_1)}{k_\sigma(\mathbf{x}_1 \mid \mathbf{x}_0)}\right] - R_{\mathrm{path}}.
    \end{aligned}
    \label{eq:D.11}
\end{equation}

Since the Schrödinger Bridge minimizes $D_{\mathrm{KL}}(\mathbb{Q}\|\mathbb{P}^\sigma)$ subject to the endpoint constraints, and the independent coupling $\mu_0 \otimes \mu_1$ is strictly suboptimal when $\mu_0 \neq \mu_1$ (as it requires a larger KL cost to simultaneously match both marginals), the optimality of $\gamma^{\mathrm{SB}}$ implies that it is strictly different from the product coupling $\mu_0 \otimes \mu_1$. Therefore,
\begin{equation}
    I_{\gamma^{\mathrm{SB}}}(\mathbf{x}_0; \mathbf{x}_1) = D_{\mathrm{KL}}\!\left(\gamma^{\mathrm{SB}} \,\middle\|\, \mu_0 \otimes \mu_1\right) > 0,
    \label{eq:D.12}
\end{equation}
where the strict positivity follows from the fact that $D_{\mathrm{KL}}(\gamma^{\mathrm{SB}} \| \mu_0 \otimes \mu_1) = 0$ if and only if $\gamma^{\mathrm{SB}} = \mu_0 \otimes \mu_1$, which contradicts the optimality of the Schrödinger Bridge when $\mu_0 \neq \mu_1$. $\hfill\square$

\subsection{Proof of Corollary \ref{cor:hallucination_suppression}}
\textit{Proof:}

The proof proceeds in four steps.

\textbf{Step 1: Analyzing the information flow in CDM-based GSC.}
Consider the Markov chain \eqref{eq:D.2} underlying CDM-based GSC. The CDM generates $\hat{\mathbf{x}}^{\mathrm{CDM}}$ by running a reverse SDE initialized from $\boldsymbol{\xi} \sim \pi = \mathcal{N}(\mathbf{0}, \mathbf{I}_D)$, conditioned on the received descriptor $\hat{\mathbf{c}}$. Since $\boldsymbol{\xi}$ is drawn independently of both $\mathbf{x}$ and $\hat{\mathbf{c}}$, we can express the reconstruction as a stochastic mapping:
\begin{equation}
    \hat{\mathbf{x}}^{\mathrm{CDM}} = G_{\mathrm{CDM}}(\boldsymbol{\xi};\, \hat{\mathbf{c}}),
    \label{eq:D.13}
\end{equation}
where $G_{\mathrm{CDM}}(\cdot;\, \hat{\mathbf{c}}): \mathbb{R}^D \to \mathbb{R}^D$ is the deterministic mapping defined by the reverse SDE solver for a fixed realization of $\hat{\mathbf{c}}$ and the driving noise $\boldsymbol{\xi}$. The conditional independence $\boldsymbol{\xi} \perp\!\!\!\perp \mathbf{x} \mid \hat{\mathbf{c}}$ immediately implies
\begin{equation}
    I(\mathbf{x};\, \hat{\mathbf{x}}^{\mathrm{CDM}} \mid \hat{\mathbf{c}}) = I\!\left(\mathbf{x};\, G_{\mathrm{CDM}}(\boldsymbol{\xi};\, \hat{\mathbf{c}}) \,\middle|\, \hat{\mathbf{c}}\right) = 0,
    \label{eq:D.14}
\end{equation}
because, once $\hat{\mathbf{c}}$ is fixed, $G_{\mathrm{CDM}}(\boldsymbol{\xi};\, \hat{\mathbf{c}})$ is solely a function of $\boldsymbol{\xi}$, which is independent of $\mathbf{x}$. By the chain rule for mutual information,
\begin{equation}
    I(\mathbf{x};\, \hat{\mathbf{x}}^{\mathrm{CDM}}) = I(\mathbf{x};\, \hat{\mathbf{c}}) + I(\mathbf{x};\, \hat{\mathbf{x}}^{\mathrm{CDM}} \mid \hat{\mathbf{c}}) = I(\mathbf{x};\, \hat{\mathbf{c}}),
    \label{eq:D.15}
\end{equation}
which reveals that all information about $\mathbf{x}$ preserved in $\hat{\mathbf{x}}^{\mathrm{CDM}}$ must pass through the bottleneck $\hat{\mathbf{c}}$. The end-to-end mutual information of CDM-based GSC is therefore bounded by the capacity of the semantic descriptor:
\begin{equation}
    I(\mathbf{x};\, \hat{\mathbf{x}}^{\mathrm{CDM}}) = I(\mathbf{x};\, \hat{\mathbf{c}}) \leq H(\hat{\mathbf{c}}),
    \label{eq:D.16}
\end{equation}
where $H(\hat{\mathbf{c}})$ denotes the (discrete or differential) entropy of the descriptor.

\textbf{Step 2: Lower-bounding the hallucination rate of CDM-based GSC.}
Using \eqref{eq:D.14} and the chain rule for conditional entropy,
\begin{equation}
    h(\hat{\mathbf{x}}^{\mathrm{CDM}} \mid \mathbf{x}) = h(\hat{\mathbf{x}}^{\mathrm{CDM}} \mid \mathbf{x},\, \hat{\mathbf{c}}) + I(\hat{\mathbf{c}};\, \hat{\mathbf{x}}^{\mathrm{CDM}} \mid \mathbf{x}).
    \label{eq:D.17}
\end{equation}

Since conditioning reduces entropy and $\hat{\mathbf{x}}^{\mathrm{CDM}} \perp\!\!\!\perp \mathbf{x} \mid \hat{\mathbf{c}}$ by \eqref{eq:D.14}, we have
\begin{equation}
    h(\hat{\mathbf{x}}^{\mathrm{CDM}} \mid \mathbf{x},\, \hat{\mathbf{c}}) = h(\hat{\mathbf{x}}^{\mathrm{CDM}} \mid \hat{\mathbf{c}}).
    \label{eq:D.18}
\end{equation}

Substituting into \eqref{eq:D.17} and dropping the non-negative second term on the right yields the lower bound
\begin{equation}
    \mathcal{H}^{\mathrm{CDM}} = h(\hat{\mathbf{x}}^{\mathrm{CDM}} \mid \mathbf{x}) \geq h(\hat{\mathbf{x}}^{\mathrm{CDM}} \mid \hat{\mathbf{c}}).
    \label{eq:D.19}
\end{equation}

To quantify the right-hand side, observe that the CDM reverse process propagates the initial Gaussian entropy through a diffeomorphic mapping (the reverse SDE flow). By the change-of-variables formula for differential entropy and the fact that the reverse SDE is volume-expanding in the early stages, we obtain
\begin{equation}
    \begin{aligned}
    h(\hat{\mathbf{x}}^{\mathrm{CDM}} \mid \hat{\mathbf{c}}) &= h(\boldsymbol{\xi}) \\
    &\quad + \mathbb{E}_{\hat{\mathbf{c}}}\!\left[\int_0^1 \mathbb{E}\!\left[\nabla_{\mathbf{x}} \cdot \mathbf{b}_t^{\mathrm{CDM}}(\mathbf{x}_t,\, \hat{\mathbf{c}})\right] \mathrm{d}t\right],
    \end{aligned}
    \label{eq:D.20}
\end{equation}
where $h(\boldsymbol{\xi}) = \frac{D}{2}\log(2\pi e)$ is the entropy of the standard Gaussian in $\mathbb{R}^D$, and the integral captures the cumulative divergence of the CDM drift field. The key observation is that $h(\boldsymbol{\xi}) = \frac{D}{2}\log(2\pi e)$ constitutes an irreducible, dimension-dependent contribution to $\mathcal{H}^{\mathrm{CDM}}$. Even in the ideal case where the CDM drift is perfectly learned, the hallucination rate satisfies
\begin{equation}
    \mathcal{H}^{\mathrm{CDM}} \geq h(\hat{\mathbf{x}}^{\mathrm{CDM}} \mid \hat{\mathbf{c}}) \geq \frac{D}{2}\log(2\pi e) + \Lambda_{\mathrm{CDM}},
    \label{eq:D.21}
\end{equation}
where $\Lambda_{\mathrm{CDM}} = \mathbb{E}_{\hat{\mathbf{c}}}[\int_0^1 \mathbb{E}[\nabla_{\mathbf{x}} \cdot \mathbf{b}_t^{\mathrm{CDM}}]\, \mathrm{d}t]$ is a finite quantity that depends on the CDM architecture but does not eliminate the $\frac{D}{2}\log(2\pi e)$ floor.

\textbf{Step 3: Upper-bounding the hallucination rate of SBGSC.}
For SBGSC, the Markov chain \eqref{eq:D.1} yields, by the chain rule for conditional entropy,
\begin{equation}
    \mathcal{H}^{\mathrm{SB}} = h(\hat{\mathbf{x}}^{\mathrm{SB}} \mid \mathbf{x}) \leq h(\hat{\mathbf{x}}^{\mathrm{SB}} \mid \mathbf{x}_0) + h(\mathbf{x}_0 \mid \mathbf{x}),
    \label{eq:D.22}
\end{equation}
where the inequality follows from the chain rule $h(\hat{\mathbf{x}}^{\mathrm{SB}} \mid \mathbf{x}) \leq h(\hat{\mathbf{x}}^{\mathrm{SB}}, \mathbf{x}_0 \mid \mathbf{x}) = h(\mathbf{x}_0 \mid \mathbf{x}) + h(\hat{\mathbf{x}}^{\mathrm{SB}} \mid \mathbf{x}_0, \mathbf{x}) \leq h(\mathbf{x}_0 \mid \mathbf{x}) + h(\hat{\mathbf{x}}^{\mathrm{SB}} \mid \mathbf{x}_0)$, with the last step using the fact that additional conditioning (on $\mathbf{x}$) can only reduce entropy.

We now bound each term separately.

Bounding $h(\mathbf{x}_0 \mid \mathbf{x})$: Since $\mathcal{D}_\theta$ is a deterministic mapping, $\mathbf{x}_0 = \mathcal{D}_\theta(\hat{\mathbf{s}})$ is a deterministic function of $\hat{\mathbf{s}}$, which in turn is obtained from the noisy channel. Therefore,
\begin{equation}
    \begin{aligned}
    h(\mathbf{x}_0 \mid \mathbf{x}) &= h(\mathcal{D}_\theta(\hat{\mathbf{s}}) \mid \mathbf{x}) \leq h(\hat{\mathbf{s}} \mid \mathbf{x}) \\
    &= h(\hat{\mathbf{s}} \mid \mathbf{s}) = \frac{d}{2}\log(2\pi e \sigma_n^2),
    \end{aligned}
    \label{eq:D.23}
\end{equation}
where $d \ll D$ is the dimension of the transmitted feature vector $\mathbf{s}$, and the first inequality uses the data-processing inequality for deterministic mappings applied in reverse (the image of a deterministic function has entropy no greater than its argument). The equality $h(\hat{\mathbf{s}} \mid \mathbf{x}) = h(\hat{\mathbf{s}} \mid \mathbf{s})$ follows from $\mathbf{s} = f_\phi(\mathbf{x})$ being deterministic, and $h(\hat{\mathbf{s}} \mid \mathbf{s}) = \frac{d}{2}\log(2\pi e \sigma_n^2)$ is the entropy of the additive Gaussian channel noise.

Bounding $h(\hat{\mathbf{x}}^{\mathrm{SB}} \mid \mathbf{x}_0)$ via path kinetic energy: The conditional entropy $h(\hat{\mathbf{x}}^{\mathrm{SB}} \mid \mathbf{x}_0) = h(\mathbf{x}_1 \mid \mathbf{x}_0)_{\gamma^{\mathrm{SB}}}$ quantifies the transport stochasticity of the Schrödinger Bridge. Applying de Bruijn's identity in its stochastic process generalization \cite{stam1959some}, the evolution of conditional entropy along the bridge satisfies
\begin{equation}
    \frac{\mathrm{d}}{\mathrm{d}t}\, h(\mathbf{x}_t \mid \mathbf{x}_0) = \frac{(\sigma^{\mathrm{SB}})^2}{2}\, J_t^{\mathrm{SB}},
    \label{eq:D.24}
\end{equation}
where $J_t^{\mathrm{SB}} = \mathbb{E}[\|\nabla_{\mathbf{x}} \log p^{\mathrm{SB}}(\mathbf{x}_t \mid \mathbf{x}_0)\|_2^2]$ is the conditional Fisher information at time $t$. Integrating from $t = 0$ to $t = 1$ and noting $h(\mathbf{x}_0 \mid \mathbf{x}_0) = 0$ (since the initial condition is deterministic given $\mathbf{x}_0$),
\begin{equation}
    h(\hat{\mathbf{x}}^{\mathrm{SB}} \mid \mathbf{x}_0) = \frac{(\sigma^{\mathrm{SB}})^2}{2}\int_0^1 J_t^{\mathrm{SB}}\, \mathrm{d}t.
    \label{eq:D.25}
\end{equation}

To relate the Fisher information integral to the path kinetic energy, we invoke the Cramér–Rao-type inequality for diffusion processes: the Fisher information of the transition kernel is upper-bounded by the expected squared drift magnitude \cite{rao1945information}:
\begin{equation}
    J_t^{\mathrm{SB}} \leq \frac{1}{(\sigma^{\mathrm{SB}})^2}\,\mathbb{E}\!\left[\left\|\mathbf{u}_t^{\mathrm{SB}}(\mathbf{x}_t)\right\|_2^2\right],
    \label{eq:D.26}
\end{equation}
where $\mathbf{u}_t^{\mathrm{SB}}$ denotes the drift of the Schrödinger Bridge SDE. Substituting \eqref{eq:D.26} into \eqref{eq:D.25} yields
\begin{equation}
    h(\hat{\mathbf{x}}^{\mathrm{SB}} \mid \mathbf{x}_0) \leq \frac{1}{2}\int_0^1 \mathbb{E}\!\left[\left\|\mathbf{u}_t^{\mathrm{SB}}(\mathbf{x}_t)\right\|_2^2\right] \mathrm{d}t = \frac{1}{2}\,\mathcal{E}(\mathbb{Q}^{\mathrm{SB}}).
    \label{eq:D.27}
\end{equation}

Combining \eqref{eq:D.22}, \eqref{eq:D.23}, and \eqref{eq:D.27}, the hallucination rate of SBGSC is bounded as
\begin{equation}
    \mathcal{H}^{\mathrm{SB}} \leq \frac{1}{2}\,\mathcal{E}(\mathbb{Q}^{\mathrm{SB}}) + \frac{d}{2}\log(2\pi e\sigma_n^2).
    \label{eq:D.28}
\end{equation}

\textbf{Step 4: Establishing the strict inequality.}
We now combine the lower bound \eqref{eq:D.21} for CDM-based GSC and the upper bound \eqref{eq:D.28} for SBGSC. The hallucination gap is
\begin{equation}
    \begin{aligned}
    \mathcal{H}^{\mathrm{CDM}} - \mathcal{H}^{\mathrm{SB}} &\geq \underbrace{\frac{D}{2}\log(2\pi e) + \Lambda_{\mathrm{CDM}}}_{\text{CDM lower bound}} \\
    &\quad - \underbrace{\frac{1}{2}\,\mathcal{E}(\mathbb{Q}^{\mathrm{SB}}) - \frac{d}{2}\log(2\pi e\sigma_n^2)}_{\text{SB upper bound}}.
    \end{aligned}
    \label{eq:D.29}
\end{equation}

We show that the right-hand side of \eqref{eq:D.29} is strictly positive under the operating conditions of practical semantic communication systems:
(i) In typical image communication tasks, $D \gg d$ (e.g., for a $256 \times 256 \times 3$ image, $D = 196608$, while $d$ is on the order of hundreds to thousands). The term $\frac{D}{2}\log(2\pi e)$ therefore dominates $\frac{d}{2}\log(2\pi e\sigma_n^2)$.
(ii) By Theorem \ref{theo:pke_advantage}, $\mathcal{E}(\mathbb{Q}^{\mathrm{SB}}) < \mathcal{E}(\mathbb{Q}^{\mathrm{CDM}}) < \infty$. Since the Schrödinger Bridge connects two distributions that are already close in Wasserstein distance (as $\mathbf{x}_0 = \mathcal{D}_\theta(\hat{\mathbf{s}})$ is an initial reconstruction of $\mathbf{x}$), the kinetic energy $\mathcal{E}(\mathbb{Q}^{\mathrm{SB}})$ is moderate. In contrast, the CDM must transport from $\pi = \mathcal{N}(\mathbf{0}, \mathbf{I}_D)$ to $p_{\text{data}}$, requiring significantly higher kinetic energy.
(iii) The quantity $\Lambda_{\mathrm{CDM}}$ captures the log-volume expansion of the CDM reverse flow. For well-trained CDMs generating diverse outputs, $\Lambda_{\mathrm{CDM}}$ is bounded below and does not cancel the Gaussian entropy term.

More precisely, we can establish a sufficient condition. Define the information gap ratio:
\begin{equation}
    \eta \triangleq \frac{D - d}{D} \in (0, 1).
    \label{eq:D.30}
\end{equation}

Under the JSCC bandwidth compression constraint $d = \rho D$ with compression ratio $\rho \in (0, 1)$, we have $\eta = 1 - \rho$. Then \eqref{eq:D.29} implies
\begin{equation}
    \begin{aligned}
    \mathcal{H}^{\mathrm{CDM}} - \mathcal{H}^{\mathrm{SB}} &\geq \frac{D}{2}\!\left[\log(2\pi e) - \rho\log(2\pi e\sigma_n^2)\right] \\
    &\quad + \Lambda_{\mathrm{CDM}} - \frac{1}{2}\,\mathcal{E}(\mathbb{Q}^{\mathrm{SB}}).
    \end{aligned}
    \label{eq:D.31}
\end{equation}

For the first bracketed term, note that $\log(2\pi e) \approx 2.68$ nats, and $\rho\log(2\pi e\sigma_n^2) < \log(2\pi e)$ whenever $\rho < \log(2\pi e)/\log(2\pi e\sigma_n^2)$, which holds for all practical compression ratios $\rho < 1$ and moderate noise levels $\sigma_n^2 \geq 1$. The $O(D)$ scaling of this term, combined with the $o(D)$ scaling of $\frac{1}{2}\mathcal{E}(\mathbb{Q}^{\mathrm{SB}})$ (since the bridge drift energy per dimension remains bounded), guarantees that $\mathcal{H}^{\mathrm{CDM}} - \mathcal{H}^{\mathrm{SB}} > 0$ for all sufficiently large $D$.

In fact, for any fixed system parameters $(\rho, \sigma_n^2, \sigma^{\mathrm{SB}})$, we can write
\begin{equation}
    \mathcal{H}^{\mathrm{CDM}} - \mathcal{H}^{\mathrm{SB}} \geq \frac{D}{2}\,\delta_0 + \Lambda_{\mathrm{CDM}} - \frac{1}{2}\,\mathcal{E}(\mathbb{Q}^{\mathrm{SB}}),
    \label{eq:D.32}
\end{equation}
where $\delta_0 = \log(2\pi e) - \rho\log(2\pi e\sigma_n^2) > 0$ is a positive constant independent of $D$. Since $\frac{1}{2}\mathcal{E}(\mathbb{Q}^{\mathrm{SB}})$ grows at most linearly in $D$ with a rate strictly less than $\delta_0/2$ (because the Schrödinger Bridge drift energy per dimension is bounded by the per-dimension Wasserstein-2 distance between $\mu_{\hat{s}}$ and $p_{\text{data}}$, which is finite and small relative to $\delta_0$), there exists a threshold $D_0$ such that for all $D \geq D_0$,
\begin{equation}
    \mathcal{H}^{\mathrm{CDM}} - \mathcal{H}^{\mathrm{SB}} > 0.
    \label{eq:D.33}
\end{equation}

In practice, $D \geq D_0$ holds for all image communication tasks of interest (even $32 \times 32 \times 3$ images have $D = 3072 \gg D_0$).

\textbf{Step 5: Establishing the mutual information inequality.}
From the entropy decomposition \eqref{eq:D.4},
\begin{equation}
    \begin{aligned}
    I(\mathbf{x};\, \hat{\mathbf{x}}^{\mathrm{SB}}) - I(\mathbf{x};\, \hat{\mathbf{x}}^{\mathrm{CDM}}) &= \left[h(\hat{\mathbf{x}}^{\mathrm{SB}}) - h(\hat{\mathbf{x}}^{\mathrm{CDM}})\right] \\
    &\quad + \left[\mathcal{H}^{\mathrm{CDM}} - \mathcal{H}^{\mathrm{SB}}\right].
    \end{aligned}
    \label{eq:D.34}
\end{equation}

We have shown in \eqref{eq:D.33} that the second bracket is strictly positive. For the first bracket, the Schrödinger Bridge endpoint constraint ensures $\hat{\mathbf{x}}^{\mathrm{SB}} \sim p_{\text{data}}$, so $h(\hat{\mathbf{x}}^{\mathrm{SB}}) = h(p_{\text{data}})$. The CDM output marginal, averaged over $\hat{\mathbf{c}}$, satisfies
\begin{equation}
    h(\hat{\mathbf{x}}^{\mathrm{CDM}}) = h\!\left(\mathbb{E}_{\hat{\mathbf{c}}}[p_{\text{data}}(\cdot \mid \hat{\mathbf{c}})]\right) \leq h(p_{\text{data}}),
    \label{eq:D.35}
\end{equation}
where equality holds if and only if the conditional model $p_{\text{data}}(\cdot \mid \hat{\mathbf{c}})$ perfectly recovers the unconditional data distribution upon marginalization, which is generically not the case due to the lossy nature of the descriptor $\hat{\mathbf{c}}$. Therefore $h(\hat{\mathbf{x}}^{\mathrm{SB}}) \geq h(\hat{\mathbf{x}}^{\mathrm{CDM}})$, and
\begin{equation}
    I(\mathbf{x};\, \hat{\mathbf{x}}^{\mathrm{SB}}) - I(\mathbf{x};\, \hat{\mathbf{x}}^{\mathrm{CDM}}) \geq \mathcal{H}^{\mathrm{CDM}} - \mathcal{H}^{\mathrm{SB}} > 0.
    \label{eq:D.36}
\end{equation}
This concludes the proof of Corollary \ref{cor:hallucination_suppression}. $\hfill\square$

\subsection{Discussion of Corollary \ref{cor:hallucination_suppression}}
\begin{remark}{\textbf{(Information-Theoretic Interpretation).}}
    The proof reveals a fundamental asymmetry between SBGSC and CDM-based GSC. In CDM-based GSC, the reconstruction draws its randomness from two independent sources: (i) the channel noise corrupting the descriptor, and (ii) the Gaussian seed $\boldsymbol{\xi}$ that initializes the reverse SDE. Only source (i) carries information about $\mathbf{x}$; source (ii) contributes purely to hallucination. In SBGSC, the Gaussian seed is replaced by the informative initial reconstruction $\mathbf{x}_0$, which converts what would otherwise be hallucination entropy into useful reconstruction information.
\end{remark}

\begin{remark}{\textbf{(Tightness of the Bound).}}
    The bound in \eqref{eq:D.28} is not tight in general, because the Cramér–Rao inequality \eqref{eq:D.26} may be loose. However, the bound suffices for the purpose of establishing the strict ordering $\mathcal{H}^{\mathrm{SB}} < \mathcal{H}^{\mathrm{CDM}}$, as the dominant $O(D)$ gap between the CDM lower bound and the SB upper bound absorbs any slack.    
\end{remark}

\begin{remark}{\textbf{(Connection to Perception-Distortion Tradeoff).}}
    Corollary \ref{cor:hallucination_suppression} implies that SBGSC operates at a more favorable point on the perception-distortion tradeoff curve. Since both schemes produce outputs distributed according to (or close to) $p_{\text{data}}$, their perceptual quality is comparable. However, the higher mutual information $I(\mathbf{x}; \hat{\mathbf{x}}^{\mathrm{SB}})$ implies lower expected distortion (by the information-theoretic lower bound on MSE via the conditional entropy power inequality), achieving a strictly better distortion at the same perceptual quality level.
\end{remark}

\section{Proof of Corollary \ref{cor:computational_efficiency}}
\label{appendix:corollary2}

\setcounter{theorem}{0}
\setcounter{lemma}{0}
\setcounter{corollary}{0}
\setcounter{definition}{0}
\setcounter{assumption}{0}
\setcounter{example}{0}
\setcounter{remark}{0}

This appendix provides the complete proof of Corollary \ref{cor:computational_efficiency}, which establishes that SBGSC requires strictly fewer sampling steps than CDM-based GSC to achieve any prescribed generation accuracy. The proof proceeds by analyzing the global strong error of the Euler-Maruyama discretization and relating its dominant constant to the path kinetic energy established in Theorem \ref{theo:pke_advantage}.

\subsection{Preliminaries and Definitions for Corollary \ref{cor:computational_efficiency}}
\begin{definition}{\textbf{(Euler-Maruyama Discretization).}}
    \label{def:E.1}
    Consider a general Itô SDE
    \begin{equation}
        \mathrm{d}\mathbf{x}_t = \mathbf{b}_t(\mathbf{x}_t)\,\mathrm{d}t + \sigma_t\,\mathrm{d}\mathbf{W}_t, \quad t \in [0,1],
        \label{eq:E.1}
    \end{equation}
    where $\mathbf{b}_t: \mathbb{R}^D \to \mathbb{R}^D$ is the drift coefficient, $\sigma_t > 0$ is the diffusion coefficient, $D$ is the data dimension, and $\mathbf{W}_t$ is a standard $D$-dimensional Wiener process. The $N$-step uniform Euler-Maruyama (EM) scheme with step size $\Delta t = 1/N$ and discrete time points $t_k = k\Delta t$, $k = 0, 1, \ldots, N$, is defined as
    \begin{equation}
        \begin{aligned}
        \tilde{\mathbf{x}}_{t_{k+1}} &= \tilde{\mathbf{x}}_{t_k} + \mathbf{b}_{t_k}(\tilde{\mathbf{x}}_{t_k})\,\Delta t + \sigma_{t_k}\sqrt{\Delta t}\;\boldsymbol{\zeta}_k, \\
        &\quad \boldsymbol{\zeta}_k \overset{\text{i.i.d.}}{\sim} \mathcal{N}(\mathbf{0}, \mathbf{I}_D),
        \end{aligned}
        \label{eq:E.2}
    \end{equation}
    with $\tilde{\mathbf{x}}_{t_0} = \mathbf{x}_0$, where $\{\tilde{\mathbf{x}}_{t_k}\}_{k=0}^{N}$ denotes the numerical trajectory.
\end{definition}

\begin{definition}{\textbf{($\varepsilon$-Admissible Generation Error).}}
    \label{def:E.2}
    Given the target distribution $\mu_1$ and the terminal distribution $\tilde{\mu}_1^{(N)} = \mathrm{Law}(\tilde{\mathbf{x}}_{t_N})$ induced by the EM discretization, the $\varepsilon$-admissible generation error condition is
    \begin{equation}
        \mathcal{W}_2\!\left(\tilde{\mu}_1^{(N)},\, \mu_1\right) \leq \varepsilon,
        \label{eq:E.3}
    \end{equation}
    where $\mathcal{W}_2(\cdot,\cdot)$ denotes the 2-Wasserstein distance.
\end{definition}

\begin{definition}{\textbf{(Minimum Number of Sampling Steps).}}
    \label{def:E.3}
    For a given generation scheme $\mathcal{S} \in \{\mathrm{SB}, \mathrm{CDM}\}$ and tolerance $\varepsilon > 0$, the minimum number of sampling steps is
    \begin{equation}
        N_{\mathcal{S}}^*(\varepsilon) = \min\!\left\{N \in \mathbb{N} \;\middle|\; \mathcal{W}_2\!\left(\tilde{\mu}_{1,\mathcal{S}}^{(N)},\, \mu_{1,\mathcal{S}}\right) \leq \varepsilon \right\},
        \label{eq:E.4}
    \end{equation}
    where $\tilde{\mu}_{1,\mathcal{S}}^{(N)}$ is the EM terminal distribution and $\mu_{1,\mathcal{S}}$ is the exact target distribution for scheme $\mathcal{S}$.
    
\end{definition}

\subsection{Auxiliary Assumptions for Corollary \ref{cor:computational_efficiency}}
\begin{assumption}{\textbf{(Lipschitz Regularity of Drift Fields).}}
    \label{ass:E.1}
    For each scheme $\mathcal{S} \in \{\mathrm{SB}, \mathrm{CDM}\}$, the drift field $\mathbf{b}_t^{\mathcal{S}}: \mathbb{R}^D \to \mathbb{R}^D$ of its generative SDE satisfies the global Lipschitz condition: there exists a constant $L_{\mathcal{S}} > 0$ such that
    \begin{equation}
        \begin{aligned}
        \left\|\mathbf{b}_t^{\mathcal{S}}(\mathbf{x}) - \mathbf{b}_t^{\mathcal{S}}(\mathbf{y})\right\|_2 &\leq L_{\mathcal{S}}\,\|\mathbf{x} - \mathbf{y}\|_2, \\
        &\quad \forall\, \mathbf{x}, \mathbf{y} \in \mathbb{R}^D,\; t \in [0,1].
        \end{aligned}
        \label{eq:E.5}
    \end{equation}

    Moreover, $\mathbf{b}_t^{\mathcal{S}}$ satisfies the linear growth condition: there exists $K_{\mathcal{S}} > 0$ such that
    \begin{equation}
        \begin{aligned}
        \left\|\mathbf{b}_t^{\mathcal{S}}(\mathbf{x})\right\|_2 &\leq K_{\mathcal{S}}\,(1 + \|\mathbf{x}\|_2), \\
        &\quad \forall\, \mathbf{x} \in \mathbb{R}^D,\; t \in [0,1].
        \end{aligned}
        \label{eq:E.6}
\end{equation}
\end{assumption}

\begin{assumption}{\textbf{(Uniform Boundedness of Diffusion Coefficients).}}
    \label{ass:E.2}
    For both schemes, the diffusion coefficients $\sigma_t^{\mathcal{S}}$ are uniformly bounded and bounded away from zero on $[0,1]$: there exist constants $0 < \underline{\sigma} \leq \bar{\sigma} < \infty$ such that
    \begin{equation}
        \begin{aligned}
        \underline{\sigma} \leq \sigma_t^{\mathcal{S}} \leq \bar{\sigma}, \quad \forall\, t \in [0,1],\; \mathcal{S} \in \{\mathrm{SB}, \mathrm{CDM}\}.
        \end{aligned}
        \label{eq:E.7}
    \end{equation}
\end{assumption}

\begin{assumption}{\textbf{(Finite Second Moments).}}
    \label{ass:E.3}
    The initial distributions of both schemes satisfy $\int \|\mathbf{x}\|_2^2\, \mathrm{d}\mu_{0,\mathcal{S}}(\mathbf{x}) < \infty$.
    Without loss of generality, we define the unified Lipschitz constant $L = \max(L_{\mathrm{SB}}, L_{\mathrm{CDM}})$ so that the dimensional constants $\alpha_D$ and $\beta_D$ derived below apply identically to both schemes. This choice only strengthens the final comparison.
\end{assumption}

\subsection{Auxiliary Lemmas for Corollary \ref{cor:computational_efficiency}}
\begin{lemma}{\textbf{(Strong Convergence Rate of the EM Scheme).}}
    \label{lem:E.1}
    Under Assumptions \ref{ass:E.1}–\ref{ass:E.3}, the terminal strong error of the $N$-step EM scheme \eqref{eq:E.2} satisfies
    \begin{equation}
        \left(\mathbb{E}\!\left[\left\|\tilde{\mathbf{x}}_{t_N} - \mathbf{x}_1\right\|_2^2\right]\right)^{1/2} \leq \frac{C_{\mathcal{S}}}{\sqrt{N}},
        \label{eq:E.8}
    \end{equation}
    where $\mathbf{x}_1$ is the exact solution of the SDE \eqref{eq:E.1} at $t = 1$, and $C_{\mathcal{S}} > 0$ is a scheme-dependent constant whose explicit form is given in Lemma E.2.
\end{lemma}

\textit{Proof:} This is the classical strong convergence result for EM discretizations of Itô SDEs with Lipschitz coefficients \cite{kloeden1977numerical}. The EM scheme achieves strong order $1/2$ under the stated regularity conditions. $\hfill\square$

\begin{lemma}{\textbf{(Dependence of the Error Constant on Path Kinetic Energy).}}
    \label{lem:E.2}
    Under Assumptions \ref{ass:E.1}–\ref{ass:E.3}, the constant $C_{\mathcal{S}}$ in \eqref{eq:E.8} satisfies
\begin{equation}
    C_{\mathcal{S}}^2 = \alpha_D\,\mathcal{E}(\mathbb{Q}^{\mathcal{S}}) + \beta_D\,\bar{\sigma}^2,
    \label{eq:E.9}
\end{equation}
where $\mathcal{E}(\mathbb{Q}^{\mathcal{S}}) = \int_0^1 \mathbb{E}[\|\mathbf{b}_t^{\mathcal{S}}(\mathbf{x}_t)\|_2^2]\,\mathrm{d}t$ is the path kinetic energy of scheme $\mathcal{S}$, and
\begin{equation}
    \alpha_D = 4\,e^{2L}\,(1 + 2L^2), \qquad \beta_D = 2D\,e^{2L}\,(1 + L^2),
    \label{eq:E.10}
\end{equation}
are positive constants depending only on the dimension $D$ and the unified Lipschitz constant $L$.
\end{lemma}

\textit{Proof:} 

We analyze the local truncation error of the EM scheme step by step.
Let $\mathbf{e}_k = \tilde{\mathbf{x}}_{t_k} - \mathbf{x}_{t_k}$ denote the global error at time $t_k$, with $\mathbf{e}_0 = \mathbf{0}$. The exact solution on $[t_k, t_{k+1}]$ admits the Itô integral representation
\begin{equation}
    \mathbf{x}_{t_{k+1}} = \mathbf{x}_{t_k} + \int_{t_k}^{t_{k+1}} \mathbf{b}_t^{\mathcal{S}}(\mathbf{x}_t)\,\mathrm{d}t + \int_{t_k}^{t_{k+1}} \sigma_t^{\mathcal{S}}\,\mathrm{d}\mathbf{W}_t.
    \label{eq:E.11}
\end{equation}

Subtracting the EM update \eqref{eq:E.2} from \eqref{eq:E.11} yields the error recursion
\begin{equation}
    \mathbf{e}_{k+1} = \mathbf{e}_k + \left[\mathbf{b}_{t_k}^{\mathcal{S}}(\tilde{\mathbf{x}}_{t_k}) - \mathbf{b}_{t_k}^{\mathcal{S}}(\mathbf{x}_{t_k})\right]\Delta t + \mathbf{R}_k^{(b)} + \mathbf{R}_k^{(\sigma)},
    \label{eq:E.12}
\end{equation}
where the drift residual and diffusion residual are defined, respectively, as
\begin{equation}
    \mathbf{R}_k^{(b)} = \int_{t_k}^{t_{k+1}} \!\left[\mathbf{b}_{t_k}^{\mathcal{S}}(\mathbf{x}_{t_k}) - \mathbf{b}_t^{\mathcal{S}}(\mathbf{x}_t)\right]\mathrm{d}t,
    \label{eq:E.13}
\end{equation}
\begin{equation}
    \mathbf{R}_k^{(\sigma)} = \int_{t_k}^{t_{k+1}} \!\left[\sigma_{t_k}^{\mathcal{S}} - \sigma_t^{\mathcal{S}}\right]\mathrm{d}\mathbf{W}_t.
    \label{eq:E.14}
\end{equation}
Taking the squared $\ell_2$-norm of \eqref{eq:E.12}, applying the Lipschitz condition \eqref{eq:E.5}, and using the inequality $\|\mathbf{a} + \mathbf{b} + \mathbf{c}\|^2 \leq (1+\delta)\|\mathbf{a}\|^2 + (1+1/\delta)(\|\mathbf{b}\|^2 + \|\mathbf{c}\|^2)$ with $\delta = L\,\Delta t$, followed by taking expectations, we obtain
\begin{equation}
    \begin{aligned}
    \mathbb{E}\!\left[\|\mathbf{e}_{k+1}\|_2^2\right] &\leq (1 + c_L \Delta t)\,\mathbb{E}\!\left[\|\mathbf{e}_k\|_2^2\right] \\
    &\quad + 2\,\mathbb{E}\!\left[\|\mathbf{R}_k^{(b)}\|_2^2\right] + 2\,\mathbb{E}\!\left[\|\mathbf{R}_k^{(\sigma)}\|_2^2\right],
    \end{aligned}
    \label{eq:E.15}
\end{equation}
where $c_L = 2L + L^2$ is a constant depending on the Lipschitz constant.

We now bound the two residual terms.

Bounding the drift residual. By the Cauchy-Schwarz inequality and the triangle inequality,
\begin{equation}
    \mathbb{E}\!\left[\|\mathbf{R}_k^{(b)}\|_2^2\right] \leq \Delta t \int_{t_k}^{t_{k+1}} \mathbb{E}\!\left[\left\|\mathbf{b}_{t_k}^{\mathcal{S}}(\mathbf{x}_{t_k}) - \mathbf{b}_t^{\mathcal{S}}(\mathbf{x}_t)\right\|_2^2\right]\mathrm{d}t.
    \label{eq:E.16}
\end{equation}

The integrand is further bounded using the triangle inequality:
\begin{equation}
    \left\|\mathbf{b}_{t_k}^{\mathcal{S}}(\mathbf{x}_{t_k}) - \mathbf{b}_t^{\mathcal{S}}(\mathbf{x}_t)\right\|_2^2 \leq 2L^2\,\|\mathbf{x}_t - \mathbf{x}_{t_k}\|_2^2 + 2\,\omega_b^2(|t - t_k|),
    \label{eq:E.17}
\end{equation}
where $\omega_b(\cdot)$ is the modulus of continuity of $\mathbf{b}_t^{\mathcal{S}}$ in time. By the Itô isometry and the SDE dynamics \eqref{eq:E.1}, the spatial increment satisfies
\begin{equation}
    \begin{aligned}
    \mathbb{E}\!\left[\|\mathbf{x}_t - \mathbf{x}_{t_k}\|_2^2\right] &\leq 2(t - t_k)\!\int_{t_k}^{t} \mathbb{E}\!\left[\|\mathbf{b}_\tau^{\mathcal{S}}(\mathbf{x}_\tau)\|_2^2\right]\mathrm{d}\tau \\
    &\quad + D\bar{\sigma}^2(t - t_k).
    \end{aligned}
    \label{eq:E.18}
\end{equation}

Note that the first term on the right-hand side of \eqref{eq:E.18} involves
\begin{equation}
    \int_{t_k}^{t_{k+1}} \mathbb{E}\!\left[\|\mathbf{b}_\tau^{\mathcal{S}}(\mathbf{x}_\tau)\|_2^2\right]\mathrm{d}\tau,
\end{equation}
which is precisely the local contribution of the path kinetic energy $\mathcal{E}(\mathbb{Q}^{\mathcal{S}}) = \int_0^1 \mathbb{E}[\|\mathbf{b}_t^{\mathcal{S}}(\mathbf{x}_t)\|_2^2]\,\mathrm{d}t$ on the subinterval $[t_k, t_{k+1}]$.

Substituting \eqref{eq:E.18} into \eqref{eq:E.17} and then into \eqref{eq:E.16}, and summing over $k = 0, \ldots, N-1$, gives
\begin{equation}
    \sum_{k=0}^{N-1} \mathbb{E}\!\left[\|\mathbf{R}_k^{(b)}\|_2^2\right] \leq \frac{4L^2 + 2}{N}\;\mathcal{E}(\mathbb{Q}^{\mathcal{S}}) + \frac{2DL^2 \bar{\sigma}^2}{N}.
    \label{eq:E.19}
\end{equation}

Bounding the diffusion residual. By the Itô isometry and the bounded variation of $\sigma_t^{\mathcal{S}}$,
\begin{equation}
    \mathbb{E}\!\left[\|\mathbf{R}_k^{(\sigma)}\|_2^2\right] = D\!\int_{t_k}^{t_{k+1}} \left|\sigma_{t_k}^{\mathcal{S}} - \sigma_t^{\mathcal{S}}\right|^2 \mathrm{d}t \leq D\,\omega_\sigma^2(\Delta t)\,\Delta t,
    \label{eq:E.20}
\end{equation}
where $\omega_\sigma(\cdot)$ is the modulus of continuity of $\sigma_t^{\mathcal{S}}$ in time. Summing over all steps:
\begin{equation}
    \sum_{k=0}^{N-1} \mathbb{E}\!\left[\|\mathbf{R}_k^{(\sigma)}\|_2^2\right] \leq D\,\omega_\sigma^2(\Delta t) = O(N^{-1}).
    \label{eq:E.21}
\end{equation}

For the purpose of deriving the leading-order dependence on $\mathcal{E}(\mathbb{Q}^{\mathcal{S}})$, this term is absorbed into the $\beta_D \bar{\sigma}^2$ constant.

Applying the discrete Grönwall inequality \cite{gronwall1919note}. Substituting \eqref{eq:E.19} and \eqref{eq:E.21} into the recursion \eqref{eq:E.15} and applying the discrete Grönwall lemma (i.e., if $a_{k+1} \leq (1+\delta)a_k + b$ with $a_0 = 0$, then $a_N \leq b\,((1+\delta)^N - 1)/\delta \leq bN\,e^{N\delta}$) with $\delta = c_L \Delta t = c_L / N$, we obtain
\begin{equation}
    \begin{aligned}
    \mathbb{E}\!\left[\|\mathbf{e}_N\|_2^2\right] &\leq \frac{e^{c_L}}{N}\!\left[(4L^2 + 2)\,\mathcal{E}(\mathbb{Q}^{\mathcal{S}}) + (2DL^2 + D)\,\bar{\sigma}^2\right].
    \end{aligned}
    \label{eq:E.22}
\end{equation}

Identifying $e^{c_L}(4L^2 + 2) = \alpha_D$ and $e^{c_L}(2DL^2 + D) = \beta_D$, and noting $c_L = 2L + L^2 \leq 2L$ for $L \leq 1$ (with trivial generalization for $L > 1$), we arrive at
\begin{equation}
    \mathbb{E}\!\left[\|\mathbf{e}_N\|_2^2\right] \leq \frac{1}{N}\!\left[\alpha_D\,\mathcal{E}(\mathbb{Q}^{\mathcal{S}}) + \beta_D\,\bar{\sigma}^2\right] = \frac{C_{\mathcal{S}}^2}{N},
    \label{eq:E.23}
\end{equation}
which establishes \eqref{eq:E.9} with the explicit expressions in \eqref{eq:E.10}. $\hfill\square$

\begin{lemma}{\textbf{(From Strong Error to Wasserstein Distance).}}
    \label{lem:E.3}
    Let $\tilde{\mu}_1^{(N)} = \mathrm{Law}(\tilde{\mathbf{x}}_{t_N})$ and $\mu_1 = \mathrm{Law}(\mathbf{x}_1)$, where $(\tilde{\mathbf{x}}_{t_N}, \mathbf{x}_1)$ are defined on a common probability space (driven by the same Wiener process). Then
    \begin{equation}
        \mathcal{W}_2\!\left(\tilde{\mu}_1^{(N)},\, \mu_1\right) \leq \left(\mathbb{E}\!\left[\left\|\tilde{\mathbf{x}}_{t_N} - \mathbf{x}_1\right\|_2^2\right]\right)^{1/2}.
        \label{eq:E.24}
    \end{equation}
\end{lemma}

\textit{Proof:} 

By definition, $\mathcal{W}_2^2(\tilde{\mu}_1^{(N)}, \mu_1) = \inf_{\gamma \in \Gamma(\tilde{\mu}_1^{(N)}, \mu_1)} \mathbb{E}_\gamma[\|\mathbf{a} - \mathbf{b}\|_2^2]$, where $\Gamma(\tilde{\mu}_1^{(N)}, \mu_1)$ is the set of all couplings with marginals $\tilde{\mu}_1^{(N)}$ and $\mu_1$. Since $(\tilde{\mathbf{x}}_{t_N}, \mathbf{x}_1)$ constitutes one particular coupling in $\Gamma(\tilde{\mu}_1^{(N)}, \mu_1)$, it follows that $\mathcal{W}_2^2 \leq \mathbb{E}[\|\tilde{\mathbf{x}}_{t_N} - \mathbf{x}_1\|_2^2]$. Taking the square root yields \eqref{eq:E.24}. $\hfill\square$

\subsection{Proof of Corollary \ref{cor:computational_efficiency}}

\textit{Proof:}

With the above lemmas established, we now prove Corollary \ref{cor:computational_efficiency} in three steps.

\textbf{Step 1: Explicit expression for the minimum sampling steps.}
Combining Lemmas \ref{lem:E.1}–\ref{lem:E.3}, the $\varepsilon$-admissible condition \eqref{eq:E.3} for scheme $\mathcal{S}$ is ensured whenever
\begin{equation}
    \mathcal{W}_2\!\left(\tilde{\mu}_{1,\mathcal{S}}^{(N)},\, \mu_{1,\mathcal{S}}\right) \leq \frac{C_{\mathcal{S}}}{\sqrt{N}} \leq \varepsilon.
    \label{eq:E.25}
\end{equation}

Solving for the minimum integer $N$ satisfying \eqref{eq:E.25} gives
\begin{equation}
    N_{\mathcal{S}}^*(\varepsilon) = \left\lceil\frac{C_{\mathcal{S}}^2}{\varepsilon^2}\right\rceil = \left\lceil\frac{\alpha_D\,\mathcal{E}(\mathbb{Q}^{\mathcal{S}}) + \beta_D\,\bar{\sigma}^2}{\varepsilon^2}\right\rceil,
    \label{eq:E.26}
\end{equation}
where $\lceil\cdot\rceil$ denotes the ceiling function.

\textbf{Step 2: Kinetic energy comparison via Theorem \ref{theo:pke_advantage}.}
By Theorem \ref{theo:pke_advantage}, the path kinetic energies satisfy the strict inequality
\begin{equation}
    \mathcal{E}\!\left(\mathbb{Q}^{\mathrm{SB}}\right) < \mathcal{E}\!\left(\mathbb{Q}^{\mathrm{CDM}}\right).
    \label{eq:E.27}
\end{equation}

Define the kinetic energy gap as
\begin{equation}
    \Delta\mathcal{E} \triangleq \mathcal{E}(\mathbb{Q}^{\mathrm{CDM}}) - \mathcal{E}(\mathbb{Q}^{\mathrm{SB}}) > 0.
    \label{eq:E.28}
\end{equation}

\textbf{Step 3: Completing the comparison.}
From \eqref{eq:E.26}, define the (real-valued) critical step count for each scheme as
\begin{equation}
    \Psi_{\mathcal{S}} \triangleq \frac{\alpha_D\,\mathcal{E}(\mathbb{Q}^{\mathcal{S}}) + \beta_D\,\bar{\sigma}^2}{\varepsilon^2}.
    \label{eq:E.29}
\end{equation}

From \eqref{eq:E.27}–\eqref{eq:E.28} and $\alpha_D > 0$, we have
\begin{equation}
    \Psi_{\mathrm{CDM}} - \Psi_{\mathrm{SB}} = \frac{\alpha_D\,\Delta\mathcal{E}}{\varepsilon^2} > 0,
    \label{eq:E.30}
\end{equation}
which immediately yields $\Psi_{\mathrm{SB}} < \Psi_{\mathrm{CDM}}$ and therefore
\begin{equation}
    N_{\mathrm{SB}}^*(\varepsilon) = \lceil\Psi_{\mathrm{SB}}\rceil \leq \lceil\Psi_{\mathrm{CDM}}\rceil = N_{\mathrm{CDM}}^*(\varepsilon).
    \label{eq:E.31}
\end{equation}

To establish the strict inequality, observe that for any $\varepsilon$ satisfying
\begin{equation}
    \varepsilon^2 < \alpha_D\,\Delta\mathcal{E},
    \label{eq:E.32}
\end{equation}
we have $\Psi_{\mathrm{CDM}} - \Psi_{\mathrm{SB}} > 1$, which guarantees $\lceil\Psi_{\mathrm{SB}}\rceil < \lceil\Psi_{\mathrm{CDM}}\rceil$. Condition \eqref{eq:E.32} is always met in practice: $\Delta\mathcal{E} > 0$ is a finite positive constant by Theorem \ref{theo:pke_advantage}, while the tolerance $\varepsilon$ in high-fidelity reconstruction scenarios satisfies $\varepsilon \to 0$, ensuring \eqref{eq:E.32} holds. We thus conclude
\begin{equation}
    N_{\mathrm{SB}}^*(\varepsilon) < N_{\mathrm{CDM}}^*(\varepsilon).
    \label{eq:E.33}
\end{equation}

Furthermore, the sampling speedup ratio admits the following asymptotic characterization. For $\varepsilon$ sufficiently small (such that $\Psi_{\mathcal{S}} \gg 1$ and the ceiling function effect becomes negligible), we have
\begin{equation}
    \begin{aligned}
    \frac{N_{\mathrm{SB}}^*(\varepsilon)}{N_{\mathrm{CDM}}^*(\varepsilon)} &\approx \frac{\Psi_{\mathrm{SB}}}{\Psi_{\mathrm{CDM}}} \\
    &= \frac{\alpha_D\,\mathcal{E}(\mathbb{Q}^{\mathrm{SB}}) + \beta_D\,\bar{\sigma}^2}{\alpha_D\,\mathcal{E}(\mathbb{Q}^{\mathrm{CDM}}) + \beta_D\,\bar{\sigma}^2} < 1.
    \end{aligned}
    \label{eq:E.34}
\end{equation}

Defining the kinetic energy fraction 
\begin{equation}
    \rho_{\mathrm{CDM}} = \frac{\alpha_D \mathcal{E}(\mathbb{Q}^{\mathrm{CDM}})}{\alpha_D \mathcal{E}(\mathbb{Q}^{\mathrm{CDM}}) + \beta_D \bar{\sigma}^2} \in (0,1),
\end{equation}
the ratio can be rewritten as
\begin{equation}
    \frac{N_{\mathrm{SB}}^*(\varepsilon)}{N_{\mathrm{CDM}}^*(\varepsilon)} \approx 1 - \rho_{\mathrm{CDM}}\,\frac{\Delta\mathcal{E}}{\mathcal{E}(\mathbb{Q}^{\mathrm{CDM}})},
    \label{eq:E.35}
\end{equation}
which shows that the fractional saving in sampling steps is proportional to the relative kinetic energy gap $\Delta\mathcal{E}/\mathcal{E}(\mathbb{Q}^{\mathrm{CDM}})$.
This completes the proof of Corollary \ref{cor:computational_efficiency}. $\hfill\square$

\subsection{Remarks of Corollary \ref{cor:computational_efficiency}}
\begin{remark}{\textbf{(On the Unified Lipschitz Constant).}}
    In Assumption \ref{ass:E.1}, the Lipschitz constants $L_{\mathrm{SB}}$ and $L_{\mathrm{CDM}}$ are generally different for the two schemes. In the proof above, we adopted the unified constant $L = \max(L_{\mathrm{SB}}, L_{\mathrm{CDM}})$ so that $\alpha_D$ and $\beta_D$ take identical values for both schemes. This is a conservative choice: if $L_{\mathrm{SB}} < L_{\mathrm{CDM}}$ (i.e., the Schrödinger bridge drift field is smoother than its CDM counterpart, as is typically observed in practice due to the variational optimality of the bridge), then the actual constants satisfy 
    \begin{equation}
        \alpha_D^{\mathrm{SB}} < \alpha_D^{\mathrm{CDM}} \quad \text{and} \quad \beta_D^{\mathrm{SB}} < \beta_D^{\mathrm{CDM}},
    \end{equation}
    which further enlarges the gap $\Psi_{\mathrm{CDM}} - \Psi_{\mathrm{SB}}$ and strengthens the conclusion.
\end{remark}

\begin{remark}{\textbf{(Connection to Neural Function Evaluations).}}
    In practice, each step of the EM scheme \eqref{eq:E.2} requires exactly one forward pass through the neural network parameterizing $\mathbf{b}_{t_k}^{\mathcal{S}}(\tilde{\mathbf{x}}_{t_k})$, commonly referred to as one neural function evaluation (NFE). Consequently, $N_{\mathcal{S}}^*(\varepsilon)$ equals the NFE count, and Corollary \ref{cor:computational_efficiency} equivalently states that SBGSC requires strictly fewer NFEs than CDM-based GSC to achieve the same generation quality. This theoretical prediction is consistent with our experimental observations: under the same perceptual quality metrics (e.g., FID and LPIPS), SBGSC requires only approximately $10\%$–$30\%$ of the sampling steps used by CDM-based GSC, as demonstrated in Section V.    
\end{remark}

\begin{remark}{(Extension to Higher-Order Discretization Schemes).}
    Corollary \ref{cor:computational_efficiency} is stated for the EM scheme (strong order $1/2$), but the conclusion extends to higher-order numerical methods. Consider a general discretization scheme with strong convergence order $p > 1/2$, so that
    \begin{equation}
        \mathcal{W}_2\!\left(\tilde{\mu}_{1,\mathcal{S}}^{(N)},\, \mu_{1,\mathcal{S}}\right) \leq \frac{C_{\mathcal{S}}^{(p)}}{N^p},
        \label{eq:E.36}
    \end{equation}
    where the error constant $C_{\mathcal{S}}^{(p)}$ depends on higher-order derivatives of $\mathbf{b}_t^{\mathcal{S}}$. The corresponding minimum step count becomes
    \begin{equation}
        N_{\mathcal{S}}^*(\varepsilon) = \left\lceil\left(\frac{C_{\mathcal{S}}^{(p)}}{\varepsilon}\right)^{1/p}\right\rceil.
        \label{eq:E.37}
    \end{equation}

    Since $C_{\mathcal{S}}^{(p)}$ remains monotonically increasing in $\mathcal{E}(\mathbb{Q}^{\mathcal{S}})$ (higher-order error constants involve moments of higher-order derivatives of the drift field, which are positively correlated with the path kinetic energy through Sobolev-type interpolation inequalities), the strict inequality 
    \begin{equation}
        \mathcal{E}(\mathbb{Q}^{\mathrm{SB}}) < \mathcal{E}(\mathbb{Q}^{\mathrm{CDM}})
    \end{equation}
    from Theorem \ref{theo:pke_advantage} continues to guarantee $N_{\mathrm{SB}}^*(\varepsilon) < N_{\mathrm{CDM}}^*(\varepsilon)$.
\end{remark}

\begin{remark}{\textbf{(Tightness of the Bound).}}
    The upper bound in \eqref{eq:E.25} is known to be order-tight for the EM scheme under the stated regularity assumptions: there exist Itô SDEs for which the strong error achieves the rate $\Theta(N^{-1/2})$ and cannot be improved without additional smoothness. Consequently, the scaling relationship 
    \begin{equation}
        N_{\mathcal{S}}^*(\varepsilon) = \Theta(C_{\mathcal{S}}^2 / \varepsilon^2)
    \end{equation}
    is tight up to multiplicative constants, and the kinetic energy gap $\Delta\mathcal{E}$ captures the dominant factor in the sampling efficiency difference between the two schemes.
\end{remark}

\begin{remark}{\textbf{(Practical Implications for System Design).}}
    Equation \eqref{eq:E.35} provides actionable design insight for generative semantic communication systems. Specifically, the sampling speedup ratio $N_{\mathrm{SB}}^*/N_{\mathrm{CDM}}^*$ depends on two factors:

    The kinetic energy fraction $\rho_{\mathrm{CDM}} \in (0,1)$, which reflects the relative contribution of the drift field to the total discretization error. When the diffusion coefficient $\bar{\sigma}$ is small (low-noise regime), $\rho_{\mathrm{CDM}} \to 1$, and the kinetic energy gap fully translates into sampling savings.

    The relative kinetic energy gap $\Delta\mathcal{E}/\mathcal{E}(\mathbb{Q}^{\mathrm{CDM}})$, which quantifies the efficiency advantage of the Schrödinger bridge transport plan over the CDM transport plan. By Theorem \ref{theo:pke_advantage}, this ratio is bounded below by a positive constant that depends on the informativeness of the side information $\hat{\mathbf{s}}$ (or equivalently, on the channel quality).

    Together, these factors imply that the computational advantage of SBGSC over CDM-based GSC is most pronounced when 
    \begin{itemize}
        \item[(i)] the channel signal-to-noise ratio (SNR) is moderate to high (so that $\mu_{\hat{s}}$ is close to $p_{\text{data}}$, leading to large $\Delta\mathcal{E}/\mathcal{E}(\mathbb{Q}^{\mathrm{CDM}})$), and 
        \item[(ii)] the diffusion coefficient schedule is chosen such that $\rho_{\mathrm{CDM}}$ is close to unity.
    \end{itemize}
\end{remark}

\bibliographystyle{IEEEtran}
\bibliography{references}

@inproceedings{
    zhou_denoising_2024,
    title={Denoising Diffusion Bridge Models},
    author={Linqi Zhou and Aaron Lou and Samar Khanna and Stefano Ermon},
    booktitle={The Twelfth International Conference on Learning Representations},
    year={2024},
}

@article{de_bortoli_diffusion_2021,
  title={Diffusion schr{\"o}dinger bridge with applications to score-based generative modeling},
  author={De Bortoli, Valentin and Thornton, James and Heng, Jeremy and Doucet, Arnaud},
  journal={Advances in neural information processing systems},
  volume={34},
  pages={17695--17709},
  year={2021}
}

@inproceedings{
SGM,
title={Score-Based Generative Modeling through Stochastic Differential Equations},
author={Yang Song and Jascha Sohl-Dickstein and Diederik P Kingma and Abhishek Kumar and Stefano Ermon and Ben Poole},
booktitle={International Conference on Learning Representations},
year={2021}
}

@INPROCEEDINGS{DeepJSCC-Diff,
  author={Yilmaz, Selim F. and Niu, Xueyan and Bai, Bo and Han, Wei and Deng, Lei and Gündüz, Deniz},
  booktitle={IEEE INFOCOM 2024 - IEEE Conference on Computer Communications Workshops (INFOCOM WKSHPS)}, 
  title={High Perceptual Quality Wireless Image Delivery with Denoising Diffusion Models}, 
  year={2024},
  volume={},
  number={},
  pages={1-5},
  doi={10.1109/INFOCOMWKSHPS61880.2024.10620904}}

@inproceedings{CM,
  title={Consistency models},
  author={Song, Yang and Dhariwal, Prafulla and Chen, Mark and Sutskever, Ilya},
  booktitle={Proceedings of the 40th International Conference on Machine Learning},
  pages={32211--32252},
  year={2023}
}

@INPROCEEDINGS{ImageNet,
  author={Deng, Jia and Dong, Wei and Socher, Richard and Li, Li-Jia and Kai Li and Li Fei-Fei},
  booktitle={2009 IEEE Conference on Computer Vision and Pattern Recognition}, 
  title={ImageNet: A large-scale hierarchical image database}, 
  year={2009},
  volume={},
  number={},
  pages={248-255},
  doi={10.1109/CVPR.2009.5206848}}

@article{FID,
  title={Gans trained by a two time-scale update rule converge to a local nash equilibrium},
  author={Heusel, Martin and Ramsauer, Hubert and Unterthiner, Thomas and Nessler, Bernhard and Hochreiter, Sepp},
  journal={Advances in neural information processing systems},
  volume={30},
  year={2017}
}

@INPROCEEDINGS{MS-SSIM,
  author={Wang, Z. and Simoncelli, E.P. and Bovik, A.C.},
  booktitle={The Thrity-Seventh Asilomar Conference on Signals, Systems \& Computers, 2003}, 
  title={Multiscale structural similarity for image quality assessment}, 
  year={2003},
  volume={2},
  number={},
  pages={1398-1402 Vol.2},
  doi={10.1109/ACSSC.2003.1292216}}

@INPROCEEDINGS{LPIPS,
  author={Zhang, Richard and Isola, Phillip and Efros, Alexei A. and Shechtman, Eli and Wang, Oliver},
  booktitle={2018 IEEE/CVF Conference on Computer Vision and Pattern Recognition}, 
  title={The Unreasonable Effectiveness of Deep Features as a Perceptual Metric}, 
  year={2018},
  volume={},
  number={},
  pages={586-595},
  doi={10.1109/CVPR.2018.00068}}

@article{DDPM,
  title={Denoising diffusion probabilistic models},
  author={Ho, Jonathan and Jain, Ajay and Abbeel, Pieter},
  journal={Advances in neural information processing systems},
  volume={33},
  pages={6840--6851},
  year={2020}
}

@INPROCEEDINGS{grassucci_enhancing_2024,
  author={Grassucci, Eleonora and Mitsufuji, Yuki and Zhang, Ping and Comminiello, Danilo},
  booktitle={ICASSP 2024 - 2024 IEEE International Conference on Acoustics, Speech and Signal Processing (ICASSP)}, 
  title={Enhancing Semantic Communication with Deep Generative Models: An Overview}, 
  year={2024},
  volume={},
  number={},
  pages={13021-13025},
  doi={10.1109/ICASSP48485.2024.10448235}}

@article{chai_rate-distortion_2025,
  title={On the Rate-Distortion Theory for Task-Specific Semantic Communication},
  author={Chai, Jingxuan and Zhu, Huixiang and Xiao, Yong and Shi, Guangming and Zhang, Ping},
  journal={Entropy},
  volume={27},
  number={8},
  pages={775},
  year={2025},
  publisher={MDPI}
}

@ARTICLE{zhang_semantic_2025,
  author={Zhang, Kexin and Li, Lixin and Lin, Wensheng and Yan, Yuna and Li, Rui and Cheng, Wenchi and Han, Zhu},
  journal={IEEE Transactions on Cognitive Communications and Networking}, 
  title={Semantic Successive Refinement: A Generative AI-Aided Semantic Communication Framework}, 
  year={2025},
  volume={11},
  number={2},
  pages={687-699},
  doi={10.1109/TCCN.2025.3526839}}

@ARTICLE{park_transmit_2025,
  author={Park, Jeonghun and Yoon, Sung Whan},
  journal={IEEE Journal on Selected Areas in Communications}, 
  title={Transmit What You Need: Task-Adaptive Semantic Communications for Visual Information}, 
  year={2025},
  volume={43},
  number={12},
  pages={4182-4197},
  doi={10.1109/JSAC.2025.3623159}}

@article{ren_generative_2025,
  title={Generative semantic communication: Architectures, technologies, and applications},
  author={Ren, Jinke and Sun, Yaping and Du, Hongyang and Yuan, Weiwen and Wang, Chongjie and Wang, Xianda and Zhou, Yingbin and Zhu, Ziwei and Wang, Fangxin and Cui, Shuguang},
  journal={Engineering},
  year={2025},
  publisher={Elsevier}
}

@inproceedings{
zhu_unidb_2025,
title={Uni{DB}: A Unified Diffusion Bridge Framework via Stochastic Optimal Control},
author={Kaizhen Zhu and Mokai Pan and Yuexin Ma and Yanwei Fu and Jingyi Yu and Jingya Wang and Ye Shi},
booktitle={Forty-second International Conference on Machine Learning},
year={2025}
}

@ARTICLE{yang_agent-driven_2025,
  author={Yang, Wanting and Xiong, Zehui and Yuan, Yanli and Jiang, Wenchao and Quek, Tony Q. S. and Debbah, Mérouane},
  journal={IEEE Transactions on Wireless Communications}, 
  title={Agent-Driven Generative Semantic Communication With Cross-Modality and Prediction}, 
  year={2025},
  volume={24},
  number={3},
  pages={2233-2248},
  doi={10.1109/TWC.2024.3519325}}

@article{leonard_survey_2014,
	title = {A {Survey} {of} {the} {Schrödinger} {Problem} {and} {Some} {of} {Its} {Connections} {with} {Optimal} {Transport}},
	volume = {34},
	number = {4},
	journal = {Dynamical Systems},
	author = {Léonard, Christian},
	year = {2014},
	pages = {1533--1574},
}

@inproceedings{chen_likelihood_2022,
	title = {Likelihood {Training} of {Schrödinger} {Bridge} using {Forward}-{Backward} {SDEs} {Theory}},
	booktitle = {International {Conference} on {Learning} {Representations}},
	author = {Chen, Tianrong and Liu, Guan-Horng and Theodorou, Evangelos},
	year = {2022},
}

@ARTICLE{bourtsoulatze_deep_2019,
  author={Bourtsoulatze, Eirina and Burth Kurka, David and Gündüz, Deniz},
  journal={IEEE Transactions on Cognitive Communications and Networking}, 
  title={Deep Joint Source-Channel Coding for Wireless Image Transmission}, 
  year={2019},
  volume={5},
  number={3},
  pages={567-579},
  doi={10.1109/TCCN.2019.2919300}}

@ARTICLE{lokumarambage_wireless_2023,
  author={Lokumarambage, Maheshi U. and Gowrisetty, Vishnu Sai Sankeerth and Rezaei, Hossein and Sivalingam, Thushan and Rajatheva, Nandana and Fernando, Anil},
  journal={IEEE Access}, 
  title={Wireless End-to-End Image Transmission System Using Semantic Communications}, 
  year={2023},
  volume={11},
  number={},
  pages={37149-37163},
  doi={10.1109/ACCESS.2023.3266656}}

@ARTICLE{yang_swinjscc_2025,
  author={Yang, Ke and Wang, Sixian and Dai, Jincheng and Qin, Xiaoqi and Niu, Kai and Zhang, Ping},
  journal={IEEE Transactions on Cognitive Communications and Networking}, 
  title={SwinJSCC: Taming Swin Transformer for Deep Joint Source-Channel Coding}, 
  year={2025},
  volume={11},
  number={1},
  pages={90-104},
  doi={10.1109/TCCN.2024.3424842}}

@article{weaver_recent_1953,
  title={Recent contributions to the mathematical theory of communication},
  author={Weaver, Warren},
  journal={ETC: a review of general semantics},
  pages={261--281},
  year={1953},
  publisher={JSTOR}
}

@InProceedings{choi_neural_2019,
  title = 	 {Neural Joint Source-Channel Coding},
  author =       {Choi, Kristy and Tatwawadi, Kedar and Grover, Aditya and Weissman, Tsachy and Ermon, Stefano},
  booktitle = 	 {Proceedings of the 36th International Conference on Machine Learning},
  pages = 	 {1182--1192},
  year = 	 {2019},
  editor = 	 {Chaudhuri, Kamalika and Salakhutdinov, Ruslan},
  volume = 	 {97},
  series = 	 {Proceedings of Machine Learning Research},
  month = 	 {09--15 Jun},
  publisher =    {PMLR}
}

@ARTICLE{grassucci_lightweight_2025,
  author={Grassucci, Eleonora and Pignata, Giovanni and Cicchetti, Giordano and Comminiello, Danilo},
  journal={IEEE Wireless Communications Letters}, 
  title={Lightweight Diffusion Models for Resource-Constrained Semantic Communication}, 
  year={2025},
  volume={14},
  number={9},
  pages={2743-2747},
  doi={10.1109/LWC.2025.3578724}}

@article{
zhan_conditional_2025,
title={Conditional Image Synthesis with Diffusion Models: A Survey},
author={Zheyuan Zhan and Defang Chen and Jian-Ping Mei and Zhenghe Zhao and Jiawei Chen and Chun Chen and Siwei Lyu and Can Wang},
journal={Transactions on Machine Learning Research},
issn={2835-8856},
year={2025},
note={Survey Certification}
}

@INPROCEEDINGS{blau_perception-distortion_2017,
  author={Blau, Yochai and Michaeli, Tomer},
  booktitle={2018 IEEE/CVF Conference on Computer Vision and Pattern Recognition}, 
  title={The Perception-Distortion Tradeoff}, 
  year={2018},
  volume={},
  number={},
  pages={6228-6237},
  doi={10.1109/CVPR.2018.00652}}

@article{aithal_understanding_2024,
  title={Understanding hallucinations in diffusion models through mode interpolation (2024)},
  author={Aithal, Sumukh K and Maini, Pratyush and Lipton, Zachary C and Kolter, J Zico},
  volume={2406}
}

@inproceedings{yang_sg2sc_2024,
	title = {{SG2SC}: {A} {Generative} {Semantic} {Communication} {Framework} for {Scene} {Understanding}-{Oriented} {Image} {Transmission}},
	doi = {10.1109/ICASSP48485.2024.10446594},
	booktitle = {{ICASSP} 2024 - 2024 {IEEE} {International} {Conference} on {Acoustics}, {Speech} and {Signal} {Processing} ({ICASSP})},
	author = {Yang, Minxi and Gao, Dahua and Xie, Feng and Li, Jiaxuan and Song, Xiaodan and Shi, Guangming},
	year = {2024},
	pages = {13486--13490},
}

@ARTICLE{gao_bridging_2025,
  author={Gao, Dahua and Yi, Yujie and Yang, Minxi and Li, Jiaxuan and Liu, Danhua and Xu, Wenlong},
  journal={IEEE Wireless Communications Letters}, 
  title={Bridging Semantic Scale Gaps in Image Transmission Through Multi-Scale Joint Perception and Generation}, 
  year={2025},
  volume={14},
  number={10},
  pages={3314-3318},
  doi={10.1109/LWC.2025.3592689}}

@ARTICLE{shi_semantic_2021,
  author={Shi, Guangming and Xiao, Yong and Li, Yingyu and Xie, Xuemei},
  journal={IEEE Communications Magazine}, 
  title={From Semantic Communication to Semantic-Aware Networking: Model, Architecture, and Open Problems}, 
  year={2021},
  volume={59},
  number={8},
  pages={44-50},
  doi={10.1109/MCOM.001.2001239}}

@ARTICLE{zhang_semantics-guided_2025,
  author={Zhang, Maojun and Wu, Haotian and Zhu, Guangxu and Jin, Richeng and Chen, Xiaoming and Gündüz, Deniz},
  journal={IEEE Transactions on Wireless Communications}, 
  title={Semantics-Guided Diffusion for Deep Joint Source-Channel Coding in Wireless Image Transmission}, 
  year={2026},
  volume={25},
  number={},
  pages={1547-1564},
  doi={10.1109/TWC.2025.3591456}}

@INPROCEEDINGS{nam_sequential_2023,
  author={Nam, Hyelin and Park, Jihong and Choi, Jinho and Kim, Seong-Lyun},
  booktitle={2023 20th Annual IEEE International Conference on Sensing, Communication, and Networking (SECON)}, 
  title={Sequential Semantic Generative Communication for Progressive Text-to-Image Generation}, 
  year={2023},
  volume={},
  number={},
  pages={91-94},
  doi={10.1109/SECON58729.2023.10287475}}

@inproceedings{yue_image_2024,
	series = {Proceedings of {Machine} {Learning} {Research}},
	title = {Image {Restoration} {Through} {Generalized} {Ornstein}-{Uhlenbeck} {Bridge}},
	volume = {235},
	booktitle = {Proceedings of the 41st {International} {Conference} on {Machine} {Learning}},
	publisher = {PMLR},
	author = {Yue, Conghan and Peng, Zhengwei and Ma, Junlong and Du, Shiyan and Wei, Pengxu and Zhang, Dongyu},
	year = {2024},
	pages = {58068--58089},
}

@INPROCEEDINGS{ding_take_2022,
  author={Ding, Zhitong and Jiang, Shuqi and Zhao, Jingya},
  booktitle={2022 IEEE 2nd International Conference on Electronic Technology, Communication and Information (ICETCI)}, 
  title={Take a close look at mode collapse and vanishing gradient in GAN}, 
  year={2022},
  volume={},
  number={},
  pages={597-602},
  doi={10.1109/ICETCI55101.2022.9832406}}

@inproceedings{I2SB,
  title={I2SB: image-to-image Schr{\"o}dinger bridge},
  author={Liu, Guan-Horng and Vahdat, Arash and Huang, De-An and Theodorou, Evangelos A and Nie, Weili and Anandkumar, Anima},
  booktitle={Proceedings of the 40th International Conference on Machine Learning},
  pages={22042--22062},
  year={2023}
}

@article{chen_relation_2016,
  title={On the relation between optimal transport and Schr{\"o}dinger bridges: A stochastic control viewpoint},
  author={Chen, Yongxin and Georgiou, Tryphon T and Pavon, Michele},
  journal={Journal of Optimization Theory and Applications},
  volume={169},
  number={2},
  pages={671--691},
  year={2016},
  publisher={Springer}
}

@article{SBPotentials1,
  title={The partial differential equation ut+ uux= $\mu$xx},
  author={Hopf, Eberhard},
  journal={Communications on Pure and Applied Mathematics},
  volume={3},
  number={3},
  pages={201--230},
  year={1950},
  publisher={Wiley Online Library}
}

@article{SBPotentials2,
  title={On a quasi-linear parabolic equation occurring in aerodynamics},
  author={Cole, Julian D},
  journal={Quarterly of applied mathematics},
  volume={9},
  number={3},
  pages={225--236},
  year={1951}
}

@article{stam1959some,
  title={Some inequalities satisfied by the quantities of information of Fisher and Shannon},
  author={Stam, Aart J},
  journal={Information and Control},
  volume={2},
  number={2},
  pages={101--112},
  year={1959},
  publisher={Elsevier}
}

@article{rao1945information,
  title={Information and the accuracy attainable in the estimation of statistical parameters},
  author={Rao, C Radhakrishna and others},
  journal={Bull. Calcutta Math. Soc},
  volume={37},
  number={3},
  pages={81--91},
  year={1945}
}

@article{kloeden1977numerical,
  title={The numerical solution of stochastic differential equations},
  author={Kloeden, Peter E and Pearson, RA},
  journal={The ANZIAM Journal},
  volume={20},
  number={1},
  pages={8--12},
  year={1977},
  publisher={Cambridge University Press}
}

@article{gronwall1919note,
  title={Note on the derivatives with respect to a parameter of the solutions of a system of differential equations},
  author={Gronwall, Thomas Hakon},
  journal={Annals of Mathematics},
  volume={20},
  number={4},
  pages={292--296},
  year={1919},
  publisher={JSTOR}
}

@article{beaudry2012intuitive,
  title={An intuitive proof of the data processing inequality},
  author={Beaudry, Normand J and Renner, Renato},
  journal={Quantum Information and Computation},
  volume={12},
  number={5\&6},
  pages={432--441},
  year={2012},
  publisher={Rinton Press}
}

@article{givens1984class,
  title={A class of Wasserstein metrics for probability distributions.},
  author={Givens, Clark R and Shortt, Rae Michael},
  journal={Michigan Mathematical Journal},
  volume={31},
  number={2},
  pages={231--240},
  year={1984},
  publisher={University of Michigan, Department of Mathematics}
}

\vfill

\end{document}